\documentclass[12pt,letterpaper]{JHEP3}

\usepackage{amscd,amsmath,amssymb,amsfonts,xspace,mathrsfs,amsthm}
\usepackage{color}
\let\normalcolor\relax
\usepackage{bbold}
\usepackage{latexsym}
\usepackage{graphicx}
\usepackage{dsfont}
\usepackage{color}
\usepackage{longtable}

 \newtheorem{theorem}{Theorem}
  \newtheorem{proposition}{Proposition}
  \newtheorem{lemma}{Lemma}
  \newtheorem{conj}{Conjecture}

\numberwithin{equation}{section}

\def\lfig#1#2#3#4#5{
\begin{figure}[t]
 \centerline{\includegraphics[width=#3]{#2}}
 \vspace{#5}
  \caption{#1 \label{#4}}
 \end{figure}
}

\def\det{\,{\rm det}\, }

\def\sign{{\rm sgn}}

\def\Ch{{\rm Ch}}

\def\Sym{\,{\rm Sym}\, }

\def\Im{\,{\rm Im}\,}
\def\Re{\,{\rm Re}\,}

\def\({\left(}
\def\){\right)}
\def\[{\left[}
\def\]{\right]}
\def\<{\left\langle}
\def\>{\right\rangle}
\def\hf{{1\over 2}}
\def\haf{\textstyle{1\over 2}}

\newcommand{\de}{\mathrm{d}}

\newcommand{\I}{\mathrm{i}}

\newcommand{\cD}{\mathcal{D}}

\newcommand{\p}{\partial}

\newcommand{\half}{\frac{1}{2}}

\newcommand{\cF}{\mathcal{F}}
\newcommand{\cV}{\mathcal{V}}

\newcommand{\cG}{\mathcal{G}}
\newcommand{\cK}{\mathcal{K}}
\newcommand{\cM}{\mathcal{M}}

\newcommand{\cN}{\mathcal{N}}
\newcommand{\cE}{\mathcal{E}}
\newcommand{\cX}{\mathcal{X}}

\newcommand{\cP}{\mathcal{P}}

\newcommand{\cT}{\mathcal{T}}
\newcommand{\cJ}{\mathcal{J}}

\newcommand{\ub}{\bar{u}}

\DeclareSymbolFont{AMSa}{U}{msa}{m}{n}
\DeclareSymbolFont{AMSb}{U}{msb}{m}{n}
\DeclareMathSymbol{\fieldR}{\mathalpha}{AMSb}{"52}

\newcommand{\kahler}{{K\"ahler}\xspace}

\newcommand{\cI}{\mathcal{I}}

\newcommand{\cA}{\mathcal{A}}

\newcommand{\pa}{\partial}
\newcommand{\nn}{\nonumber}

\newcommand{\eps}{\epsilon}

\newcommand{\IT}{\mathds{T}}
\newcommand{\IR}{\mathds{R}}
\newcommand{\IC}{\mathds{C}}
\newcommand{\IZ}{\mathds{Z}}
\newcommand{\IQ}{\mathds{Q}}

\newcommand{\IP}{\mathds{P}}

\newcommand{\sgn}{\mbox{sgn}}

\newcommand{\tzeta}{\tilde\zeta}

\newcommand{\tc}{\tilde c}

\newcommand{\txi}{\tilde\xi}
\newcommand{\tPhi}{\tilde\Phi}

\newcommand{\CP}{\IC P^1}

\def\bea{\begin{eqnarray}}
\def\eea{\end{eqnarray}}
\def\be{\begin{equation}}
\def\ee{\end{equation}}
\def\ba{\begin{align}}
\def\ea{\end{align}}
\def\bse{\begin{subequations}}
\def\ese{\end{subequations}}

\def\ba{\bar a}

\def\bZ{\bar Z}

\def\bF{\bar F}

\def\hK{\hat K}

\def\hD{\hat D}
\def\hT{\hat T}

\def\Hcl{H^{\rm cl}}

\newcommand{\cB}{\mathcal{B}}

\def\cij#1{c}
\def\ci#1{c}

\def\XXint#1#2#3{{\setbox0=\hbox{$#1{#2#3}{\int}$}
\vcenter{\hbox{$#2#3$}}\kern-.5\wd0}}

\def\Lv#1{L(#1)}
\def\Rv#1{R(#1)}

\def\Fwl{F^{(0)}}
\def\Fref{F^{\rm (ref)}}

\def\Fcl{F^{\rm cl}}

\newcommand{\expe}[1]{{\bf e}\!\left( #1\right)}

\def\gamD#1{\tilde\gamma}

\def\CY{\mathfrak{Y}}
\def\CYm{{\widehat{\mathfrak{Y}}}}

\DeclareMathOperator{\Erf}{Erf}
\DeclareMathOperator{\Erfc}{Erfc}

\def\cl0{\tilde c_0}

\newcommand{\bfu}{{\boldsymbol u}}
\newcommand{\bfb}{{\boldsymbol b}}
\newcommand{\bfc}{{\boldsymbol c}}
\newcommand{\bfv}{{\boldsymbol v}}

\newcommand{\bfk}{{\boldsymbol k}}

\newcommand{\bfp}{{\boldsymbol p}}
\newcommand{\bfq}{{\boldsymbol q}}

\newcommand{\bft}{{\boldsymbol t}}

\newcommand{\bfx}{{\boldsymbol x}}
\newcommand{\bfy}{{\boldsymbol y}}

\newcommand{\bfmu}{{\boldsymbol \mu}}

\newcommand{\gammap}{\gamma}
\newcommand{\gammam}{-\gamma}

\def\trans{g}

\def\Erfc{\text{Erfc}}

\def\bOm{\bar\Omega}

\def\wh{\mathfrak{h}}
\def\bwh{\bar\wh}

\def\cXt{\cX^{(\theta)}}

\def\cXsf{\cX^{\rm sf}}
\def\cXcl{\cX^{\rm cl}}
\def\Hcl{H^{\rm cl}}

\def\tcV{\tilde\cV}

\def\hcT{\hat\cT}

\def\gref{g^{\rm (ref)}}
\def\gs{g^{\star}}
\def\gf{g^{(0)}}
\def\cEf{\cE^{(0)}}
\def\cEs{\cE^{\star}}
\def\whh{\widehat h}

\def\whg{\widehat g}

\def\whPhi{\widehat\Phi}

\def\tPhi{\widetilde\Phi}
\def\tPhif{\widetilde\Phi^{(0)}}
\def\intPhi{\Phi^{\scriptscriptstyle\,\int}}
\def\gPhi{\Phi^{\,g}}
\def\whgPhi{\Phi^{\,\whg}}
\def\gfPhi{\Phi^{\,\gf}}
\def\cEPhi{\Phi^{\,\cE}}

\def\tcEPhi{\widetilde\Phi^{\,\cE}}

\def\cEp{\cE^{(+)}}
\def\whgp{\whg^{(+)}}

\def\Delat#1{\Delta^{#1}}

\def\OmMSW{\Omega^{\rm MSW}}
\def\bOmMSW{{\bar \Omega}^{\rm MSW}}

\def\Zv{\mathscr{Z}}

\def\Zv{\mathscr{Z}}

\def\Lat{\mathbf{\Lambda}}

\def\ptt{\mathfrak{p}}

\def\ver{\mathfrak{v}}

\newcommand{\gtr}{g_{{\rm tr},n}}
\def\gtri#1{g_{{\rm tr},#1}}
\newcommand{\Ftr}[1]{F_{{\rm tr},#1}}

\def\cs{S}
\def\cl{c^{(\ell)}}
\def\vz{\mathbb{z}}
\def\vu{\mathbb{u}}
\def\vv{\mathbb{v}}
\def\tbfv{\tilde\bfv}
\def\tbfu{\tilde\bfu}

\def\cVt{\cV^{\rm tot}}

\def\gama{\check\gamma}

\title{Black holes and higher depth mock modular forms}

\preprint{L2C:18-112\\ arXiv:1808.08479v3}

\author{Sergei Alexandrov$^{1,2}$ and Boris Pioline$^{3}$
\\
$^1$ {\it
Laboratoire Charles Coulomb (L2C), Universit\'e de Montpellier,
CNRS, F-34095, Montpellier, France}\\
$^2$ {\it
Department of High Energy and Elementary Particle Physics,
Saint Petersburg State University,
7/9 Universitetskaya nab., St. Petersburg 199034, Russia}

$^3$ {\it Laboratoire de Physique Th\'eorique et Hautes
Energies (LPTHE), UMR 7589 CNRS-Sorbonne Universit\'e,
Campus Pierre et Marie Curie,
4 place Jussieu, F-75005 Paris, France} \\

\vspace*{2mm} {\tt e-mail:
\email{sergey.alexandrov@umontpellier.fr},
\email{pioline@lpthe.jussieu.fr}
}

\vspace*{-3mm}

}

\abstract{By enforcing invariance under S-duality in type IIB string theory compactified on a Calabi-Yau threefold,
we derive modular properties of the generating function
of BPS degeneracies of D4-D2-D0 black holes in type IIA string theory compactified on the same space.
Mathematically, these BPS degeneracies are the generalized
Donaldson-Thomas invariants counting coherent sheaves with support on a divisor $\cD$,
at the large volume attractor point.
For $\cD$ irreducible, this function is closely related to the elliptic genus of the superconformal field theory
obtained by wrapping M5-brane on $\cD$
and is therefore known to be modular.
Instead, when $\cD$ is the sum of $n$ irreducible divisors $\cD_i$,
we show that the generating function acquires a modular anomaly.
We characterize this anomaly for arbitrary $n$ by providing an explicit expression
for a non-holomorphic modular completion in terms of generalized error functions.
As a result, the generating function turns out to be a (mixed) mock modular form of depth $n-1$.
}

\begin{document}

\section{Introduction and summary}
\label{sec-intro}

The degeneracies of BPS black holes in string vacua with extended supersymmetry
possess remarkable modular properties, which have been instrumental in recent
progress on explaining the statistical origin of the Bekenstein-Hawking
entropy in \cite{Strominger:1996sh} and many subsequent works.
Namely,  the indices $\Omega(\gamma)$ counting --- with sign --- microstates
of BPS black holes with electromagnetic charge $\gamma$ may often be collected
into a suitable generating function which exhibits modular invariance,
providing powerful constraints on its Fourier
coefficients and enabling direct access to their asymptotic growth.
When the black holes can be realized as black strings wrapped on a circle,
a natural candidate for such a generating function is the elliptic genus of
the superconformal field theory supported by the black string, which is modular invariant by
construction \cite{Maldacena:1997de,Gaiotto:2006wm,deBoer:2006vg}. Equivalently,
one may consider the partition function of the effective three-dimensional
gravity living on the near-horizon geometry of the black string \cite{Dijkgraaf:2000fq}.

In most cases however, the BPS indices depend not only on the charge
$\gamma$ but also on the moduli $z^a$ at spatial infinity, due to the wall-crossing phenomenon:
some of the BPS bound states with total charge $\gamma$ only exist in a certain chamber
in moduli space, and decay as the moduli are varied across `walls of marginal stability'
which delimit this chamber.  At strong coupling where the black hole description is accurate,
this phenomenon has a transparent interpretation in terms of the (dis)appearance of multi-centered black hole configurations,
which can be used to derive a universal wall-crossing formula
\cite{Denef:2007vg,Diaconescu:2007bf,Manschot:2010qz,Andriyash:2010qv}.

In the case of four-dimensional string vacua with $\cN=4$ supersymmetry, where the BPS
index is sensitive only to single-centered 1/4-BPS black holes and to
bound states of two 1/2-BPS black holes, the resulting moduli dependence is
reflected in poles in the generating function, requiring a suitable choice of contour for extracting the
Fourier coefficients in a given chamber \cite{Dijkgraaf:1996it,Cheng:2007ch,Banerjee:2008yu}.
Upon subtracting contributions of two-centered
bound states, the generating function of single-centered indices is no longer modular in
the usual sense but it transforms as a {\it mock} Jacobi form with specific `shadow' --- a property
which is almost as constraining as standard modular invariance  \cite{Dabholkar:2012nd}.

In four-dimensional string vacua with $\cN=2$ supersymmetry, the situation is much more complicated,
firstly due to the fact that the moduli space of scalars receives quantum corrections,
and secondly due to BPS bound states potentially involving  an arbitrary number of constituents,
resulting in an  extremely intricate pattern of walls of marginal stability. Thus, it does not
seem plausible that a single generating function may capture the BPS indices $\Omega(\gamma,z^a)$ ---
which are known in this context as generalized Donaldson-Thomas (DT) invariants ---
in all chambers. Nevertheless, modular invariance is still expected to constrain them.
In particular, D4-D2-D0 black holes in type IIA string theory compactified on a generic compact Calabi-Yau (CY) threefold
$\CY$ can be lifted to an M5-brane wrapped  on a  divisor $\cD\subset \CY$ \cite{Maldacena:1997de}. If the divisor $\cD$
labelled by the D4-brane charge $p^a$ is {\it irreducible},
the indices $\Omega(\gamma,z^a)$
are independent of the moduli of $\CY$, at least in the limit where the volume of $\CY$ is scaled to infinity,
and their generating function  is known to be a holomorphic (vector valued) modular form of weight
$-\hf\, b_2(\CY)-1$ \cite{Gaiotto:2006wm,deBoer:2006vg,Denef:2007vg,Alexandrov:2016tnf}.
This includes the case
of vertical, rank 1 D4-D2-D0 branes in K3-fibered Calabi-Yau threefolds \cite{Gholampour:2013hfa,Diaconescu:2015koa,Bouchard:2016lfg}.
But if the divisor $\cD$
is a sum of $n$ effective divisors $\cD_i$, the indices $\Omega(\gamma,z^a)$ do depend on the
\kahler moduli $z^a$, even in the large volume limit. In general however, the black string SCFT
is supposed to capture the states associated to a single $AdS_3$ throat, while for generic values
of the moduli multiple $AdS_3$ throats can contribute \cite{deBoer:2008fk,Andriyash:2008it}.
It is thus natural to consider the modular properties of the DT invariants
$\Omega(\gamma,z^a)$ at the large volume attractor point\footnote{Here $p^a$ and $q_a$ are D4 and D2-brane charges,
respectively, and the index of $q_a$ is raised with help of the inverse of the metric
$\kappa_{ab}=\kappa_{abc}p^a$ with $\kappa_{abc}$ being the triple intersection numbers on
$H_4(\CY,\IZ)$.}
\be
\label{lvap}
z^a_\infty(\gamma)= \lim_{\lambda\to +\infty}\(-q^a+\I\lambda  p^a\),
\ee
where only a single $AdS_3$ throat is allowed \cite{Manschot:2009ia}. Following
\cite{Alexandrov:2012au} we denote these invariants by $\OmMSW(\gamma)=\Omega(\gamma,z^a_\infty(\gamma))$ and call them
Maldacena-Strominger-Witten (MSW) invariants. As we discuss in Section \ref{sec-DT},
the DT invariants $\Omega(\gamma,z^a)$ can be recovered from the MSW invariants
$\OmMSW(\gamma)$ by using a version of the split attractor flow conjecture  \cite{Denef:2001xn,Denef:2007vg}
developed in \cite{Alexandrov:2018iao}, which we call the flow tree formula.

The case where $\cD$ is the sum of two irreducible  divisors was first
considered in \cite{Manschot:2009ia,Alim:2010cf,Klemm:2012sx},
and studied more recently in \cite{Alexandrov:2016tnf,Alexandrov:2017qhn}.
In that case, the generating function of MSW invariants turns out to be a mock modular
form, with a specific non-holomorphic completion obtained by smoothing out the sign functions entering in
the bound state contributions, recovering the prescription of \cite{Zwegers-thesis,Manschot:2009ia}.
The goal of this paper is to extend this result to the general case where $\cD$ is the sum of
$n$ irreducible divisors $\cD_i$, where $n$ can be arbitrarily large.
In such generic situation, we characterize the modular properties of the generating function $h_{p,\mu}$ of MSW invariants
and find an explicit expression for its non-holomorphic completion $\widehat h_{p,\mu}$
in terms of the generating functions $h_{p_i,\mu_i}$ associated to the $n$ constituents,
multiplied by certain iterated integrals introduced in \cite{Alexandrov:2016enp,Nazaroglu:2016lmr},
which generalize the usual error function appearing when $n=2$.
This result implies that in this case $h_{p,\mu}$ is a (mixed, vector valued) mock modular
form of depth $n-1$, in the sense that the antiholomorphic derivative of its modular completion
is itself a linear combination of modular completions of mock modular forms of lower depth, times
antiholomorphic modular forms
(with the depth 1 case reducing to the standard mock modular forms introduced in
\cite{Zwegers-thesis,MR2605321}, and the depth 0 case to usual weakly holomorphic modular forms;
see \cite[(3.16)]{Gupta:2018krl} for a more precise definition).

In order to establish this result, we follow the same strategy as in our earlier works \cite{Alexandrov:2016tnf,Alexandrov:2017qhn}
and  analyze D3-D1-D(-1) instanton corrections to the metric on
the hypermultiplet moduli space $\cM_H$ in type IIB string theory compactified on $\CY$,
at arbitrary order in the instanton expansion. After reducing on a circle and T-dualizing,
this moduli space is identical to the vector multiplet moduli space in type IIA string theory
compactified on $\CY\times S_1$, where it receives instanton corrections from
D4-D2-D0 black holes winding around the circle. In either case, each instanton contribution
is weighted by the same generalized DT invariant $\Omega(\gamma)$ which
counts the number of BPS black hole microstates with electromagnetic charge $\gamma$.
The modular properties of the generalized DT invariants are fixed by requiring that the quaternion-K\"ahler
(QK) metric on $\cM_H$ admits an isometric action of $SL(2,\IZ)$, which comes from S-duality in type IIB,
or equivalently from large diffeomorphisms of the torus appearing when viewing type IIA/$\CY\times S_1$
as M-theory on $\CY\times T^2$ \cite{Alexandrov:2008gh}. This QK metric is most efficiently
encoded in the complex contact structure  on the associated twistor space,
a $\CP$-bundle over $\cM_H$ \cite{MR1327157,Alexandrov:2008nk}.
The latter is specified by a set of gluing conditions determined by the DT invariants \cite{Alexandrov:2008gh,Alexandrov:2009zh},
which can in turn be expressed in terms of the MSW invariants using the tree flow formula.

An important quantity appearing in this twistorial formulation\footnote{A familiarity with
the twistorial formulation is not required for this work.
Here we use only two equations relevant for the twistorial description of D-instantons:  the integral
equation \eqref{eqTBA} for certain Darboux coordinates, which appears also in the study of four-dimensional $N=2$
gauge theory on a circle \cite{Gaiotto:2008cd}, and the expression for the contact potential \eqref{phiinstmany}
in terms of these Darboux coordinates.
These two equations lead to \eqref{expcX} and \eqref{Phitwo-tcF}, respectively,
which can be taken as the starting point of our analysis.
Note that in the context of gauge theories the contact potential can be interpreted
as a supersymmetric index \cite{Alexandrov:2014wca}.}
is the so called {\it contact potential} $e^\phi$, a real function on $\cM_H$
related to the \kahler potential on the twistor space,
and afforded by the existence of a continuous isometry unbroken by D-instantons.
On the type IIB side, $e^\phi$ can be identified
with the four-dimensional dilaton $1/g_4^2$.  When expressed in terms of the ten-dimensional
axio-dilaton $\tau=c_0+\I/g_s$ (or in terms of the modulus of the 2-torus on the M-theory side), it becomes
a complicated function having a classical contribution, a one-loop correction, and a series of instanton corrections
expressed as contour integrals on the $\CP$ fiber. The importance of the contact potential
stems from the fact that it must be a non-holomorphic modular form of weight $(-\half,-\half)$ in the variable $\tau$
in order for $\cM_H$ to admit an isometric action of $SL(2,\IZ)$~\cite{Alexandrov:2008gh}.
This requirement imposes very non-trivial constraints on the instanton contributions to $e^\phi$,
which can be used to deduce the modular properties of generating functions of DT invariants,
at each order in the instanton expansion.
This strategy was used in \cite{Alexandrov:2016tnf} at two-instanton order
to characterize the modular behavior of the generating function of MSW invariants
in the case of a divisor equal to the sum of $n=2$ irreducible divisors.
In this paper we generalize this
result to all $n$, by analyzing the instanton expansion to all orders.
Below we summarize the main steps in our analysis and our main results.

\begin{enumerate}

\item
First, we show that the contact potential $e^\phi$ in the large volume limit, where the K\"ahler parameters of the CY are sent to infinity,
can be expressed (see \eqref{Phitwo-tcF}) through another (complex valued) function $\cG$ \eqref{defcF2}
on the moduli space, which we call the instanton generating function.
The expansion of $\cG$ in powers of DT invariants
is governed by a sum over unrooted trees decorated by charges $\gamma_i$ (see \eqref{treeF} and \eqref{expl-tcA}).
The modularity of the contact potential requires $\cG$ to transform as a modular form of weight $(-\tfrac32,\tfrac12)$.

\item
After expressing the DT invariants $\Omega(\gamma,z^a)$ through the moduli independent MSW invariants
$\OmMSW(\gamma)$ using the tree flow formula of \cite{Alexandrov:2018iao},
and expanding $\cG$  in powers of $\OmMSW(\gamma)$,
each order in this expansion
can be decomposed into a sum of products of certain indefinite theta series
and of holomorphic generating functions $h_{p_i,\mu_i}$ of the invariants
$\OmMSW(\gamma)$ (see \eqref{treeFh-fl}),
similarly to the usual decomposition of standard Jacobi forms.
Thus, the modular properties of the indefinite theta series $\vartheta_{\bfp,\bfmu}\bigl(\Phi^{{\rm tot}}_{n},n-2\bigr)$
are tied with  the modular properties of the generating functions $h_{p_i,\mu_i}$,
in order for $\cG$ to be modular.

\item
In order for an indefinite theta series $\vartheta_{\bfp,\bfmu}\bigl(\Phi,\lambda)$ to be modular,
its kernel $\Phi$ must satisfy a certain differential equation \eqref{Vigdif},
which we call Vign\'eras' equation \cite{Vigneras:1977}.
By construction, the kernels $\Phi^{\rm tot}_n$ appearing in our problem are given by
iterated contour integrals along the $\CP$ fiber of the twistor space,
multiplied by so-called `tree indices' $\gtr$ coming from the expression of $\Omega(\gamma,z^a)$ in terms of $\OmMSW(\gamma)$.
We evaluate the twistorial integrals in terms of the generalized error functions introduced in
\cite{Alexandrov:2016enp,Nazaroglu:2016lmr}, and show that the resulting kernels satisfy Vign\'eras' equation away
from certain loci where they have discontinuities.
Furthermore, we prove that the discontinuities corresponding to walls of marginal
stability cancel between the integrals and the tree indices.
But there are additional discontinuities coming from certain moduli independent contributions to the tree index.
They spoil Vign\'eras' equation at the multi-instanton level so that the theta series
$\vartheta_{\bfp,\bfmu}\bigl(\Phi^{{\rm tot}}_{n},n-2\bigr)$  are {\it not} modular.
In turn, this implies that the holomorphic generating functions $h_{p_i,\mu_i}$ are not modular either.

\item
However, we show that one can complete $h_{p,\mu}$ into a non-holomorphic modular
form $\whh_{p,\mu}(\tau)$, by adding to it  a series of corrections proportional
to products of $h_{p_i,\mu_i}$ with the same total D4-brane charge $p=\sum_i p_i$ (see \eqref{exp-whh}).
The non-holomorphic functions $R_n$ entering the completion
are determined by the condition that the expansion of $\cG$ rewritten in terms of
$\whh_{p,\mu}$ gives rise to a non-anomalous, modular theta series.
Equivalently, one can work with functions $\whg_n$ appearing in the expansion \eqref{multihd-full}
of the holomorphic generating function of DT invariants $h^{\rm DT}_{p,q}$ in powers of
the non-holomorphic functions~$\whh_{p_i,\mu_i}$.

\item
Imposing the conditions for modularity, we find that $\whg_n$ can be represented in an iterative form \eqref{iterDn},
or more explicitly as a sum \eqref{soliterg}
over rooted trees with valency $\geq 3$ at each vertex  (known as Schr\"oder trees),
where $\gf_n$ are certain locally polynomial functions defined in \eqref{defDf-gen},
while $\cE_n$ are smooth solutions of Vign\'eras' equation constructed in terms
of generalized error functions.
The functions $R_n$ are similarly given by a sum over Schr\"oder trees \eqref{solRn} in terms of
exponentially decreasing and non-decreasing parts of $\cE_n$.
These equations represent the main technical result of this work.

\item
For $n\le 5$ our general formulas can be drastically simplified. In particular,
the main building blocks, the functions $\gf_n$ and $\cE_n$, can be written as a sum over a suitable subset of flow trees,
as in \eqref{gfn-upto5}. In addition, we show that $\gf_n$ has a natural extension including the refinement
parameter conjugate to the spin $J_3$ of the black hole.

\end{enumerate}

\lfig{Various types of trees arising in this work.
$\IT_{n,m}$ denotes the set of unrooted trees with $n$ vertices and $m$ marks distributed between the vertices.
$\IT_n\!\equiv \IT_{n,0}$ comprises the trees without marks.
$\IT_n^{\rm r}$ is the set of rooted trees with $n$ vertices.
$\IT_n^{\rm af}$ is the set of attractor flow trees with $n$ leaves (of which we only draw the different topologies).
$\IT_n^{\rm S}$ denotes the set of Schr\"oder trees with $n$ leaves.
In addition, an important r\^ole is played by the sets  $\IT_{n,m}^\ell$ and $\IT_n^\ell=\IT_{n,0}^\ell$
of unrooted, labelled, marked trees which are obtained from $\IT_{n,m}$ and $\IT_{n}$ by assigning different labels and markings to the vertices
in such way that a vertex with $m_\ver$ marks is decorated by $2m_\ver+1$ labels.}
{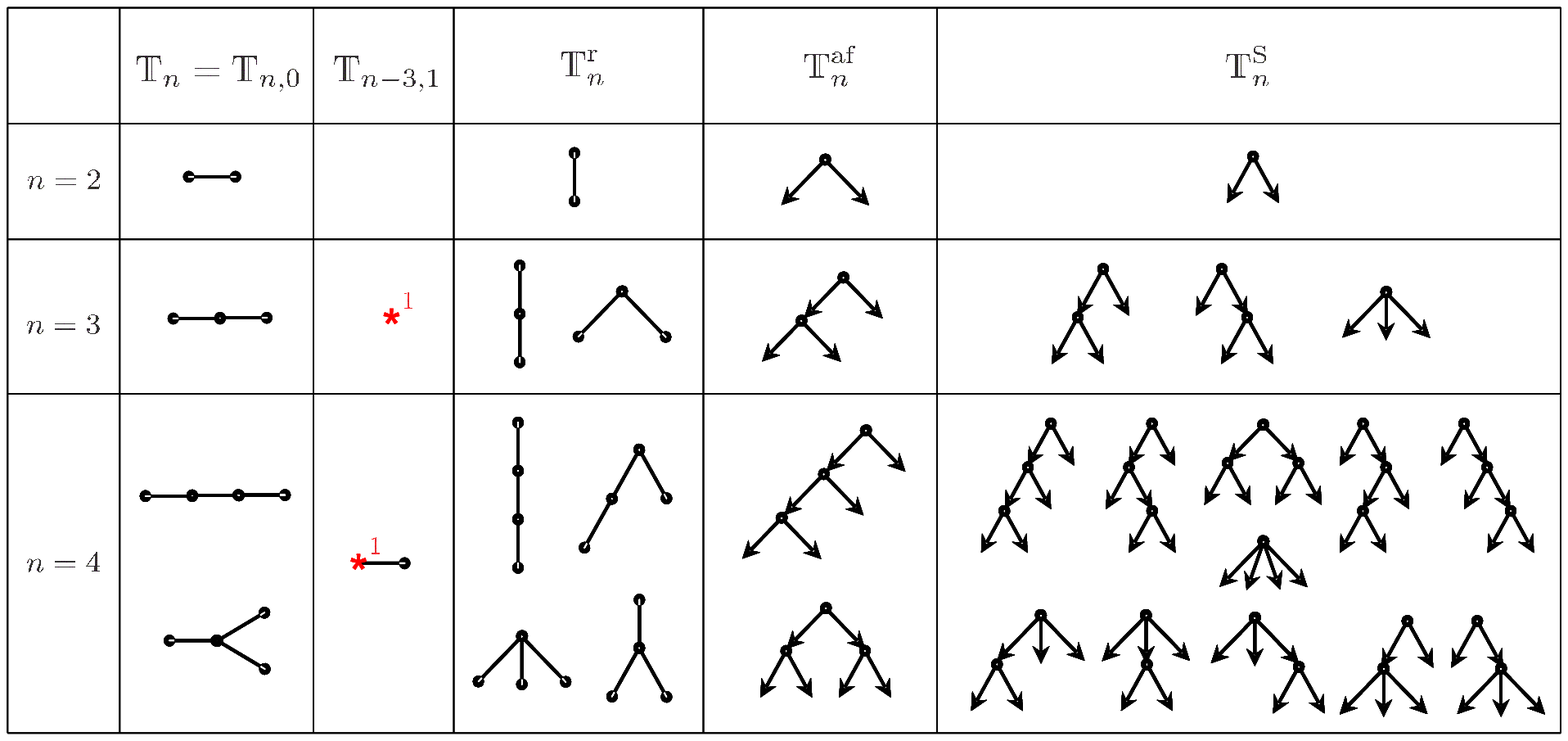}{16cm}{fig-table}{-1.5cm}

At the end of this lengthy analysis, we thus find a modular completion of
the generating functions $h_{p,\mu}$ of MSW invariants for an arbitrary divisor,
i.e. decomposable into a sum of {\it any} number of irreducible divisors.
The result is expressed through sums of products of generalized error functions labelled by trees of various types.
For the reader's convenience, in Fig. \ref{fig-table} we display
the various trees which appear in our construction, up to $n=4$.
Since generalized error functions are known to be related to iterated Eichler integrals \cite{Alexandrov:2016tnf,2017arXiv170406891B},
which occur in the non-holomorphic completion of mock modular forms, we loosely refer to the
generating functions of MSW invariants as higher depth mock modular forms, although we do not
spell out the precise meaning of this notion.

An unexpected byproduct of our analysis is an interesting combinatorial identity, relating rooted trees to the binomial coefficients,
which plays a r\^ole in our derivation of the modular completion.
Since we are not aware of such an identity in the mathematical literature,\footnote{We were informed
by Karen Yeats that the special case $m=n-1$ appears in \cite[(124)]{kreimer1999chen}.
The denominator in \eqref{combident} is sometimes known as the tree factorial $T'!$, see \eqref{defcT}.}
we state it here as a theorem whose proof can be found in appendix \ref{ap-theorem}.
\begin{theorem}\label{theorem}
Let $V_T$ be the set of vertices of a rooted ordered tree $T$
and $n_v(T)$ is the number of descendants of vertex $v$ inside $T$ plus 1
(alternatively, the number of vertices of the subtree in $T$ with the root
being the vertex $v$ and the leaves being the leaves of $T$). Then for a rooted tree $T$ with $n$ vertices and $m<n$ one has
\be
\sum_{T'\subset T}\prod_{v\in V_{T'}}\frac{n_v(T)}{n_v(T')}=\frac{n!}{m!(n-m)!}\, ,
\label{combident}
\ee
where the sum goes over all subtrees with $m$ vertices, having the same root as $T$ (see Fig. \ref{fig-subtrees}).
\end{theorem}

The organization of the paper is as follows.
In section \ref{sec-DT} we review known results about DT invariants, their expression in terms of MSW invariants,
and specialize them to the case of D4-D2-D0 black holes in type IIA string theory on a Calabi-Yau threefold.
In section \ref{sec-twist} we present the twistorial description of the D-instanton corrected hypermultiplet
moduli space in the dual type IIB string theory, evaluate the contact potential in the large volume approximation
by expressing it through a function $\cG$, and obtain the instanton expansion for this function via unrooted labelled trees.
In section \ref{sec-modul} we obtain a theta series decomposition for each order of the expansion of $\cG$
in MSW invariants and analyze the modular anomaly of the resulting theta series, implying a modular anomaly for the generating functions $h_{p,\mu}$.
In section \ref{sec-compl} we construct the non-holomorphic completion $\whh_{p,\mu}$, for which the anomaly is cancelled,
and determine its explicit form. Section \ref{sec-disc} is devoted to discussion of the obtained results.
Finally, several appendices contain details of our calculations and proofs of various propositions.
In addition, in appendix \ref{ap-explicit} we present explicit results up to order $n=4$,
and in appendix \ref{sec_index}, as an aid to the reader, we provide an index of notations.

\lfig{An example illustrating the statement of the Theorem for a tree with $n=7$ vertices and subtrees with $m=3$ vertices.
The subtrees are distinguished by red color.
Near each vertex of the subtrees we indicated the pair of numbers $(n_v(T),n_v(T'))$.}
{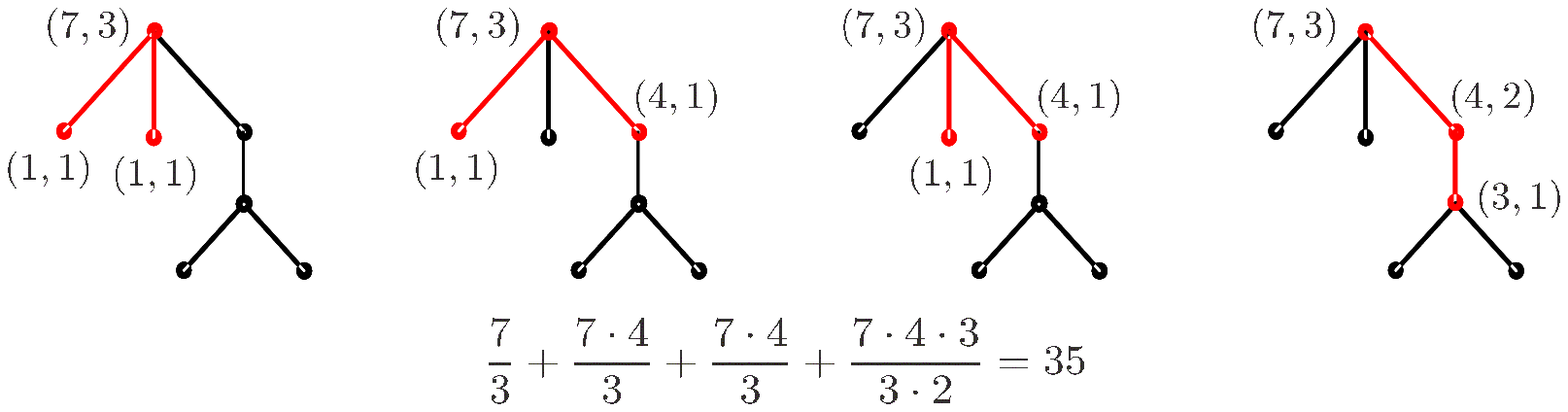}{16.5cm}{fig-subtrees}{-1.9cm}

\section{BPS indices and wall-crossing}
\label{sec-DT}

In this section, we review some aspects of BPS indices in theories with $\cN=2$ supersymmetry,
including the tree flow formula relating the moduli-dependent index $\Omega(\gamma,z^a)$
to the attractor index $\Omega_\star(\gamma)$. We then apply this formalism in the context of Calabi-Yau
string vacua, and express the generalized DT invariants $\Omega(\gamma,z^a)$ in terms of their
counterparts evaluated at the large volume attractor point \eqref{lvap}, known as MSW
invariants.

\subsection{Wall-crossing and attractor flows}
\label{subsec-aftrees}

The BPS index $\Omega(\gamma,z^a)$ counts (with sign) micro-states of BPS black holes
with total electro-magnetic charge $\gamma=(p^\Lambda,q_\Lambda)$, for a given value
$z^a$ of the moduli at spatial infinity. While $\Omega(\gamma,z^a)$ is a locally constant
function over the moduli space, it can jump across real codimension one
loci where certain  bound states, represented by multi-centered black hole solutions of $N=2$ supergravity, become unstable.
The positions of these loci,  known as {\it walls of marginal stability},
are determined by the central charge $Z_\gamma(z^a)$, a complex-valued linear function of $\gamma$
whose modulus gives the mass of a BPS state of charge $\gamma$,
while the phase determines the supersymmetry subalgebra preserved by the state.
Since a bound state can only decay when its mass becomes equal to the sum of masses of its constituents,
it is apparent that the walls correspond to hypersurfaces where the phases of two central charges,
say $Z_{\gamma_L}(z^a)$ and $Z_{\gamma_R}(z^a)$,  become aligned. The bound states which may become
unstable are then those whose constituents have charges in the positive cone spanned by
$\gamma_L$ and $\gamma_R$. We shall assume that the charges $\gamma_L,\gamma_R$
have non-zero Dirac-Schwinger-Zwanziger (DSZ) pairing
$\langle\gamma_L,\gamma_R\rangle\ne 0$, since otherwise marginal bound states may form,
whose stability is hard to control.

The general relation between the values of $\Omega(\gamma,z^a)$ on the two sides of a wall has been found
in the mathematics literature by Kontsevich--Soibelman \cite{ks} and Joyce--Song \cite{Joyce:2008pc,Joyce:2009xv},
and justified physically in a series of works \cite{Denef:2007vg,Diaconescu:2007bf,Manschot:2010qz,Andriyash:2010qv}.
However, in this work we require a somewhat different result: an expression of $\Omega(\gamma,z^a)$ in terms of moduli-independent indices.
One such representation, known as the Coulomb branch formula, was developed in a series of papers
\cite{Manschot:2011xc,Manschot:2012rx,Manschot:2013sya}
(see \cite{Manschot:2014fua} for a review) where the moduli-independent index is the so-called `single-centered invariant'
counting single-centered, spherically symmetric  BPS black holes.
Unfortunately, this representation (and its inverse) is quite involved, as it requires
disentangling genuine single-centered solutions from so-called scaling solutions,
i.e. multi-centered solutions with $n\geq 3$ constituents which can become arbitrarily close to each other
\cite{Denef:2007vg,Bena:2012hf}.

A simpler alternative is to consider the {\it attractor index}, i.e. the value of the BPS index in the attractor chamber
$\Omega_*(\gamma)\equiv \Omega(\gamma,z^a_\star(\gamma))$, where $z^a_\star(\gamma)$ is fixed
in terms of the charge $\gamma$ via the attractor mechanism \cite{Ferrara:1995ih}
(recall that for a spherically symmetric BPS black hole with charge $\gamma$,
the scalars in the vector multiplets
have fixed value $z^a_\star(\gamma)$ at the horizon independently of their value $z^a$ at spatial infinity.).
By definition, the attractor indices are of course moduli independent.
The problem of expressing $\Omega(\gamma,z^a)$ in terms of attractor indices was addressed
recently in \cite{Alexandrov:2018iao}, extending earlier work in \cite{Manschot:2009ia,Manschot:2010xp}.
Relying on the split attractor flow conjecture \cite{Denef:2001xn,Denef:2007vg}, it was argued that
the rational BPS index
\be
\label{defntilde}
\bar\Omega(\gamma,z^a) = \sum_{d|\gamma}  \frac{1}{d^2}\,
\Omega(\gamma/d,z^a)
\ee
can be expanded in powers of $\bOm_*(\gamma_i)$,
\be
\label{Omsumtree}
\bOm(\gamma,z^a) =
\sum_{\sum_{i=1}^n \gamma_i=\gamma}
\gtr(\{\gamma_i\},z^a)\,
\prod_{i=1}^n \bOm_*(\gamma_i),
\ee
where the sum runs over {\it ordered}\footnote{
In \cite{Alexandrov:2018iao} a similar formula was written as a sum over {\it unordered} decompositions,
weighted by the symmetry factor $1/|{\rm Aut}\{\gamma_i\}|$. Since
$\gtr(\{\gamma_i\},z^a)$ is symmetric under
permutations of $\{\gamma_i\}$, we can sum over all ordered decompositions with unit weight,
at the expense
of inserting the factor $1/n!$ in its definition \eqref{defgtree}.
In the sequel, all similar sums are always assumed to run over ordered decompositions.
}
decompositions of $\gamma$ into sums of vectors $\gamma_i\in\Gamma_+$, with $\Gamma_+$ being the set of all vectors $\gamma$ whose
central charge $Z_\gamma(z^a_\infty)$ lies in a fixed half-space defining the splitting between
BPS particles ($\gamma\in\Gamma_+$) and anti-BPS particles ($\gamma\in - \Gamma_+$).
Such decompositions correspond to contributions of multi-centered black hole solutions with constituents carrying charges $\gamma_i$.

The coefficient $\gtr$, called the {\it tree index}, is defined as
\be
\label{defgtree}
\gtr(\{\gamma_i\},z^a)=
\frac{1}{n!}
\sum_{T\in \IT_n^{\rm af}}\Delta(T)\, \kappa(T) ,
\ee
where the sum goes over the set $\IT_n^{\rm af}$ of {\it attractor flow trees}
with $n$ leaves.
These are unordered
rooted binary trees\footnote{The number of such trees is
$|\IT_n^{\rm af}|=(2n-3)!!=(2n-3)!/[2^{n-2}(n-2)!]=\{1,1,3,15,105,945,...\}$ for $n\geq 1$.}
$T$ with vertices decorated by electromagnetic charges $\gamma_v$,
such that the leaves of the tree carry the constituent charges $\gamma_i$, and
the charges propagate along the tree according to $\gamma_v= \gamma_{\Lv{v}}+\gamma_{\Rv{v}}$
at each vertex, where $\Lv{v}$, $\Rv{v}$ are the two children\footnote{This assignment
requires an ordering of the children at each vertex, which can be chosen arbitrarily for
each tree. With such an ordering, and assuming that all the charges $\gamma_i$ are distinct,
the flow tree can be labelled by a $2$-bracketing of
a permutation of the set $\{1,\dots, n\}$, as shown in Fig. \ref{fig-AFtree}.\label{foobracket}} of
the vertex $v$ (see Fig. \ref{fig-AFtree}). The
charge carried by the root of the tree is then the total charge $\gamma=\sum\gamma_i$.
The idea of the split attractor flow conjecture is that each tree represents
a nested sequence of two-centered bound states describing
a multi-centered solution built out of constituents with charges $\gamma_i$.
With this interpretation, the edges of the graph  represent the evolution of the moduli under attractor flow,
so that one starts from the moduli at spatial infinity $z^a_\infty\equiv z^a$ and assigns to the root $v_0$
the point in the moduli space $z^a_{v_0}$ where the attractor flow with charge $\gamma$
crosses the wall of marginal stability
where $\Im \bigl[Z_{\gamma_{\Lv{v_0}}} \bZ_{\gamma_{\Rv{v_0}}}(z^a_{v_0})\bigr]=0$.
Then one repeats this procedure for every edge, obtaining a set of charges
and moduli $(\gamma_v,z^a_v)$  assigned to each vertex, with the bound state constituents
and their attractor moduli
$(\gamma_i,z^a_{\gamma_i})$ assigned to the leaves.

\lfig{An example of attractor flow tree corresponding to the bracketing $((13)(2(45)))$.}{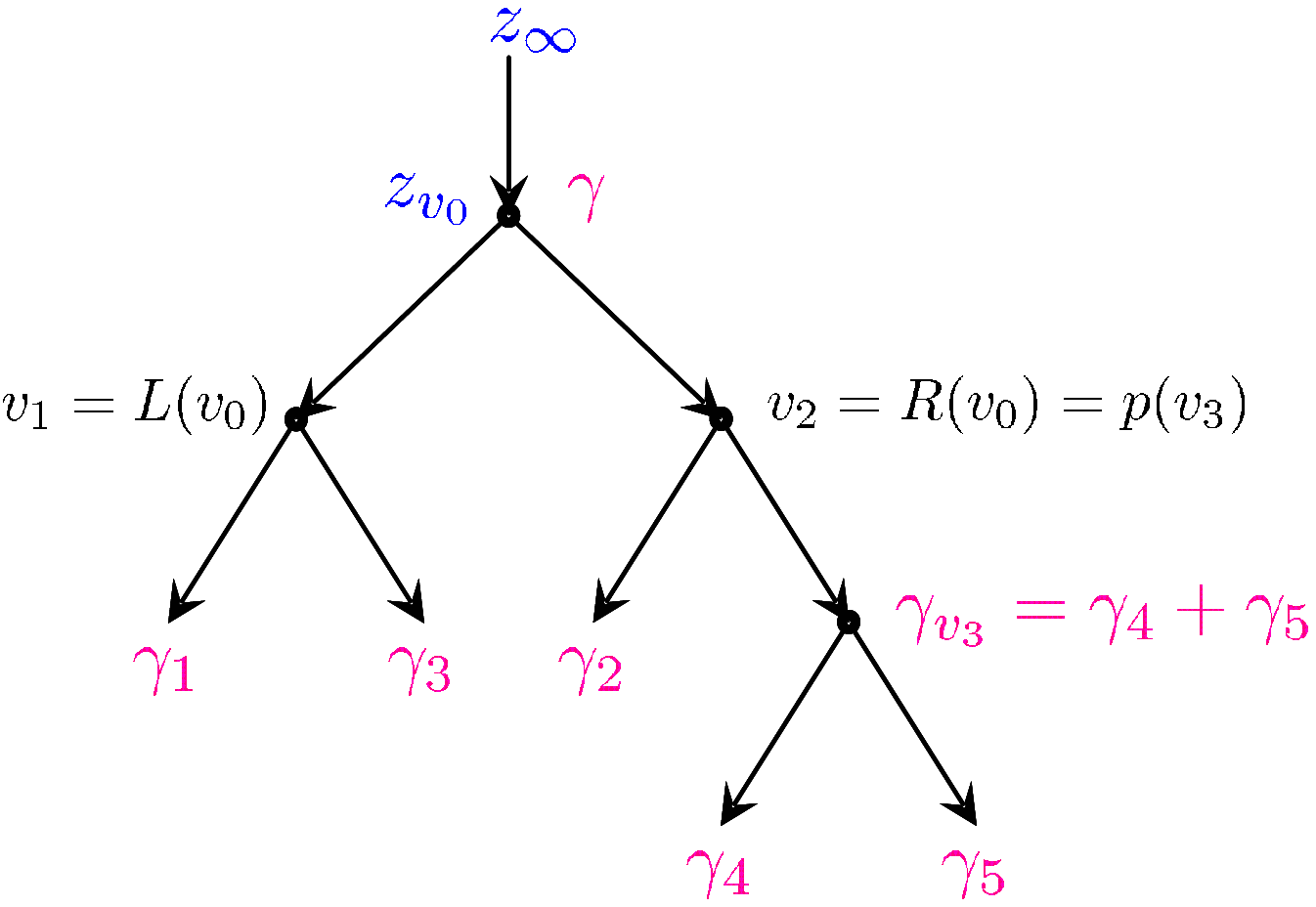}{8.5cm}{fig-AFtree}{-1.1cm}

Given these data, the factor $\Delta(T)$ in \eqref{defgtree} is given by
\be
\label{kappaT}
\Delta(T)=\prod_{v\in V_T}\Delta_{\gamma_{\Lv{v}}\gamma_{\Rv{v}}}^{z_{p(v)}},
\qquad
\Delta_{\gamma_L\gamma_R}^z=
\hf\,\Bigl[\sign \,\Im\bigl[ Z_{\gamma_L}\bZ_{\gamma_R}(z^a)\bigr]+ \sign (\gamma_{LR}) \Bigr] ,
\ee
where $V_T$ denotes the set of vertices of $T$ excluding the leaves, $p(v)$ is the parent of vertex $v$,
and $\gamma_{LR}=\langle \gamma_L,\gamma_R \rangle$.
This factor vanishes unless the stability condition\footnote{In fact, the admissibility also requires
$\Re\bigl[ Z_{\gamma_{\Lv{v}}}\bZ_{\gamma_{\Rv{v}}}(z^a_{v})\bigr]>0$ at each vertex. This condition will hold automatically
for the case of our interest, namely, D4-D2-D0 black holes in the large volume limit, so we do not impose it explicitly.}
\be
\label{condtree}
\gamma_{\Lv{v}\Rv{v}}
\, \Im\bigl[ Z_{\gamma_{\Lv{v}}}\bZ_{\gamma_{\Rv{v}}}(z^a_{p(v)})\bigr]>0
\ee
is satisfied for all $v\in V_T$, which ensures admissibility of the flow tree $T$, i.e. the existence
of the corresponding nested bound state. Importantly, the sign of
$\Im\bigl[ Z_{\gamma_{\Lv{v}}}\bZ_{\gamma_{\Rv{v}}}(z^a_{p(v)})\bigr]$ entering  \eqref{kappaT} can
be computed recursively in terms of asymptotic data, without evaluating the attractor flow along the edges \cite{Alexandrov:2018iao}.
More precisely, the signs depend only on  the stability parameters (also known as Fayet-Iliopoulos parameters)
\be
\label{fiparam}
c_i(z^a) =  \Im\bigl[ Z_{\gamma_i}\bZ_\gamma(z^a)\bigr]\, .
\ee
Note that due to $\gamma=\sum_{i=1}^n \gamma_i$, these parameters satisfy
$\sum_{i=1}^n c_i=0$. Accordingly, we shall often denote $\gtr(\{\gamma_i\},z^a)$
by $\gtr(\{\gamma_i,c_i\})$, and always assume that $\{c_i\}$ are generic real
parameters subject to the condition $\sum_{i=1}^n c_i=0$,
such that no proper subset of them sums up to zero.

The second factor in \eqref{defgtree} is independent of the moduli, and given by
\be
\label{kappaT0}
\kappa(T) \equiv (-1)^{n-1} \prod_{v\in V_T} \kappa( \gamma_{\Lv{v}\Rv{v}}),
\qquad
\kappa(x)=(-1)^x\,  x.
\ee
This is simply the product of the BPS indices of the nested two-centered solutions associated with the tree $T$.
Note that the signs of $\Delta(T)$ and $\kappa(T)$ separately depend on the choice of ordering
$(\gamma_{\Lv{v}},\gamma_{\Rv{v}})$ at each vertex, but  their product
is independent of that choice.

Sometimes, it is useful to consider the refined BPS index $\Omega(\gamma,z^a,y)$ which carries additional
dependence on the fugacity $y$ conjugate to the spin $J_3$. All the above equations remain valid in this case as well
(except for the definition of the rational invariant \eqref{defntilde}, which must be slightly modified,
see e.g. \cite[(1.3)]{Manschot:2010qz}), but now the function $\kappa(x)$
appearing in \eqref{kappaT0} becomes a symmetric Laurent polynomial in $y$
\be
\kappa(x)=(-1)^x \, \frac{ y^x-y^{-x}}{y-y^{-1}}\, ,
\label{kappadef}
\ee
reducing to $(-1)^x x$ in the unrefined  limit $y\to 1$.

It is easy to see that the `flow tree formula' \eqref{Omsumtree} is consistent with the primitive wall-crossing formula
\cite{Denef:2007vg,Diaconescu:2007bf}
\be
\label{primwc}
\Delta\bOm(\gamma_L+\gamma_R) =  -\sign(\gamma_{LR})\, \kappa(\gamma_{LR})\, \bOm(\gamma_L,z^a)\, \bOm(\gamma_R,z^a) ,
\ee
which gives the jump of the BPS index due to the decay of bound states after crossing the wall defined by
a pair of primitive\footnote{Here, by primitive we mean that all charges with non-zero index in the two-dimensional lattice
spanned by $\gamma_L$ and $\gamma_R$ are linear combinations $N_L \gamma_L + N_R\gamma_R$ with coefficients $N_L, N_R$ of the same sign.}
charges $\gamma_L$ and $\gamma_R$. To this end, it suffices to consider all flow trees which
start with the splitting $\gamma\to \gamma_L+\gamma_R$ at the root of the tree.
It is also consistent with the general wall-crossing formula of \cite{ks}, provided the tree index is computed
for a small generic perturbation of the DSZ matrix $\gamma_{ij}$ \cite{Alexandrov:2018iao}. Finally,
it is useful to note that, assuming that all charges $\gamma_i$ are distinct, the sum over splittings and flow trees in \eqref{Omsumtree}
can be generated by iterating the quadratic equation \cite{Alexandrov:2018iao}
\bea
\label{itereq}
\bOm(\gamma,z^a) &= &\bOm_*(\gamma)
-\hf\sum_{\substack{\gamma=\gamma_L+\gamma_R
\\
\langle \gamma_L,\gamma_R \rangle \neq 0}}
\Delta_{\gamma_L\gamma_R}^z \, \kappa(\gamma_{LR})\,
\bOm(\gamma_L,z^a_{LR})\, \bOm(\gamma_R,z^a_{LR}),
\eea
where $z^a_{LR}$ is the point where the attractor flow of charge $\gamma$  crosses the wall of marginal stability
$\Im\bigl[Z_{\gamma_L}\bar Z_{\gamma_R}(z^a_{LR})\bigr]=0$.

\subsection{Partial tree index}
\label{subsec-partialindex}

While the representation of the BPS index based on attractor flows
is useful for many purposes,  it produces a sum of products of  sign functions
depending on non-linear combinations of
DSZ products $\gamma_{ij}$, which are very difficult to work with.
A solution to overcome this problem was found  in \cite{Alexandrov:2018iao}.
The key idea is to introduce a refined index
with a fugacity $y$ conjugate to angular momentum, and represent it as
\be
\gtr(\{\gamma_i,c_i\},y)
= \frac{(-1)^{n-1+\sum_{i<j} \gamma_{ij} }}{(y-y^{-1})^{n-1}}
\,\Sym\Bigl\{
\Ftr{n}(\{\gamma_i,c_i\})\,y^{\sum_{i<j} \gamma_{ij}}\Bigr\},
\label{gtF}
\ee
where $\Sym$ denotes symmetrization (with weight $1/n!$) with respect to the charges $\gamma_i$,
and $\Ftr{n}$ is the `partial tree index' defined by\footnote{Unlike the tree index $\gtr$,
the partial tree index $\Ftr{n}$ is {\it not} a symmetric function of charges $\gamma_i$ and stability parameters $c_i$,
however we abuse notation and still denote it by $\Ftr{n} (\{\gamma_{i},c_i\})$.}
\be
\label{defFpl}
\Ftr{n} (\{\gamma_{i},c_i\}) =
\sum_{T\in \IT_n^{\mbox{\tiny af-pl}}} \Delta(T) .
\ee
Here the sum runs over the set $\IT_n^{\mbox{\scriptsize af-pl}}$ of {\it planar} flow trees with $n$
leaves\footnote{The number of such trees is  the $n-1$-th Catalan number
$|\IT_n^{\mbox{\scriptsize af-pl}}|=\frac{(2n-2)!}
{n[(n-1)!]^2}=\{1,1,2,5,14,42,132,\dots\}$ for $n\geq 1$.}
carrying ordered charges $\gamma_1,\dots, \gamma_n$.  Although this is not manifest,
the refined tree index \eqref{gtF} is regular at $y=1$, and its value
(computed e.g. using l'H\^opital rule) reduces to the tree index \eqref{defgtree}.
The advantage of the representation \eqref{gtF} is that  the partial tree index $\Ftr{n}$
does not involve the $\kappa$-factors \eqref{kappaT0} and is independent of the refinement parameter.

The partial index $\Ftr{n}$ satisfies two important recursive relations.
To formulate them, let us introduce some convenient notations:
\be
\cs_k=\sum_{i=1}^k c_i,
\quad\
\beta_{k\ell}=\sum_{i=1}^k \gamma_{i\ell},
\quad\
\Gamma_{k\ell}=\sum_{i=1}^k\sum_{j=1}^\ell \gamma_{ij}.
\label{notaion-c}
\ee
In terms of these notations, the partial index satisfies the iterative equation \cite[(2.59)]{Alexandrov:2018iao},
\be
\Ftr{n}(\{\gamma_i,c_i\})=\hf\sum_{\ell=1}^{n-1} \bigl( \sign(\cs_\ell)-\sign (\Gamma_{n\ell})\bigr)\,
\Ftr{\ell}(\{\gamma_i,c_i^{(\ell)}\}_{i=1}^\ell)\,
\Ftr{n-\ell}(\{\gamma_i,c_i^{(\ell)}\}_{i=\ell+1}^n),
\label{inductFn}
\ee
where $c_i^{(\ell)}$ is the value of the stability parameters at the point
where the attractor flow crosses
the wall for the decay $\gamma\to(\gamma_1+\cdots +\gamma_\ell,\gamma_{\ell+1}+\cdots +\gamma_n)$, given by
\be
c_i^{(\ell)}=c_i -\frac{\beta_{ni}}{\Gamma_{n\ell}}\, \cs_{\ell}\, .
\label{flowc}
\ee
Importantly, $c_i^{(\ell)}$ satisfies $\sum_{i=1}^\ell c_i^{(\ell)}= \sum_{i=\ell+1}^n c_i^{(\ell)}=0$
so that the two factors on the r.h.s. of \eqref{inductFn} are well-defined.
Note that the iterative equation \eqref{inductFn} is in the spirit of the quadratic equation \eqref{itereq}.

According to \cite[Prop. 2]{Alexandrov:2018iao},
the partial tree index satisfies another recursion
\be
\begin{split}
\Ftr{n}(\{\gamma_i,c_i\})
=&\,\Fwl_n(\{c_i\})- \sum_{n_1+\cdots +n_m= n\atop n_k\ge 1, \ m<n}
\Ftr{m}(\{\gamma'_k,c'_k\})
\prod_{k=1}^m \Fwl_{n_k}(\beta_{n,j_{k-1}+1},\dots,\beta_{nj_{k}}),
\end{split}
\label{F-ansatz}
\ee
where the sum runs over ordered partitions of $n$, $m$ is the number of parts,
and for $k=1,\dots,m$ we defined
\be
\begin{split}
&\qquad \qquad
j_0=0,
\qquad
j_k=n_1+\cdots + n_k,
\\
&\gamma'_k=\gamma_{j_{k-1}+1}+\cdots +\gamma_{j_{k}},
\qquad
c'_k=c_{j_{k-1}+1}+\cdots +c_{j_{k}}.
\label{groupindex}
\end{split}
\ee
The function appearing in \eqref{F-ansatz} is simply a product of signs,
\be
\Fwl_n(\{c_i\})=\frac{1}{2^{n-1}}\prod_{i=1}^{n-1}\sgn(\cs_i) .
\label{Fwlstar-a}
\ee
This new recursive relation allows to express the partial index in a way which does not involve sign functions
depending on non-linear combinations of parameters, in contrast to the previous relation \eqref{inductFn} where
such sign functions arise due to the discrete attractor flow relation \eqref{flowc}.

\subsection{D4-D2-D0 black holes and BPS indices}
\label{subsec-MSW}

The results presented above are applicable in any theory with $\cN=2$ supersymmetry.
Let us now specialize to the BPS black holes obtained as bound states of D4-D2-D0-branes in type IIA string theory compactified on
a CY threefold $\CY$. In this case the moduli $z^a=b^a+\I t^a$ ($a=1,\dots, b_2(\CY)$)
are the complexified K\"ahler moduli with respect to a basis of $H^2(\CY,\IZ)$,
parametrizing the K\"ahler moduli space $\cM_\cK(\CY)$.
The charge vectors $\gamma\in H_{\rm even}(\CY,\IQ)$ have the form
$\gamma=(0,p^a,q_a,q_0)$ where the first entry corresponds to the D6-brane charge,
which is taken to vanish, whereas the other components, corresponding to the D4,
D2 and D0-brane charges, satisfy
the following quantization conditions \cite{Alexandrov:2010ca}:
\be
\label{fractionalshiftsD5}
p^a\in\IZ ,
\qquad
q_a \in \IZ  + \frac12 \,(p^2)_a ,
\qquad
q_0\in \IZ-\frac{1}{24}\, p^a c_{2,a},
\ee
where $c_{2,a}$ are components of the second Chern class of $\CY$.
In the second relation we used the notations $(kp)_a=\kappa_{abc}k^b p^c$
and $(lkp)=\kappa_{abc}l^a k^b p^c$ (recall that $\kappa_{abc}$ are the intersection numbers
on $H_4(\CY,\IZ)$)
which will be extensively used below.
The lattice of charges $\gamma$ satisfying \eqref{fractionalshiftsD5} will be denoted by $\Gamma$.
The cone $\Gamma_+\subset \Gamma$ is obtained by imposing the further restriction
that the D4-brane charge $p^a$ corresponds to an effective divisor in $\CY$ and
belongs to the K\"ahler cone, i.e.
\be
\label{khcone}
p^3> 0,
\qquad
(r p^2)> 0,
\qquad
k_a p^a > 0,
\ee
for all effective divisors $r^a \gamma_a \in H_4^+(\CY,\IZ)$ and
effective curves $k_a \gamma^a \in H_2^+(\CY,\IZ)$, where $\gamma_a$ denotes irreducible divisors giving
an integer basis of $ \Lambda=H_4(\CY,\IZ)$, dual to the basis $\gamma^a$ of $\Lambda^*=H_2(\CY,\IZ)$.
The charge $p^a$ induces a quadratic form $\kappa_{ab}=\kappa_{abc} p^c$ on $\Lambda\otimes \IR\simeq \IR^{b_2}$
of signature $(1,b_2-1)$. This quadratic form allows to embed $\Lambda$ into $\Lambda^*$, but the map $\epsilon^a \mapsto \kappa_{ab} \epsilon^b$
is in general not surjective, the quotient $\Lambda^*/\Lambda$ being a finite group of order $|\det\kappa_{ab}|$.

The holomorphic central charge, governing the mass of BPS states, is given by
\be
Z_\gamma(z^a)=q_\Lambda X^\Lambda(z^a) -p^\Lambda F_\Lambda(z^a),
\label{defZg}
\ee
where $X^\Lambda(z^a)=(1,z^a)$ are the special coordinates
and $F_\Lambda=\p_{X^\Lambda} F(X)$ is the derivative
of the holomorphic prepotential $F(X)$ on $\cM_\cK$.
In the large volume limit $t^a\to\infty$, the prepotential reduces to the classical cubic contribution
\be
F(X)\approx \Fcl(X)=-\kappa_{abc}\, \frac{ X^a X^b X^c}{6X^0}\, ,
\label{Fcl}
\ee
and the central charge can be approximated as
\be
Z_\gamma\approx -\hf\, (pt^2)+\I\( q_a t^a+(pbt)\)+q_0+q_a b^a+\hf\,(pb^2).
\ee
Note, in particular, that it always has a large negative real part.
Another useful observation is that both quantities appearing in the definition of $\Delta_{\gamma_L\gamma_R}^z$ \eqref{kappaT}
are independent of the last component $q_0$ of the charge vector. Indeed,
\be
\begin{split}
\langle\gamma,\gamma'\rangle=&\, q_a p'^a-q'_ap^a,
\\
\Im\bigl[ Z_{\gamma}\bZ_{\gamma'}\bigr]=&\, -\hf\,\Bigl((p' t^2)(q_{a} +(pb)_a)t^a - (p t^2) (q'_{a}+(p'b)_a) t^a\Bigr).
\end{split}
\label{gamZZ}
\ee

The BPS index $\bOm(\gamma,z^a)$ counting D4-D2-D0 black holes is given mathematically
by the generalized Donaldson-Thomas invariant, which counts\footnote{More precisely,
the generalized DT invariant computes the weighted Euler characteristic of the moduli space of
semi-stable coherent sheaves~\cite{Joyce:2008pc}; in this context, the DSZ product $\langle\gamma,\gamma'\rangle$
coincides with the antisymmetrized Euler form.}
semi-stable coherent sheaves supported
on a divisor $\cD$ in the homology class $p^a\gamma_a$, with first and second Chern numbers
determined by $(q_a,q_0)$.
An important property of these invariants is that
they are unchanged under a combined integer
shift of the Kalb-Ramond field, $b^a\mapsto b^a +\epsilon^a$,
and a spectral flow transformation acting on the D2 and D0 charges
\be
\label{flow}
q_a \mapsto q_a - \kappa_{abc}p^b\epsilon^c,
\qquad
q_0 \mapsto q_0 - \epsilon^a q_a + \frac12\, (p\epsilon \epsilon).
\ee
The shift of $b^a$ is important since the DT invariants are only
piecewise constant as functions of the complexified K\"ahler moduli $z^a=b^a+\I t^a$
due to wall-crossing.

In contrast, the MSW invariants $\bOmMSW(\gamma)$, defined as the generalized
DT invariants $\bOm(\gamma,z^a)$ evaluated at their respective large volume attractor point
\eqref{lvap}, are by construction independent of the moduli, and therefore
invariant under the spectral flow \eqref{flow}.
As a result, they only depend on $p^a, \mu_a$ and $\hat q_{0}$, where we traded the electric charges $(q_a,q_0)$
for $(\epsilon^a, \mu_a,\hat q_0)$. The latter comprise the spectral flow parameter $\epsilon^a$, the residue class
$\mu_a\in \Lambda^*/\Lambda$ defined by the decomposition
\be
\label{defmu}
q_a = \mu_a + \frac12\, \kappa_{abc} p^b p^c + \kappa_{abc} p^b \epsilon^c,
\qquad
\eps^a\in\Lambda\, ,
\ee
and the invariant charge ($\kappa^{ab}$ is the inverse of $\kappa_{ab}$)
\be
\label{defqhat}
\hat q_0 \equiv
q_0 -\frac12\, \kappa^{ab} q_a q_b\, ,
\ee
which is invariant under \eqref{flow}.
This allows to write $\bOmMSW(\gamma)=\bOm_{p,\mu}( \hat q_0)$.

An important fact is that the invariant charge $\hat q_0$ is
bounded from above by $\hat q_0^{\rm max}=\tfrac{1}{24}((p^3)+c_{2,a}p^a)$.
This allows to define two generating functions
\bea
\label{defchimu}
h^{\rm DT}_{p,q}(\tau,z^a) &=& \sum_{\hat q_0 \leq \hat q_0^{\rm max}}
\bOm(\gamma,z^a)\,\expe{-\hat q_0 \tau },
\\
h_{p,\mu}(\tau) &=& \sum_{\hat q_0 \leq \hat q_0^{\rm max}}
\bOm_{p,\mu}(\hat q_0)\,\expe{-\hat q_0 \tau },
\label{defhDT}
\eea
where we used notation $\expe{x}=e^{2\pi\I x}$.
Whereas the generating function  of DT invariants $h^{\rm DT}_{p,q}$ depends on the full electric charge $q_a$
and depends  on the moduli $z^a$ in a  piecewise constant fashion,
the generating function of MSW invariants $h_{p,\mu}(\tau)$, due to the spectral flow symmetry,
depends only on the residue class $\mu_a$. This generating function will be the central object of interest in this paper,
and our main goal will be to understand its behavior under modular transformations of $\tau$.

In general, the MSW invariants $\bOmMSW(\gamma)$ are distinct from the attractor
moduli $\bOm_\star(\gamma)$, since the latter coincide with the generalized DT invariants
$\bOm(\gamma,z^a)$
evaluated at the true attractor point $z^a_*(\gamma)$ for the charge $\gamma$,
while the former are the generalized DT invariants evaluated at the large volume attractor
point $z^a_\infty(\gamma)$ defined in \eqref{lvap}. Nevertheless, we claim that in the large volume limit $t^a\to\infty$,
the tree flow formula reviewed in the previous subsections still allows to express $\bOm(\gamma,z^a)$
in terms of the MSW invariants, namely
\be
\label{Omsumtreelv}
\bOm(\gamma,z^a) =
\sum_{\sum_{i=1}^n \gamma_i=\gamma}
\gtr(\{\gamma_i\}, z^a)\,
\prod_{i=1}^n \bOmMSW(\gamma_i) \qquad (\Im z^a\to\infty)
\ee
The point is that the only walls of marginal stability which extend to infinite volume are
those where the constituents carry no D6-brane charge, and that non-trivial bound states
involving constituents with D4-brane charge are ruled out at the large volume attractor point, similarly to the usual
attractor chamber. Since the r.h.s. of  \eqref{Omsumtreelv} is consistent with wall-crossing
in the infinite volume limit and agrees with the left-hand side at $z^a=z^a_\infty(\gamma)$,
it must therefore hold everywhere at large volume. Of course, some of the states contributing to
$\bOmMSW(\gamma_i)$ may have some substructure, e.g. be realized as D6-$\overline{\rm D6}$ bound states,
but this structure cannot be probed in the large volume limit.  Importantly, since the
quantities  \eqref{gamZZ} entering in the definition of the tree index are independent of
the D0-brane charge $q_0$, the flow tree formula \eqref{Omsumtreelv} may be rewritten as a relation between
the generating functions,
\be
\label{multih}
h^{\rm DT}_{p,q}(\tau,z^a) =
\sum_{\sum_{i=1}^n \gama_i=\gama}
\gtr(\{\gama_i\},z^a)\,e^{\pi\I \tau Q_n(\{\gama_i\})}
\prod_{i=1}^n h_{p_i,\mu_i}(\tau),
\ee
where $\gama=(p^a,q_a)$ denotes the projection of the charge vector $\gamma$ on $H_4\oplus H_2$, and the phase proportional to
\be
Q_n(\{\gama_i\})= \kappa^{ab}q_a q_b-\sum_{i=1}^n\kappa_i^{ab}q_{i,a} q_{i,b}
\label{defQlr}
\ee
appears due to the quadratic term in the definition \eqref{defqhat} of  the invariant charge $\hat q_0$.

\section{D3-instantons and contact potential}
\label{sec-twist}

In this section,  we switch to the dual setup\footnote{Our preference for the type IIB set-up
is merely for consistency with our earlier works on hypermultiplet moduli spaces in $d=4$ Calabi-Yau vacua.
The same considerations apply verbatim, with minor changes of wording, to the vector multiplet moduli space in type IIA string theory
compactified on $\CY\times S^1$, which is more directly related to the counting of D4-D2-D0
black holes in four dimensions.}
of type IIB string theory compactified on the same CY manifold $\CY$.
The DT invariants, describing the BPS degeneracies of D4-D2-D0 black holes in type IIA, now appear as coefficients in front of
the D3-D1-D(-1) instanton effects affecting the metric on the hypermultiplet moduli space $\cM_H$.
The main idea of our approach is that these instanton effects are strongly constrained by
demanding that $\cM_H$ admits an isometric action of the type IIB S-duality $SL(2,\IZ)$.
This constraint uniquely fixes the modular behavior of the generating functions introduced in the previous section.
Here we recall the twistorial construction of D-instanton corrections to the hypermultiplet metric,
describe the action of S-duality, and analyze the instanton expansion of a particular function on $\cM_H$ known as contact potential.

\subsection{$\cM_H$ and twistorial description of instantons \label{sec-twistinst}}

The moduli space of four-dimensional $N=2$ supergravity is a direct product
of vector and hypermultiplet moduli spaces, $\cM_V\times \cM_H$.
The former is a (projective) special K\"ahler manifold, whereas the latter is a quaternion-K\"ahler (QK) manifold.
In type IIB  string theory compactified on a CY threefold $\CY$,
$\cM_H$ is a space of real dimension $4b_2(\CY)+4$, which is fibered over the
complexified K\"ahler moduli space $\cM_\cK(\CY)$ of dimension $2b_2(\CY)$.
In addition to the K\"ahler moduli $z^a=b^a+\I t^a$,
it describes the dynamics of the ten-dimensional axio-dilaton $\tau=c^0+\I/g_s$,
the  Ramond-Ramond (RR) scalars $c^a,\tc_a,\tc_0$, corresponding to periods of the RR
2-form, 4-form and 6-form on a basis of $H^{\rm even}(\CY,\IZ)$, and finally,
the NS-axion $\psi$, dual to the Kalb-Ramond two-form $B$ in four dimensions.

At tree-level, the QK metric on $\cM_H$ is obtained from the \kahler moduli space $\cM_{\cK}$ via the
$c$-map construction \cite{Cecotti:1988qn,Ferrara:1989ik} and thus is completely
determined by the holomorphic prepotential $F(X)$.
But this metric receives $g_s$-corrections, both perturbative and non-perturbative.
The latter can be of two types: either from Euclidean D-branes wrapping even dimensional
cycles on $\CY$, or from NS5-branes wrapped around the whole $\CY$.
In this paper we shall be interested only
in the effects of D3-D1-D(-1) instantons, and ignore the effects of NS5 and  D5-instantons,
which are subleading in the large volume limit. Since NS5-instantons only mix with D5-instantons
under S-duality, this truncation does not spoil modular invariance~\cite{Alexandrov:2012au}.

The most concise way to describe the D-instanton corrections is to
consider type IIA string theory compactified on the mirror CY threefold $\CYm$ and use
the twistor formalism for quaternionic geometries \cite{MR1327157,Alexandrov:2008nk}.
In this approach the metric is encoded in the complex contact structure on
the twistor space, a $\CP$-bundle over $\cM_H$. The D-instanton corrected contact structure
has been constructed to all orders in the instanton expansion in \cite{Alexandrov:2008gh,Alexandrov:2009zh},
and an explicit expression for the metric has been derived recently in \cite{Alexandrov:2014sya,Alexandrov:2017mgi}.
Here we will present only those elements of the construction which are relevant for the subsequent analysis, and refer
to reviews \cite{Alexandrov:2011va,Alexandrov:2013yva} for more details.

The crucial point is that, locally, the contact structure is determined by a set of holomorphic Darboux coordinates
$(\xi^\Lambda,\txi_\Lambda,\alpha)$ on the twistor space,
considered as functions of coordinates on $\cM_H$ and of the stereographic coordinate $t$ on the $\CP$ fiber,
so that the contact one-form takes the canonical form $\de\alpha+\txi_\Lambda\de\xi^\Lambda$.
Although all Darboux coordinates are important for recovering the metric, for the purposes of this paper
the coordinate $\alpha$ is irrelevant. Therefore, we consider only $\xi^\Lambda$ and $\txi_\Lambda$
which can be conveniently packaged into holomorphic Fourier modes
$\cX_\gamma = \expe{p^\Lambda \txi_\Lambda-q_\Lambda \xi^\Lambda}$
labelled by a charge vector $\gamma=(p^\Lambda,q_\Lambda)$.

At tree level, the Darboux coordinates (multiplied by $t$) are known to be simple quadratic polynomials in $t$
so that $\cX_\gamma$ take the form\footnote{The superscript `sf' stands for `semi-flat',
which refers to the flatness of the classical geometry in the directions along the torus fibers parametrized by
$\zeta^\Lambda,\tzeta_\Lambda$.}
\be
\cXsf_\gamma(t)=\expe{\frac{\tau_2}{2}\(\bZ_\gamma(\ub^a)\,t-\frac{Z_\gamma(u^a)}{t}\)
+p^\Lambda \tzeta_\Lambda- q_\Lambda \zeta^\Lambda},
\label{defXsf}
\ee
where $Z_\gamma(u^a)$ is the central charge \eqref{defZg}, now expressed in terms of the complex structure moduli $u^a$
of the CY threefold $\CYm$ mirror to $\CY$,
$\zeta^\Lambda$ and $\tzeta_\Lambda$ are periods of the RR 3-form in the type IIA formulation, and $\tau_2=g_s^{-1}$ is the inverse ten-dimensional
string coupling.
At the non-perturbative level, this expression gets modified and the Darboux coordinates are determined by the integral equation
\be
\cX_\gamma(t) = \cXsf_\gamma(t)\, \expe{
\frac{1}{8\pi^2}
\sum_{\gamma'} \sigma_{\gamma'} \, \bOm(\gamma')\, \langle \gamma ,\gamma'\rangle
\int_{\ell_{\gamma'} }\frac{\de t'}{t'}\, \frac{t+t'}{t-t'}\,
\cX_{\gamma'}(t')},
\label{eqTBA}
\ee
where the sum goes over all charges labelling cycles wrapped by D-branes,
$\bOm(\gamma')=\bOm(\gamma',z^a)$ is the corresponding rational Donaldson-Thomas
invariant,
\be
\ell_\gamma = \{ t\in \CP\ :\ \ Z_\gamma/t\in \I \IR^-\}
\label{defBPSay}
\ee
is the so called BPS ray, a contour on $\CP$ extending from $t=0$ to $t=\infty$
along the direction fixed by the central charge, and $\sigma_\gamma$
is a quadratic refinement of the DSZ product on the charge lattice $\Gamma$, i.e. a sign factor satisfying the defining relation
\be
\sigma_{\gamma_1}\sigma_{\gamma_2}=(-1)^{\langle\gamma_1,\gamma_2\rangle}\sigma_{\gamma_1+\gamma_2},
\qquad
\forall \gamma_1,\gamma_2\in\Gamma.
\label{defqf}
\ee
The system of integral equations \eqref{eqTBA} can be solved iteratively by first substituting $\cX_{\gamma'}(t')$
on the r.h.s.  with its zero-th order value $\cXsf_{\gamma'}(t')$ in the weak coupling limit
$\tau_2\to\infty$,
computing the leading correction from the integral and iterating this process.
This produces an asymptotic series at weak coupling, in powers of
the DT invariants $\bOm(\gamma)$. Using the saddle point method, it is easy to check that
the coefficient of each monomial $\prod_i \bOm(\gamma_i)$ is suppressed by a factor
$e^{-\pi\tau_2 \sum_i |Z_{\gamma_i}|}$, corresponding to an $n$-instanton effect \cite{Gaiotto:2008cd,stoppa2014}.
Note that multi-instanton effects become of the same order as one-instanton
effects on walls of marginal stability where the phases of $Z_{\gamma_i}$ become aligned,
and that the wall-crossing formula ensures that the QK metric on $\cM_H$
is smooth across the walls.
 \cite{Gaiotto:2008cd,Alexandrov:2009zh}.

\subsection{D3-instantons in the large volume limit}

The above construction of D-instantons is adapted to the type IIA formulation because the equation
\eqref{eqTBA} defines the Darboux coordinates in terms of the type IIA fields appearing explicitly in the tree level expression \eqref{defXsf}.
To pass to the mirror dual type IIB formulation, one should apply the mirror map, a coordinate transformation from
the type IIA to the type IIB physical fields. This transformation was determined in the classical
limit in  \cite{Bohm:1999uk}, but it also receives instanton corrections. In order to fix the
form of these corrections, we require that the metric on $\cM_H$ carries an isometric action of S-duality group $SL(2,\IZ)$ of type IIB
string theory, which acts on the type IIB fields by an element $\trans={\scriptsize \begin{pmatrix} a & b \\ c & d \end{pmatrix}}$
in the following way
\be\label{SL2Z}
\begin{split}
&\tau \mapsto \frac{a \tau +b}{c \tau + d} \, ,
\qquad
t^a \mapsto |c\tau+d| \,t^a,
\qquad
\begin{pmatrix} c^a \\ b^a \end{pmatrix} \mapsto
\begin{pmatrix} a & b \\ c & d  \end{pmatrix}
\begin{pmatrix} c^a \\ b^a \end{pmatrix} ,
\\
&\qquad
\tc_a\mapsto \tc_a - c_{2,a} \varepsilon(g) ,
\qquad
\begin{pmatrix} \tc_0 \\ \psi \end{pmatrix} \mapsto
\begin{pmatrix} d & -c \\ -b & a  \end{pmatrix}
\begin{pmatrix} \tc_0 \\ \psi \end{pmatrix},
\end{split}
\ee
where $\varepsilon(\trans)$ is the logarithm of the multiplier system of the Dedekind eta function \cite{Alexandrov:2010ca}.

For this purpose, one uses the fact that any isometric action on a quaternion-K\"ahler manifold
(preserving the quaternionic structure) can be lifted to a holomorphic contact transformation on twistor space.
In the present case, $SL(2,\IZ)$ acts on the fiber coordinate $t$ by a fractional-linear
transformation with $\tau$-dependent coefficients.
This transformation takes a much simpler form when formulated in terms of another coordinate $z$ on $\CP$
(not to be confused with the \kahler moduli $z^a$), which is related to $t$
by a Cayley transformation,
\be
\label{Cayley}
z =\frac{t+\I}{t-\I}\, .
\ee
Then the action of $SL(2,\IZ)$ on the fiber is given by a simple phase rotation\footnote{Actually, this is true only
when five-brane instanton corrections are ignored. Otherwise, the lift
also gets a non-trivial deformation \cite{Alexandrov:2013mha}.}
\be
\label{ztrans}
z\mapsto \frac{c\bar\tau+d}{|c\tau+d|}\, z\,  .
\ee

Using the holomorphy constraint for the $SL(2,\IZ)$ action on the twistor space,
quantum corrections to the classical mirror map were computed in
\cite{Alexandrov:2009qq,Alexandrov:2012bu,Alexandrov:2012au,Alexandrov:2017qhn},
in the large volume limit where the K\"ahler moduli
are taken to be large, $t^a\to \infty$. In this limit, one finds
\be
\begin{split}
u^a=&\,b^a+\I t^a-\frac{\I}{2\tau_2}{\sum_{\gamma\in\Gamma_+}}p^a\[
\int_{\ell_{\gammap}}\de z\,(1-z)\,H_{\gammap}
+\int_{\ell_{\gammam}}\frac{\de z}{z^3} (1-z)\,H_{\gammam}\]
\\
\zeta^0=&\,\tau_1,
\qquad
\zeta^a=-(c^a-\tau_1 b^a) -3\sum_{\gamma\in\Gamma_+} p^a\Re\int_{\ell_{\gammap}} \de z\, z\,H_{\gammap},
\\
\tzeta_a =&\,\tc_a+ \frac{1}{2}\, \kappa_{abc} \,b^b (c^c - \tau_1 b^c)+
\kappa_{abc}t^b \sum_{\gamma\in\Gamma_+} p^c\Im\int_{\ell_{\gammap}} \de z\,H_{\gammap},
\end{split}
\label{inst-mp}
\ee
where we introduced the convenient notation\footnote{The functions $H_\gamma$
have a simple geometric meaning \cite{Alexandrov:2008gh,Alexandrov:2014mfa}: they
generate contact transformations (i.e. preserving the contact structure) relating the Darboux coordinates
living on patches separated by BPS rays. In fact, these functions together with the contours $\ell_\gamma$ are
the fundamental data fixing the contact structure on the twistor space.}
\be
H_\gamma(t)= \frac{\bOm(\gamma)}{(2\pi)^2}\,\sigma_\gamma\cX_\gamma(t)\, .
\label{prepHnew}
\ee
Similar results are known for $\tzeta_0$ and the NS-axion dual to the $B$-field, but will not be needed in this paper.

Note that the integral contributions to the mirror map are written in terms of the coordinate $z$ \eqref{Cayley}.
The reason for using this variable is that,
in the large volume limit, the integrals along BPS rays $\ell_\gamma$ in \eqref{eqTBA}
are dominated by the saddle point \cite{Alexandrov:2012au}
\be
z'_\gamma\approx
-\I\,\frac{(q_a+(pb)_a)\,t^a}{(pt^2)}\, ,
\label{saddle}
\ee
for $(pt^2)>0$, and $z'_{-\gamma}=1/z'_\gamma$ in the opposite case. This shows that all integrands can be expanded
in Fourier series either around $z=0$ or $z=\infty$, keeping constant $t^a z$ or $t^a/z$, respectively.
This allows to extract the leading order in the large volume limit in a simple way.

Let us therefore evaluate the combined limit $t^a\to\infty$, $z\to 0$ of the
system of integral equations \eqref{eqTBA}, assuming that only D3-D1-D(-1) instantons contribute.
As a first step, we rewrite the tree level expression \eqref{defXsf} in terms of the type IIB fields.
To this end, we restrict the charge $\gamma$ to lie in the cone $\Gamma_+$,
take the central charge as in \eqref{defZg} with the cubic\footnote{The other contributions to the prepotential, representing
perturbative  $\alpha'$-corrections and worldsheet instantons,
combine with D(-1) and D1-instantons, but are irrelevant for our discussion of D3-instantons
in the large volume limit.}
prepotential \eqref{Fcl}, and substitute the mirror map \eqref{inst-mp}.
Furthermore, we change the coordinate $t$ to $z$ and take the combined limit.
In this way one finds
\be
\cXsf_\gamma(z)=\cXcl_\gamma(z)
\exp\[2\pi\sum_{\gamma'\in\Gamma_+} (tpp')\int_{\ell_{\gamma'}} \de z'\,H_{\gamma'}\],
\label{defXsfIIB}
\ee
where
\be
\cXcl_\gamma(z) =\expe{ - \hat q_0\tau}\cXt_{p,q}(z)
\label{clactinst}
\ee
is the classical part of the Darboux coordinates which we represented as a product of two factors:
exponential of the invariant charge \eqref{defqhat} and the remaining $q_0$-independent exponential
\be
\cXt_{p,q}(z) = e^{-S^{\rm cl}_p}\,
\expe{- \frac{\tau}{2}\,(q+b)^2+c^a (q_a +\haf (pb)_a)+\I \tau_2 (pt^2)(z^2-2zz_\gamma)},
\label{Xtheta}
\ee
with $S^{\rm cl}_p$ being the leading part of the Euclidean D3-brane action in the large volume limit
given by $S^{\rm cl}_p= \pi\tau_2(pt^2) - 2\pi \I  p^a \tc_a$.
Next, we can approximate
\be
\frac{\de t'}{t'}\, \frac{t+t'}{t-t'}=\frac{2\de z'}{1-z'^2}\, \frac{1-zz'}{z-z'}\approx
\left\{\begin{array}{cc}
\frac{2\de z'}{z-z'}\, , \qquad & \gamma'\in \Gamma_+,
\\
\frac{2(1-zz')\de z'}{z'^3}\, , \qquad & \gamma'\in -\Gamma_+.
\end{array} \right.
\ee
This shows that the contribution of $\gamma'\in -\Gamma_+$ is suppressed comparing to $\gamma'\in \Gamma_+$ and therefore can be neglected.
As a result, the system of integral equations \eqref{eqTBA} in the large volume limit
where only D3-D1-D(-1) instantons contribute reduces to the following system
of integral equations for $H_\gamma$,
\be
H_\gamma(z)=\Hcl_\gamma(z) \, \exp\[\sum_{\gamma'\in\Gamma_+}\int_{\ell_{\gamma'}}\de z'\, K_{\gamma\gamma'}(z,z')\,H_{\gamma'}(z')\].
\label{expcX}
\ee
Here $\Hcl_\gamma$ is the classical limit of $H_\gamma$, i.e. the function \eqref{prepHnew} with $\cX_\gamma$ replaced by $\cXcl_\gamma$,
the integration kernel is now
\be
K_{\gamma_1\gamma_2}(z_1,z_2)
=2\pi\((tp_1p_2)+\frac{\I\langle\gamma_1,\gamma_2\rangle}{z_1-z_2}\),
\label{defkerK}
\ee
and the BPS ray $\ell_{\gamma}$ effectively extends from $z'=-\infty$ to $z'=+\infty$, going through
the saddle point \eqref{saddle} \cite{Alexandrov:2012au}.

Below we shall need a perturbative solution of the integral equation \eqref{expcX}.
Applying the iterative procedure outlined below \eqref{defqf}, or equivalently using the Lagrange inversion theorem,
such solution can be written as a sum over
{\it rooted trees}  \cite[\S C]{Gaiotto:2008cd},
\be
\label{Hexpand}
H_{\gamma_1}(z_1) = \Hcl_{\gamma_1}(z_1) \, \sum_{n=1}^{\infty} \left(
\prod_{i=2}^n  \sum_{\gamma_i\in\Gamma_+}\,
\int_{\ell_{\gamma_i}} \de z_i\,  \Hcl_{\gamma_i}(z_i) \right) \sum_{\cT\in \IT_n^{\rm r}}\,
\frac{\cA(\cT)}{|{\rm Aut}(\cT)|}\, ,
\ee
where $\IT_n^{\rm r}$ is the set of rooted trees with $n$ vertices and
\be
\cA(\cT) = \prod_{e\in E_{\cT}}
 K_{\gamma_{s(e)},\gamma_{t(e)}}(z_{s(e)},z_{t(e)}).
 \label{amplit}
\ee
A rooted tree\footnote{We will use calligraphic letters $\cT$
for trees where charges $\gamma_i$ are assigned to vertices to distinguish them from
rooted trees $T$ where the charges are assigned to leaves (hence $T$ has always more than $n$ vertices).
Similarly, we will use notations $\ver$ and $v$ for vertices of these two types of trees, respectively.
Note also that whereas $V_\cT$ denotes the set of all vertices, $V_T$ does not includes the leaves.
An example of trees of the latter type are attractor flow trees.\label{foot-trees}}
$\cT$ consists of $n$ vertices joined by directed edges so that the root vertex
has only outgoing edges, whereas all other vertices have one incoming edge and an arbitrary number of outgoing ones.
We label the vertices of $\cT$ by $\ver=1,\dots, n$ in an arbitrary fashion, except
for the root which is labelled by $\ver=1$. The symmetry factor $|{\rm Aut}(\cT)|$ is the order
of the symmetry group which permutes the labels $2,\dots, n$ without changing the topology of the
tree. Each vertex is decorated by a charge vector $\gamma_\ver$ and a complex variable $z_\ver\in\CP$.
We denote  the set of edges by  $E_{\cT}$, the set of vertices by $V_{\cT}$, and
the source and target vertex of an edge $e$ by $s(e)$ and $t(e)$, respectively.
Unpacking these notations, we get, at the few leading orders,
\bea
\label{Hexpandshort}
H_{\gamma_1} &= &  \Hcl_{\gamma_1} + \sum_{\gamma_2} K_{12} \Hcl_{\gamma_1} \,
\Hcl_{\gamma_2}
+  \sum_{\gamma_2,\gamma_3} \left( \tfrac12\, K_{12} K_{13} + K_{12} K_{23} \right)
 \Hcl_{\gamma_1}\,  \Hcl_{\gamma_2}\,  \Hcl_{\gamma_3}
 \\
& + &  \sum_{\gamma_2,\gamma_3,\gamma_4} \!\!\!\left( \tfrac16\, K_{12} K_{13} K_{14}
+ \tfrac12\, K_{12} K_{23} K_{24} + K_{12} K_{13} K_{24} + K_{12} K_{23} K_{34} \right)
 \Hcl_{\gamma_1} \, \Hcl_{\gamma_2}\,  \Hcl_{\gamma_3}\,  \Hcl_{\gamma_4}+ \dots
\nn
\eea
where we omitted the integrals and denoted $K_{ij}=K_{\gamma_i,\gamma_j}(z_i,z_j)$.
The expansion \eqref{Hexpand} is effectively a multi-instanton expansion in powers of
the DT invariants $\bOm(\gamma_i)$, which is asymptotic to the exact solution
to \eqref{expcX} in the  weak coupling limit $\tau_2\to\infty$.

\subsection{From the contact potential to the instanton generating function}

Recall that our goal is to derive constraints imposed by S-duality on the DT invariants $\bOm(\gamma)$
appearing as coefficients in the multi-instanton expansion. To achieve this goal,
rather than  studying the full metric on $\cM_H$, it suffices to consider a suitable function
on this moduli space which has a non-trivial dependence on $\bOm(\gamma)$
and specified transformations under S-duality.
There is a natural candidate with the above properties:
the so-called contact potential $e^{\phi}$, a real function which is well-defined on any
quaternion-K\"ahler manifold with a continuous isometry \cite{Alexandrov:2008nk}.
Furthermore, there is a general expression for the contact potential in terms of
Penrose-type integrals on the $\CP$ fiber.
In the present case, the required isometry is the shift of the NS-axion, which survives all
quantum corrections as long as NS5-instantons are switched off. The contact
potential is  then given  by the exact formula \cite{Alexandrov:2008gh}
\be
e^{\phi} = \frac{\I\tau_2^2}{16}\(\ub^\Lambda F_\Lambda- u^\Lambda \bF_\Lambda\)-\frac{\chi_{\CYm}}{192\pi}
+\frac{\I\tau_2}{16}\,\sum_\gamma
\int_{\ell_\gamma} \frac{\text{d}t}{t} \( t^{-1} Z_\gamma(u^a)-t\bZ_\gamma(\ub^a)\) H_\gamma,
\label{phiinstmany}
\ee
where $\chi_{\CYm}$ is the Euler characteristic of $\CYm$. This formula indeed captures
contribution from D-instantons due to the last term proportional to $H_\gamma$.

On the other hand, in the classical, large volume limit one finds $e^\phi=\frac{\tau_2^2}{12}(t^3)$, which shows that the contact potential
can be identified with the four-dimensional dilaton and in this approximation behaves as a modular form of weight $(-\tfrac12, -\tfrac12)$
under S-duality transformations \eqref{SL2Z}.
In fact, one can show \cite{Alexandrov:2008gh} that $SL(2,\IZ)$ preserves the contact structure, i.e. it is an isometry of $\cM_H$,
only if the {\it full non-perturbative} contact potential transforms in this way,
\be
\label{SL2phi}
e^\phi \mapsto \frac{e^\phi}{|c\tau+d|}\, .
\ee
Furthermore, since S-duality acts by rescaling the K\"ahler moduli $t^a$ and by a phase rotation of the fiber coordinate $z$ (see \eqref{ztrans}),
it preserves each order in the expansion around the large volume limit.
This implies that the large volume limit of the D3-instanton contribution to $e^\phi$, which we denote by $(e^\phi )_{\rm D3}$,
must itself transform as \eqref{SL2phi}. It is this condition that we shall exploit
to derive modularity  constraints on the DT invariants.

To make this condition more explicit, let us extract the D3-instanton contribution to the function \eqref{phiinstmany}.
The procedure is the same as the one used to get \eqref{expcX}, and we relegate the details of the calculation to appendix \ref{ap-contact}.
The result can be written in a concise way using the complex function defined by
\bea
\cG&=& \sum_{\gamma\in\Gamma_+}\int_{\ell_{\gammap}} \de z\, H_{\gamma}(z)-\hf\,\sum_{\gamma_1,\gamma_2\in\Gamma_+}
\int_{\ell_{\gamma_1}}\de z_1\,\int_{\ell_{\gamma_2}} \de z_2
\, K_{\gamma_1\gamma_2}(z_1,z_2)\,H_{\gammap_1}(z_1)H_{\gammap_2}(z_2)
\label{defcF2}
\eea
and the Maass raising operator
\be
\cD_{\wh} = \frac{1}{2\pi \I}\(\partial_\tau+\frac{\wh}{2\I\tau_2}+ \frac{\I t^a}{4\tau_2}\, \partial_{t^a}\),
\label{modcovD}
\ee
which maps modular functions of weight $(\wh,\bwh)$ to modular functions of weight $(\wh+2,\bwh)$.
Then one has (generalizing \cite[(4.5)]{Alexandrov:2016tnf} to all orders in the instanton expansion)
\be
(e^\phi )_{\rm D3}
=\frac{\tau_2}{2}\Re\left(\cD_{-\frac{3}{2}}\cG \right)
+\frac{1}{32\pi^2}\,\kappa_{abc}t^c\p_{\tc_a}\cG\p_{\tc_b}\overline{\cG}.
\label{Phitwo-tcF}
\ee
It is immediate to see that $(e^\phi )_{\rm D3}$ transforms under S-duality as \eqref{SL2phi} provided
the function $\cG$ transforms as a modular form of weight $(-\tfrac32,\tfrac12)$. In order to derive
the implications of this fact, we need to  express $\cG$ in terms of the generalized DT invariants.

For this purpose, we substitute the multi-instanton expansion \eqref{Hexpand} into \eqref{defcF2}. We claim that the result takes the simple form
\be
\cG=\sum_{n=1}^\infty\[\prod_{i=1}^{n} \sum_{\gamma_i\in \Gamma_+}\int_{\ell_{\gamma_i}}\de z_i\, \Hcl_{\gamma_i}(z_i) \]
\cG_n(\{\gamma_i,z_i\}),
\label{treeF}
\ee
where $\cG_n(\{\gamma_i,z_i\})$ is now a sum over {\it unrooted} trees with $n$ vertices,
\be
\cG_n(\{\gamma_i,z_i\})=\sum_{\cT\in\, \IT_n} \frac{\cA(\cT)}{|{\rm Aut}(\cT)|}=\frac{1}{n!}\sum_{\cT\in\, \IT_n^\ell} \cA(\cT),
\label{expl-tcA}
\ee
and in the second equality we rewrote the result as a sum over {\it unrooted labelled}
trees.\footnote{The number of such trees is $|\IT_n^\ell|=n^{n-2}=\{1,1,3,16,125,1296,\dots\}$ for $n\ge 1$.
Such trees also appear in the Joyce-Song wall-crossing
formula \cite{Joyce:2008pc,Joyce:2009xv} and are conveniently labelled by their Pr\"ufer code.}

To see why this is the case, observe that under the action of the Euler operator $\hD=\Hcl\p_{\Hcl}$
rescaling all functions $\Hcl_\gamma$, the function $\cG$  maps to the first term
in \eqref{defcF2}, which we denote by $\cF$. Namely,
\be
\hD\cdot \cG=\cF,
\qquad
\cF\equiv\sum_{\gamma\in\Gamma_+}\int_{\ell_{\gammap}} \de z\, H_{\gamma},
\label{actS}
\ee
as can be verified with the help of the integral equation \eqref{expcX}. The multi-instanton
expansion of $\cF$ follows immediately from \eqref{Hexpand},
\be
\label{Fexpand}
\cF = \sum_{n=1}^{\infty} \left(
\prod_{i=1}^n  \sum_{\gamma_i\in\Gamma_+}\,
\int_{\ell_{\gamma_i}} \de z_i\,  \Hcl_{\gamma_i}(z_i) \right) \sum_{\cT\in \IT_n^{\rm r}}\,
\frac{\cA(\cT)}{|{\rm Aut}(\cT)|}\, .
\ee
Integrating back the action of the derivative operator $\hD$, we see that the sum over rooted trees in \eqref{Fexpand} turns
into the sum over unrooted trees in \eqref{expl-tcA}.
At the first few  orders we get, using
the same shorthand notation as in \eqref{Hexpandshort},
\be
\begin{split}
\cG= &\sum_{\gamma} \Hcl_\gamma
+ \frac12 \sum_{\gamma_1,\gamma_2}K_{12}\Hcl_{\gamma_1} \,\Hcl_{\gamma_2}
+ \frac12 \sum_{\gamma_1,\gamma_2,\gamma_3}
K_{12}K_{23}\Hcl_{\gamma_1}\, \Hcl_{\gamma_2}\, \Hcl_{\gamma_3}
\\
&+\sum_{\gamma_1,\gamma_2,\gamma_3,\gamma_4}
\(\frac16 \, K_{12} K_{13} K_{14}  + \frac12\,K_{12} K_{23} K_{34} \)
\Hcl_{\gamma_1}\, \Hcl_{\gamma_2}\, \Hcl_{\gamma_3}\, \Hcl_{\gamma_4}
+\dots
\end{split}
\ee

The  simplicity of the expansion  \eqref{expl-tcA}, and the relation \eqref{Phitwo-tcF}
to the contact potential, show that the function $\cG$ is very natural and,
in some sense, more fundamental\footnote{In \cite{Alexandrov:2016tnf}, it was noticed that
the function $\cG$, denoted by $\tilde\cF$ in that reference and computed at second order
in the multi-instanton expansion, could be obtained from the seemingly simpler function $\cF$
by {\it halving} the coefficient of its second order contribution.
Now we see that this {\it ad hoc} prescription is the consequence
of going from rooted to unrooted trees,
as a result of adding the second term in \eqref{defcF2}.}
than the naive instanton sum  $\cF$. We shall henceforth refer to $\cG$ as the
`instanton generating function'.
In the following we shall postulate that $\cG$ transforms as a modular form of weight $(-\frac32,\frac12)$,
and analyze the consequences of this requirement for the DT invariants.

\section{Theta series decomposition and modularity}
\label{sec-modul}

In this section, we use the spectral flow symmetry to decompose the
instanton generating function $\cG$ into a sum of indefinite
theta series multiplied by holomorphic generating functions of MSW invariants. We then
study the modular properties of these  indefinite
theta series, and identify the origin of the modular anomaly.

\subsection{Factorisation}
\label{subsec-naive}

To derive modularity constraints on the DT invariants, we need to perform a theta series decomposition of
the generating function $\cG$ defined in \eqref{defcF2}.
To this end, let us make use of the fact noticed in \eqref{gamZZ} that the DSZ products
$\langle\gamma,\gamma'\rangle$ and hence the kernels \eqref{defkerK} do not depend on the $q_0$ charge.
Choosing the quadratic refinement $\sigma_\gamma$ as in \eqref{qf}, which is also $q_0$-independent,
and using the factorization \eqref{clactinst} of $\cXcl_\gamma$, one can rewrite the expansion \eqref{treeF}
as follows
\be
\cG=\sum_{n=1}^\infty\frac{1}{(2\pi)^{2n}}\[\prod_{i=1}^{n}
\sum_{p_i,q_i}\sigma_{p_i,q_i}h^{\rm DT}_{p_i,q_i}\int_{\ell_{\gamma_i}}\de z_i\, \cXt_{p_i,q_i}(z_i) \]
\cG_n(\{\gamma_i,z_i\}),
\label{treeFh}
\ee
where the sum over the invariant charges $\hat q_{i,0}$ gave rise to the generating functions of DT invariants defined in \eqref{defchimu}.
This is not yet the desired form because these generating functions depend non-trivially on the remaining electric charges $q_{i,a}$.
If it were not for this dependence, the sum over $q_{i,a}$ would produce certain non-Gaussian theta series, and at each order
we would have a product of this theta series and $n$ generating functions. Then the modular properties of the theta series
would dictate the modular properties of the generating function.

Such a theta series decomposition can be achieved by expressing the DT invariants in terms of
the MSW invariants, for which the dependence on electric charges $q_{i,a}$ reduces to the dependence on the residue classes $\mu_{i,a}$
due to the spectral flow symmetry.
Substituting the expansion \eqref{multih} of $h^{\rm DT}_{p,q}$ in terms of $h_{p,\mu}$,
the expansion \eqref{treeFh} of the function $\cG$ can be brought to the following factorized form
\be
\cG=\sum_{n=1}^\infty\frac{2^{-\frac{n}{2}}}{\pi\sqrt{2\tau_2}}\[\prod_{i=1}^{n}
\sum_{p_i,\mu_i}\sigma_{p_i}h_{p_i,\mu_i}\]
e^{-S^{\rm cl}_p}\vartheta_{\bfp,\bfmu}\bigl(\Phi^{{\rm tot}}_{n},n-2\bigr),
\label{treeFh-fl}
\ee
where $\vartheta_{\bfp,\bfmu}$ is a theta series \eqref{Vignerasth} with parameter $\lambda=n-2$,
whose kernel has the following structure
\be
\Phi^{{\rm tot}}_n(\bfx)=\Sym\left\{\sum_{n_1+\cdots n_m=n \atop n_k\ge 1}  \intPhi_m(\bfx')
\prod_{k=1}^m \gPhi_{n_k}(x_{j_{k-1}+1},\dots,x_{j_k})\right\}.
\label{totker}
\ee
Here the sum runs over ordered partitions of $n$, whose number of parts is denoted by $m$,
and we adopted notations from \eqref{groupindex} for indices $j_k$.
The argument $\bfx$ of the kernel encodes the electric components of the charges
(shifted by the B-field and rescaled by $\sqrt{2\tau_2}$) and lives in
a vector space $\(\oplus_{i=1}^n \Lambda_i\) \otimes \IR$ of dimension $d=nb_2$,
given by $b_2$ copies of the lattice $\Lambda$, where the $i$-th copy
$\Lambda_i$ carries the bilinear form $\kappa_{i,ab}=\kappa_{abc}p_i^c$ of signature $(1,b_2-1)$.
Therefore, $\vartheta_{\bfp,\bfmu}$ is an indefinite theta series
associated to the bilinear form given explicitly in \eqref{biform}, which has
signature $(n,n(b_2-1))$.

Finally, the kernel \eqref{totker} is constructed from two other functions.
The first, $\intPhi_n$, is the iterated integral of the  coefficient $\cG_n$ in
the expansion \eqref{treeF},
\be
\intPhi_n(\bfx) =
\(\frac{\sqrt{2\tau_2}}{2\pi}\)^{n-1}
\[\prod_{i=1}^n\int_{\ell_{\gamma_i}}\frac{\de z_i}{2\pi} \, W_{p_i}(x_i,z_i)\]\cG_n(\{\gamma_i,z_i\}),
\label{kerPhi}
\ee
weighted with the Gaussian measure factor
\be
W_{p}(x,z)=e^{-2\pi\tau_2 z^2(pt^2)-2\pi\I\sqrt{2\tau_2}\, z\, (pxt)}
\ee
coming from the $z$-dependent part of \eqref{Xtheta}.
Although this function is written in terms of $\cG_n$ depending on full electromagnetic charge
vectors $\gamma_i$,
it is actually independent of their $q_0$ components.
Indeed, using the result \eqref{expl-tcA}, it can be rewritten as
\be
\intPhi_n(\bfx) = \frac{1}{n!}
\[\prod_{i=1}^n\int_{\ell_{\gamma_i}}\frac{\de z_i}{2\pi} \, W_{p_i}(x_i,z_i)\]\sum_{\cT\in\, \IT_n^\ell}\prod_{e\in E_\cT} \hK_{s(e) t(e)},
\label{kerPhin}
\ee
where we introduced a rescaled version of the kernel \eqref{defkerK}
\be
\label{defkerKr}
\hK_{ij}(z_i,z_j)=\(\(\sqrt{2\tau_2}\,t+\I\,\frac{x_i-x_j}{z_i-z_j}\)p_ip_j\).
\ee
Note that $\tau_2$ and $t^a$ appear only in the modular invariant combination $\sqrt{2\tau_2}\, t^a$.
In \eqref{totker} this function appears with the argument $\bfx'$ and carries a dependence on $\bfp'$ (not indicated explicitly)
which are both $mb_2$-dimensional vectors with components (cf. \eqref{groupindex})
\be
p'^a_k=\sum_{i=j_{k-1}+1}^{j_k}p^a_i,
\qquad
x'^a_k=\kappa'^{ab}_k\sum_{i=j_{k-1}+1}^{j_k} \kappa_{i,bc} x^c_i,
\label{defprimevar}
\ee
where $\kappa'_{k,ab}=\kappa_{abc}p'^c_k$.

The second function, $\gPhi_{n}$, appears due to the expansion of DT invariants in terms
of the MSW invariants
and is given by a suitably rescaled tree index
\be
\gPhi_{n}(\bfx)=\frac{\sigma_{\gamma}(\sqrt{2\tau_2})^{n-1}}{\prod_{i=1}^n\sigma_{\gamma_i}}\,\gtr(\{\gamma_i,c_i\}).
\label{kerPhiint}
\ee
It is also written in terms of functions depending on the full electromagnetic charge vectors $\gamma_i$ (with $\gamma=\gamma_1+\cdots\gamma_n$).
However, using \eqref{defqf}, all quadratic refinements can be expressed through $(-1)^{\langle\gamma_i,\gamma_j\rangle}$
which cancel the corresponding sign factors in the tree index (see \eqref{kappaT0}).
Furthermore, as was noticed in the end of section \ref{subsec-MSW}, the tree index is independent of
the $q_0$ components of the charge vectors. Therefore, it can be written as a function of $p_i^a$, $\mu_{i,a}$
and $x_i^a=\sqrt{2\tau_2}(\kappa_i^{ab} q_{i,b}+b^a)$.
Then, since after cancelling the sign factors, $\gtr$ becomes a homogeneous function of degree $n-1$ in the D2-brane charge $q_{i,a}$,
all factors of $\sqrt{2\tau_2}$ in \eqref{kerPhiint} cancel as well.

\subsection{Modularity and Vign\'eras' equation}

As explained in appendix \ref{subsec-Vign}, the theta series $\vartheta_{\bfp,\bfmu}\bigl(\Phi,\lambda\bigr)$
is a vector-valued modular form of weight $(\lambda+d/2,0)$ provided the kernel
$\Phi$ satisfies Vign\'eras' equation \eqref{Vigdif} --- along with certain growth conditions which
we expect to be automatically satisfied for the kernels of interest in this work.
In our case $\lambda=n-2$ and the dimension of the lattice is $d=nb_2(\CY)$ so that the expected weight
of the theta series is $(2(n-1)+nb_2/2,0)$.
It is consistent with weight $(-\tfrac32,\tfrac12)$ of $\cG$ given in \eqref{treeFh-fl} only if $h_{p,\mu}$ is a vector-valued
holomorphic modular form of weight $(-b_2/2-1,0)$.
However for this to be true, the kernel $\Phi^{{\rm tot}}_n$ ought to satisfy Vign\'eras' equation.
Let us examine whether or not this is the case.

To this end, we first consider the kernel $\intPhi_n$ \eqref{kerPhin}.
In appendix \ref{ap-twistint} we evaluate explicitly the iterated integrals defining this kernel.
To present the final result, let us introduce the following $d$-dimensional vectors
$\bfv_{ij}, \bfu_{ij}$:
\be
\begin{split}
(\bfv_{ij})_k^a=&\, \delta_{ki} p_j^a-\delta_{kj} p_i^a \qquad\qquad\quad\
\mbox{such that} \quad \bfv_{ij}\cdot\bfx=(p_ip_j(x_i-x_j)),
\\
(\bfu_{ij})_k^a=&\, \delta_{ki}(p_jt^2)t^a-\delta_{kj} (p_it^2)t^a \quad
\mbox{such that} \quad
\bfu_{ij}\cdot\bfx=(p_jt^2)(p_ix_it)-(p_it^2)(p_jx_jt),
\end{split}
\label{defvij}
\ee
where  $k$ labels the copy in $\oplus_{k=1}^n \Lambda_k$, $a=1,\dots, b_2$, and the bilinear form is given in \eqref{biform}.
The first scalar product $\bfv_{ij}\cdot\bfx$ corresponds to the DSZ product $\langle\gamma_i,\gamma_j\rangle$,
whereas the second product $\bfu_{ij}\cdot\bfx$ corresponds
 to $-2\Im[Z_{\gamma_i}\bZ_{\gamma_j}]$ \eqref{gamZZ}, both rescaled by $\sqrt{2\tau_2}$
and expressed in terms of $x_i^a$. From these vectors we can construct two sets of vectors
which are assigned to
the edges of an unrooted labelled tree $\cT$, such as the trees appearing in \eqref{expl-tcA} and \eqref{kerPhin}. Namely,
\be
\bfv_e=\sum_{i\in V_{\cT_e^s}}\sum_{j\in V_{\cT_e^t}}\bfv_{ij},
\qquad
\bfu_e=\sum_{i\in V_{\cT_e^s}}\sum_{j\in V_{\cT_e^t}}\bfu_{ij},
\label{defue}
\ee
where $\cT_e^s$, $\cT_e^t$ are the two disconnected trees obtained from
the tree $\cT$ by removing the edge $e$.
Then the kernel $\intPhi_n$ can be expressed as follows\footnote{Both vectors $\bfu_e$ and $\bfv_{s(e) t(e)}$
depend on the choice of orientation of the edge $e$, but this ambiguity is cancelled in the function $\tPhi^M_{n-1}$.}
\be
\intPhi_n(\bfx)=
\frac{\intPhi_1(x)}{2^{n-1} n!}\sum_{\cT\in\, \IT_n^\ell}
\tPhi^M_{n-1}(\{ \bfu_e\}, \{\bfv_{s(e) t(e)}\};\bfx).
\label{Phin-final}
\ee
Here the first factor is  simply a Gaussian
\be
\intPhi_1(x)=\frac{e^{-\frac{\pi(pxt)^2}{(pt^2)}}}{2\pi\sqrt{2\tau_2(pt^2)}}
\label{defPhi1-main}
\ee
which ensures the suppression along the direction of the total charge in the charge lattice.
In the second factor one sums over unrooted labelled trees $\cT$ with $n$ vertices,
with summand  given by a function $\tPhi^M_n$ defined as in \eqref{deftPhigen}, upon
replacing  $\Phi_n^E$ by $\Phi_n^M$ in that expression and setting $m=n$.
Both $\Phi_n^E$ and $\Phi_n^M$ are the so-called generalized (complementary) error functions
introduced in
\cite{Alexandrov:2016enp} and further studied in \cite{Nazaroglu:2016lmr},
whose definitions are recalled in \eqref{generr-M}, \eqref{generr-E} and \eqref{generrPhiME}.
The functions $\tPhi^M_{n-1}$ in \eqref{Phin-final} depend on
two sets of $n-1$ $d$-dimensional vectors:
the vectors in the first set are given by $\bfu_e$ defined above, whereas
the vectors $\bfv_{s(e) t(e)}$ in the second set
coincide with $\bfv_{ij}$ for $i$ and $j$ corresponding to
the source and target vertices of edge $e$ of the labelled tree.

The remarkable property of the generalized error functions $\Phi^M_{n-1}$ and their
uplifted versions $\tPhi^M_{n-1}$ is that,
away from certain loci where these functions are discontinuous,
they satisfy Vign\'eras' equation for $\lambda=0$ and $n-1$, respectively.
Given that $\intPhi_1$ is also a solution for $\lambda=-1$, and the vector $\bft=(t^a,\dots,t^a)$
(such that $\bft\cdot \bfx=(pxt)$) is orthogonal to all vectors $\bfu_e$ and $\bfv_{s(e) t(e)}$,
the kernel \eqref{Phin-final} satisfies this equation for $\lambda=n-2$.

However, as mentioned above, it fails to do so on the loci where it is discontinuous.
These discontinuities arise due to dependence of
the integration contours $\ell_\gamma$ on moduli and electric charges.
Of course, since the integrands are meromorphic functions, the integrals do not depend on deformations of the contours
provided they do not cross the poles. But this is exactly what happens when two BPS rays, say $\ell_\gamma$ and $\ell_{\gamma'}$,
exchange their positions, which in turn takes place when the phases of the corresponding central charges $Z_\gamma$ and $Z_{\gamma'}$ align,
as follows from \eqref{defBPSay}.
The loci where such alignment takes place are nothing else but the walls of marginal stability.
This point will play an important r\^ole in the next subsection since it makes it possible to recombine the discontinuities of
the generalized error functions with discontinuities of the tree indices.

We now turn to the action of Vign\'eras' operator on $\gPhi_{n}$.
To this end, it is convenient to use the representation of the tree index as a sum over attractor flow trees \eqref{defgtree}.
Let us assign a $nb_2$-dimensional vector $\tbfv_{v}$ to each vertex $v$ of a flow tree.
Denoting by $\cI_v$ the set of leaves which are descendants of vertex $v$, we set
\be
\tbfv_v=\sum_{i\in\cI_{\Lv{v}}}\sum_{j\in\cI_{\Rv{v}}}\bfv_{ij}.
\label{deftv}
\ee
With these definitions the kernel \eqref{kerPhiint} can be written as
\bea
\gPhi_n(\bfx) &=&
(-1)^{n-1}\sum_{T\in \IT_n^{\rm af}}\prod_{v\in V_T}
(\tbfv_v,\bfx)\,
\Delta_{\gamma_{\Lv{v}}\gamma_{\Rv{v}}}^{z_{p(v)}}.
\label{ker-fl}
\eea
The factors $\Delta_{\gamma_L\gamma_R}^z$ are locally constant and therefore,
away from the loci where they are discontinuous, the action of Vign\'eras operator reduces to its action on the scalar products $(\tbfv_v,\bfx)$.
For a single such factor one finds
\be
V_\lambda(\tbfv_v,\bfx)=(\tbfv_v,\bfx)V_{\lambda-1}+2\tbfv_v\cdot\p_\bfx.
\ee
The crucial observation is that all vectors $\tbfv_v$ appearing in the product \eqref{ker-fl} for a single tree are mutually orthogonal
$(\tbfv_v,\tbfv_{v'})=0$, which is clear because $\langle\gamma_{\Lv{v}},\gamma_{\Rv{v}}\rangle$ is antisymmetric in charges
$\gamma_{\Lv{v}},\gamma_{\Rv{v}}$, whereas the factors associated with vertices which are not descendants of $v$ either depend
on their sum or do not depend on them at all.
Therefore, one obtains
\be
V_\lambda\gPhi_{n_k}(\bfx)=\gPhi_{n_k}(\bfx)V_{\lambda-n_k+1}+
2(-1)^{n_k-1}
\sum_{T\in \IT_{n_k}^{\rm af}}\Delta(T)\sum_{v\in V_T}\[\prod_{v'\in V_T\setminus \{v\}}(\tbfv_{v'},\bfx)\]\tbfv_v\cdot\p_\bfx.
\label{actVgPhi}
\ee

Let us now evaluate the action of $V_{n-2}$ on the full kernel $\Phi^{{\rm tot}}_n$.
Applying the result \eqref{actVgPhi}, we observe that the second term vanishes
on the other factors in \eqref{totker} due to the same reason that they either do not depend
(in the case of $\gPhi_{n_{k'}}$, $k'\ne k$)
or depend (in the case of $\intPhi_{m}$) only on the sum of charges entering $\gPhi_{n_k}$.
Therefore, one finds that, {\it away from discontinuities of generalized error functions and $\Delta(T)$},
one has $V_{n-2}\cdot \Phi^{{\rm tot}}_n=0$.

\subsection{Discontinuities and the anomaly}

Let us now turn to the discontinuities of $\Phi^{{\rm tot}}_n$ which we ignored so far and
which spoil Vign\'eras' equation and hence modularity of the theta series.
There are three potential sources of such discontinuities:
\begin{enumerate}
\item
walls of marginal stability --- at these loci $\intPhi_{m}$ are discontinuous due to exchange of integration contours
and $\gPhi_{n_k}$ jump due to factors $\Delta_{\gamma_{\Lv{v}}\gamma_{\Rv{v}}}^{z_{p(v)}}$ assigned to
the root vertices of attractor flow trees;

\item
`fake walls' --- these are loci in the moduli space where $\Im\bigl[ Z_{\gamma_{\Lv{v}}}\bZ_{\gamma_{\Rv{v}}}(z^a_{p(v)})\bigr]=0$,
and hence the corresponding $\Delta$-factor jumps, where $v$ is not a root vertex ---
they correspond to walls of marginal stability for the intermediate bound states appearing in the attractor flow;

\item
moduli independent loci where $(\tbfv_v,\bfx)=0$ --- at these loci the factors $\Delta_{\gamma_{\Lv{v}}\gamma_{\Rv{v}}}^{z_{p(v)}}$
and hence $\gPhi_{n_k}$ are discontinuous due to the second term in \eqref{kappaT}.

\end{enumerate}

Remarkably, the two effects due to the non-trivial charge and moduli dependence of the DT invariants and the exchange of contours cancel each other
and the function $\cG$ turns out to be smooth at loci of the first type.
This is expected because the whole construction of D-instantons has been designed to make the resulting metric on the moduli space
smooth across these loci, which required the cancellation of these discontinuities \cite{Gaiotto:2008cd,Alexandrov:2009zh}.
Moreover, in \cite{Alexandrov:2014wca} it was proven that the contact potential is also smooth, which indicates that
the function $\cG$ must be smooth as well.
In appendix \ref{ap-smooth} we present an explicit proof of this fact based on the representation in terms of trees.

Furthermore, in \cite{Alexandrov:2018iao} it was shown that the discontinuities across `fake walls' cancel in the sum over flow trees as well.
In fact, this cancellation is explicit in the representation of the partial tree index given by the recursive formula \eqref{F-ansatz}
where the signs responsible for such `fake discontinuities' do not arise at all.
As a result, it remains to consider only the discontinuities of the third type corresponding to the moduli independent loci.

It is straightforward to check that already for $n=2$ these discontinuities are indeed
present and do spoil modularity of
the theta series. For small $n$ one can explicitly evaluate the anomaly in Vign\'eras' equation.
It is given by a series of terms proportional to $\delta(\tbfv_v,\bfx)$. Note that no derivatives of delta functions or products
of two delta functions arise despite the presence of the second derivative in Vign\'eras' operator.
This is because each $\sign(\tbfv_v,\bfx)$ from $\Delta(T)$ is multiplied by $(\tbfv_v,\bfx)$ from $\kappa(T)$ in \eqref{defgtree}
and one gets a non-vanishing result only if the second order derivative operator acts on both factors.
In particular, this implies that the anomaly is completely characterized by the action of $V_\lambda$ on $\gPhi_{n}$.

\section{Modular completion}
\label{sec-compl}

Since the theta series $\vartheta_{\bfp,\bfmu}(\Phi^{{\rm tot}}_{n})$ are not modular for $n\geq 2$,
the analysis of the previous section implies that the generating function $h_{p,\mu}$ of the MSW invariants
is not modular either whenever the divisor $\cD=p^a\gamma_a$ is the sum of
$n\geq 2$ irreducible divisors.
However, its modular anomaly has a definite structure. In particular, in \cite{Alexandrov:2016tnf}
it was shown that for $n=2$, $h_{p,\mu}$ must be a vector-valued {\it mixed mock} modular form,
i.e. it has a non-holomorphic completion $\whh_{p,\mu}$
constructed in a specific way from a set of holomorphic modular forms and their Eichler  integrals  \cite{MR2605321,Dabholkar:2012nd}.
In this section we generalize this result for arbitrary $n$,
i.e. for any degree of reducibility of the divisor.

\subsection{Completion of the generating function}
\label{subsec-complh}

Let us recall the notations $\gama=(p^a,q_a)$ and $Q_n$ from \eqref{defQlr},
and decompose the electric component $q_a$ using spectral flow as in \eqref{defmu}.
Then we define
\be
\whh_{p,\mu}(\tau)= h_{p,\mu}(\tau)
+\sum_{n=2}^\infty
\sum_{\sum_{i=1}^n \gama_i=\gama}
R_n(\{\gama_i\},\tau_2)
\, e^{\pi\I \tau Q_n(\{\gama_i\})}
\prod_{i=1}^n h_{p_i,\mu_i}(\tau).
\label{exp-whh}
\ee
We are looking for non-holomorphic functions $R_n$, exponentially suppressed for large $\tau_2$, such
that $\whh_{p,\mu}$ transforms as a modular form. The condition for this to be true
can be found along the same lines as before:
one needs to rewrite the expansion of the function $\cG$
as a series in $\whh_{p,\mu}$ and require that at each order the coefficient is given by a
modular covariant theta series. For such a theta series decomposition to be possible however,
it is important that $\whh_{p,\mu}$ be invariant under the spectral flow,
which implies that the functions $R_n$ be {\it independent} of the spectral flow parameter $\eps^a$
in the decomposition \eqref{defmu} of the total charge $\gama$.\footnote{As a result, this parameter can be fixed to zero
so that the sum over the D2-brane charges $q_{i,a}$ is restricted to those which satisfy the constraint
$\sum_{i=1}^n q_{i,a}=\mu_a+\frac12 \kappa_{abc}p^b p^c$. \label{foot-D2specflow}}
This condition will be an important consistency requirement on our construction.

Rather than inverting \eqref{exp-whh} and substituting the result into \eqref{treeFh-fl}, we can consider the generating function
of DT invariants $h^{\rm DT}_{p,\mu}$ and, as a first step, rewrite it as a series in $\whh_{p,\mu}$.
Denoting the coefficient of such expansion by $\whg_n(\{\gama_i\},z^a)$ (with $\whg_1\equiv1$), we get
\be
\begin{split}
h^{\rm DT}_{p,q}(\tau,z^a)=&\,
\sum_{\sum_{i=1}^n \gama_i=\gama}
\whg_n(\{\gama_i\};z^a, \tau_2) \,e^{\pi\I \tau Q_n(\{\gama_i\})}
\prod_{i=1}^n \whh_{p_i,\mu_i}(\tau).
\end{split}
\label{multihd-full}
\ee
Comparing with \eqref{multih}, we see that the coefficients $\whg_n$
are a direct analogue of the tree index $\gtr$.\footnote{In fact, they also depend on $z^a$ only through the
stability parameters \eqref{fiparam}, so we shall often denote them by $\whg_n(\{\gama_i,c_i\};\tau_2)$.}
The expansion of $\cG$ in terms of $\whh_{p,\mu}$ is then simply obtained by replacing $\gtr$ by $\whg_n$ in \eqref{kerPhiint},
which affects the kernel of the corresponding theta series.
Our first problem is to express these coefficients in terms of the functions $R_n$.

The result can be nicely formulated using so-called Schr\"oder trees, which are rooted ordered trees such that
all vertices (except the leaves) have at least 2 children (see Fig. \ref{fig-Wrtree}).
Their vertices are decorated by charges in such way that the leaves have charges $\gamma_i$,
whereas the charges at other vertices are defined by the inductive rule $\gamma_v=\sum_{v'\in \Ch(v)}\gamma_{v'}$
where $\Ch(v)$ is the set of children of vertex $v$.
Note that these trees should not be confused with flow trees, since
they are not restricted to be binary nor do they carry moduli at the vertices.
We denote the set of Schr\"oder trees\footnote{The number of Schr\"oder trees
with $n$ leaves is the $n-1$-th super-Catalan number,
$|\IT_n^{\rm S}|=\{ 1,1, 3,11,45,197,903,\dots \}$ for $n\geq 1$ (sequence A001003 on
the Online Encyclopedia of Integer Sequences).}
with $n$ leaves by $\IT_n^{\rm S}$.

\lfig{An example of Schr\"oder tree contributing to $W_8$. Near each vertex we showed the corresponding factor
using the shorthand notation $\gamma_{i+j}=\gamma_i+\gamma_j$.}
{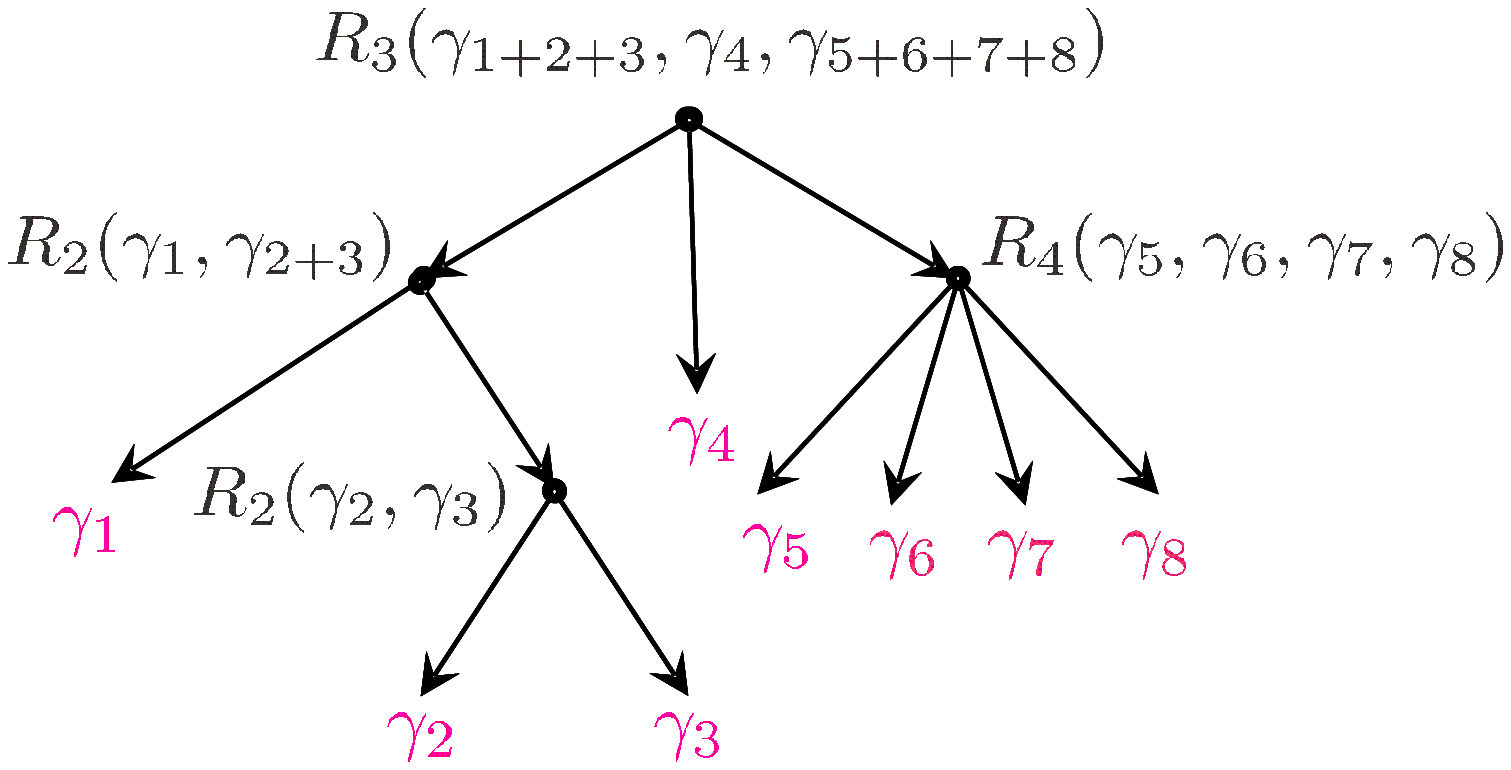}{10.cm}{fig-Wrtree}{-1.2cm}

Let us also introduce a convenient notation: for any set of functions $f_n(\{\gama_i\})$ depending on $n$ charges
and a given Schr\"oder tree $T$,
we set $f_{v}\equiv f_{k_v}(\{\gama_{v'}\})$ where $v'\in \Ch(v)$ and $k_v$ is their number.
Using these notations, the expression of $\whg_n$ in terms of $R_n$ can be encoded into
a recursive equation provided by the following proposition, whose proof we relegate to appendix \ref{ap-proofs}:

\begin{proposition}\label{prop-iterwhg}
The coefficients $\whg_n$ are determined recursively by the following equation
\be
\whg_n(\{\gama_i,c_i\})=-\hf\Sym\left\{\sum_{\ell=1}^{n-1}
\Delta_{\gamma_L^\ell\gamma_R^\ell}^z\,\kappa(\gamma_{LR}^\ell)\,
\whg_\ell(\{\gama_i,c_i^{(\ell)}\}_{i=1}^\ell)\,\whg_{n-\ell}(\{\gama_i,c_i^{(\ell)}\}_{i=\ell+1}^n)\right\}
+W_n(\{\gama_i\}),
\label{recurs-whd-full}
\ee
where $\gamma_L^\ell=\sum_{i=1}^\ell\gamma_i$, $\gamma_R^\ell=\sum_{i=\ell+1}^n\gamma_i$,
$\gamma_{LR}^\ell=\langle\gamma_L^\ell,\gamma_R^\ell\rangle$, $c_i^{(\ell)}$
are the stability parameters at the point where the attractor flow crosses
the wall for the decay $\gamma\to\gamma_L^\ell+\gamma_R^\ell$ (cf. \eqref{inductFn}),
and $W_n$ are functions given by the sum over Schr\"oder trees
\be
W_n(\{\gama_i\})= \Sym\left\{\sum_{T\in\IT_n^{\rm S}}(-1)^{n_T}\prod_{v\in V_T} R_{v}\right\},
\label{defSm}
\ee
with $n_T$ being the number of vertices of the tree $T$ (excluding the leaves).
The same functions $W_n$ also provide the inverse formula to \eqref{exp-whh}, namely
\be
h_{p,\mu}(\tau)= \whh_{p,\mu}(\tau)
+\sum_{n=2}^\infty
\sum_{\sum_{i=1}^n \gama_i=\gama}
W_n(\{\gama_i\})
\, e^{\pi\I \tau Q_n(\{\gama_i\})}
\prod_{i=1}^n \whh_{p_i,\mu_i}(\tau).
\label{exp-hwh}
\ee
\end{proposition}

What are the conditions on $\whg_n$ for the corresponding theta series to be modular?
Let us denote by $\whgPhi_{n}$ the kernel defined by $\whg_n$ analogously to \eqref{kerPhiint}, namely
\be
\whgPhi_{n}(\bfx)=\frac{\sigma_{\gamma}(\sqrt{2\tau_2})^{n-1}}{\prod_{i=1}^n\sigma_{\gamma_i}}\,\whg_n(\{\gamma_i,c_i\}).
\label{kerPhiinth}
\ee
Then the above analysis implies that the modularity requires $\whgPhi_{n}$
to satisfy Vign\'eras equation away from walls of marginal stability, whereas at these walls its discontinuities should follow
the same pattern as before in order to cancel the discontinuities from the contour exchange in $\intPhi_m$.
Thus, the completion should smoothen out the discontinuities from the moduli independent signs $\sign(\tbfv_v,\bfx)$,
but otherwise leave the action of Vign\'eras' operator unchanged.
Formally, this means that $\whgPhi_{n}$ must satisfy the following equation
\be
V_{n-1} \cdot \whgPhi_{n}=\Sym\sum_{\ell=1}^{n-1}
\Bigl[\bfu_\ell^2\,\Delta_{n,\ell}^{\whg} \,\delta'(\bfu_\ell\cdot\bfx)
+ 2\bfu_\ell\cdot\p_\bfx \Delta_{n,\ell}^{\whg} \,\delta(\bfu_\ell\cdot\bfx)\Bigr],
\qquad
\Delta_{n,\ell}^{\whg}=\frac12\,(\bfv_\ell,\bfx)\,\whgPhi_\ell\,\whgPhi_{n-\ell},
\label{wanted}
\ee
where we introduced two sets of vectors constructed from the vectors \eqref{defvij},
\be
\bfv_\ell=\sum_{i=1}^\ell\sum_{j=\ell+1}^n\bfv_{ij}\,,
\qquad
\bfu_\ell=\sum_{i=1}^\ell\sum_{j=\ell+1}^n\bfu_{ij}\,.
\label{deftvl}
\ee
Note that $(\bfv_\ell,\bfx)$ and $(\bfu_\ell,\bfx)$ correspond to the quantities $-\Gamma_{n\ell}$ and $-\cs_\ell$ \eqref{notaion-c}, respectively.

To solve the above constraints, let us consider the following iterative ansatz
\be
\whg_n(\{\gama_i,c_i\})=\gf_n(\{\gama_i,c_i\})
-\Sym\left\{\sum_{n_1+\cdots +n_m= n\atop n_k\ge 1, \ m<n}
\whg_m(\{\gama'_k,c'_k\})
\prod_{k=1}^m \cE_{n_k}(\gama_{j_{k-1}+1},\dots,\gama_{j_{k}})
\right\},
\label{iterDn}
\ee
where the notations for indices and primed variables are the same as in \eqref{groupindex}.
This ansatz is motivated by analogy with the iterative equation for the (partial) tree index \eqref{F-ansatz}.
It involves two functions to be determined below, $\gf_n$ and $\cE_{n}$.
The former depend on the moduli through the variables $c_i$ \eqref{fiparam}, whereas the latter are moduli independent.
We set $\gf_1=\cE_1=1$ and also assume that $\gf_n$ have discontinuities only at walls of marginal stability,
i.e. at $\sum_{i\in\cI} c_i=0$ for various subsets $\cI$ of indices.

The unknown functions $\gf_n$ and $\cE_{n}$ together with the functions $R_n$ defining the completion, or their combinations \eqref{defSm},
are fixed by the conditions \eqref{recurs-whd-full} and \eqref{wanted}.
In appendix \ref{ap-proofs} we prove the following result:

\begin{proposition}\label{prop-solvecond}
Let us split $\cE_n=\cEf_n+\cEp_n$ into
$\cEp_n$, which is the part exponentially decreasing
for large $\tau_2$, and the non-decreasing
remainder $\cEf_n$.
Then the ansatz \eqref{iterDn} satisfies the recursive equation \eqref{wanted} provided
\begin{enumerate}
\item
the functions $\gf_n$ are subject to a similar recursive relation
\be
\begin{split}
&\,
\frac14\Sym\left\{\sum_{\ell=1}^{n-1} \bigl( \sign(\cs_\ell)-\sign (\Gamma_{n\ell})\bigr)\,
\kappa(\Gamma_{n\ell})\,
\gf_{\ell}(\{\gama_i,\cl_i\}_{i=1}^\ell)\,
\gf_{n-\ell}(\{\gama_i,\cl_i\}_{i=\ell+1}^n)
\right\}
\\
&\,\qquad\qquad\qquad
=\gf_n(\{\gama_i,c_i\})-\gf_n(\{\gama_i,\beta_{ni}\}),
\end{split}
\label{iterseed}
\ee
where $\cs_k$, $\beta_{k\ell}$, $\Gamma_{k\ell}$ and $\cl_i$ were defined in \eqref{notaion-c}, \eqref{flowc};

\item
the non-decreasing part of $\cE_n$ is fixed in terms of $\gf_n(\{\gama_i,c_i\})$ as
\be
\cEf_n(\{\gama_i\})=\gf_n(\{\gama_i,\beta_{ni}\});
\label{rel-gE}
\ee

\item
its decreasing part is given by
\be
\cEp_n(\{\gama_i\})=-\Sym\left\{\sum_{n_1+\cdots +n_m= n\atop n_k\ge 1, \ m>1}
W_m(\{\gama'_k\})
\prod_{k=1}^m \cE_{n_k}(\gama_{j_{k-1}+1},\dots,\gama_{j_{k}})\right\}.
\label{Emrec}
\ee
\end{enumerate}
Furthermore, provided the functions $\cE_n$ depend on electric charges $q_{i,a}$ only through the DSZ products $\gamma_{ij}$
and their kernels $\cEPhi_n(\bfx)$ defined as in \eqref{kerPhiint} are smooth solutions of Vign\'eras' equation,
\be
V_{n-1}\cdot \cEPhi_n=0,
\ee
then the ansatz \eqref{iterDn} also satisfies the modularity constraint \eqref{wanted}.
\end{proposition}

This proposition allows in principle to fix all unknown functions. Indeed,
the recursive relation \eqref{iterseed} determines all $\gf_n$.
Then equations \eqref{rel-gE} and \eqref{Emrec} give the two parts of $\cE_n=\cEf_n+\cEp_n$ in terms of $\gf_n$
and $W_n$. At this point the latter are still undetermined and are defined in terms of the unknown functions $R_n$.
The additional constraint that $\cE_n$ should satisfy Vign\'eras' equation links together $\cEf_n$ and $\cEp_n$
and thus establishes a relation between $W_n$
and $\gf_n$. Lastly, inverting \eqref{defSm} generates a solution for $R_n$.

\medskip

We end this discussion by making an observation which will become relevant in the
next subsection: by comparing  \eqref{multihd-full} and \eqref{multih}, it is clear that  $\whg_n$ must agree with
the tree index $\gtr$ in the limit $\tau_2\to\infty$. In particular, $\gtr$ satisfies a similar ansatz
as \eqref{iterDn}, upon replacing $\cE_n$ by its non-decaying part:
\be
\gtr(\{\gama_i,c_i\})=\gf_n(\{\gama_i,c_i\})
-\Sym\left\{\sum_{n_1+\cdots +n_m= n\atop n_k\ge 1, \ m<n}
g_{{\rm tr},m}(\{\gama'_k,c'_k\})
\prod_{k=1}^m \cEf_{n_k}
\right\}.
\label{iterDn0}
\ee
It follows that the function $\gf_n$ should be independent of $\tau_2$,
or at least that any such dependence should cancel in the recursion \eqref{iterDn0}.

\subsection{Generalized error functions and the completion}
\label{subsec-fullcompl}

From the previous discussion, the first step in the construction of the modular completion
is to provide an explicit expression for $\gf_n$.
Once such an expression is known, all other functions can be determined algebraically.
The problem, however, is that the solution of the recursive relation \eqref{iterseed} is not unique.
On the other hand, $\gf_n$ determines the non-decaying part $\cEf_n$ of $\cE_n$, which is
strongly constrained by the fact that $\cE_n$ must satisfy Vign\'eras' equation.
This restriction turns out to be strong enough to remove the redundancy in the solution of \eqref{iterseed},
but it shows that we have to solve simultaneously two problems: satisfy the recursive relation \eqref{iterseed}
and ensure that it can be promoted to a solution of Vign\'eras' equation.
We will do this by first constructing a solution of Vign\'eras' equation with the asymptotic part
possessing the properties expected from $\gf_n$, and then proving that it does satisfy the recursive relation.

A hint towards a solution of these two problems can be found by examining the
 form of the kernel $\intPhi_n$ \eqref{Phin-final}.
The point is that both $\intPhi_n$ and $\gf_n$ have discontinuities at walls of marginal stability
which, as we know, must cancel each other. Furthermore, they should recombine into a smooth solution of Vign\'eras' equation.
Thus, $\intPhi_n$ is expected to encode at least part of the completion of $\gf_n$.
In addition, as explained in appendix \ref{ap-generror}, the function $\tPhi^M_{n-1}$, from which $\intPhi_n$ is constructed,
appears as the term with fastest decay in the expansion similar to \eqref{expPhiE} of the function $\tPhi^E_{n-1}$ defined in \eqref{deftPhigen}
with the arguments $\cV=\{ \bfu_e\}$, $\tcV=\{\bfv_{s(e) t(e)}\}$, which is a smooth solution of Vign\'eras' equation.
This motivates us to consider
the following function\footnote{Note that the sign factor $(-1)^{\sum_{i<j} \gamma_{ij} }$ is equal
to the ratio of quadratic refinements appearing in \eqref{kerPhiint}.}
\be
G_n(\{\gama_i,c_i\};\tau_2)=
\frac{(-1)^{\sum_{i<j} \gamma_{ij} }}{(\sqrt{2\tau_2})^{n-1}}\, \tPhif_n(\bfx),
\label{defgf-hPhi}
\ee
where $\tPhif_n(\bfx)$ denotes the large $\bfx$ limit of the function
\be
\tPhi_n(\bfx)=
\frac{1}{2^{n-1} n!}\sum_{\cT\in\, \IT_n^\ell}
\tPhi^E_{n-1}(\{ \bfu_e\}, \{\bfv_{s(e) t(e)}\};\bfx).
\label{fullkeru}
\ee
Our first goal will be to evaluate this limit explicitly.

\lfig{An example of marked unrooted tree belonging to $\IT_{7,4}^\ell$ where near each vertex we showed in red
an integer counting marks.}{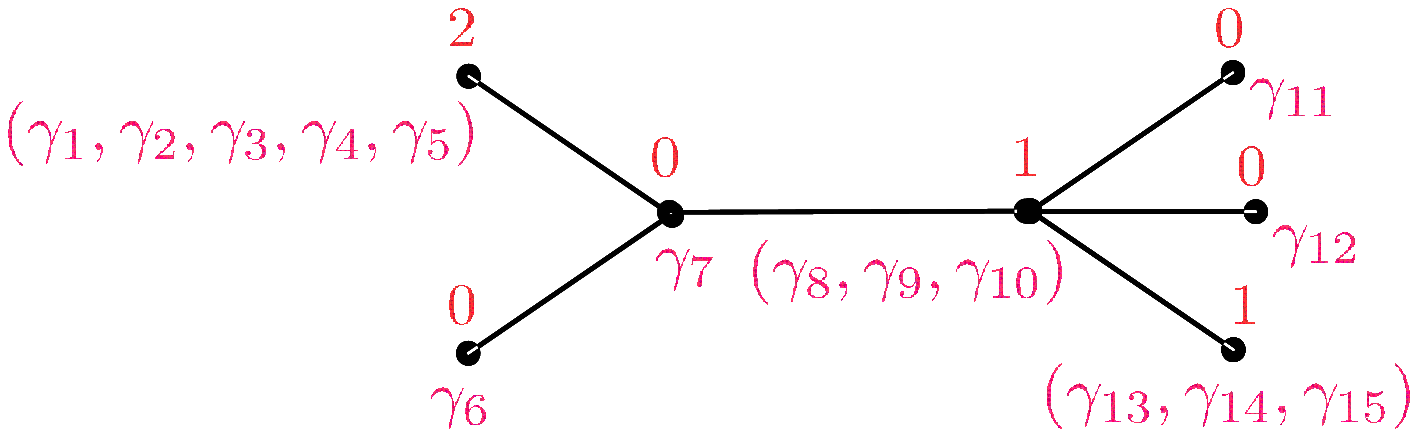}{10.6cm}{fig-mark-trees}{-2.cm}

To express the result, we need to introduce two new types of trees, beyond those already encountered.
First, we denote by $\IT_{n,m}^\ell$ the set of {\it marked} unrooted labelled trees with $n$ vertices and $m$ marks assigned to vertices
(see Fig. \ref{fig-mark-trees}).
Let $m_\ver\in \{0,\dots m\}$ be the number of marks carried by the vertex $\ver$, so that $\sum_\ver m_\ver=m$.
Furthermore, the vertices are decorated by charges from the set $\{\gamma_1,\dots,\gamma_{n+2m}\}$ such that a vertex $\ver$ with
$m_\ver$ marks carries $1+2m_\ver$ charges $\gamma_{\ver,s}$, $s=1,\dots,1+2m_\ver$ and we set $\gamma_\ver=\sum_{s=1}^{1+2m_\ver}\gamma_{\ver,s}$.
Second, we define $\IT_{n}^{(3)}$ to be the set of (unordered, full) rooted ternary trees with $n$ leaves decorated by charges $\gamma_i$ and
other vertices carrying charges defined by the inductive rule $\gamma_v=\gamma_{d_1(v)}+\gamma_{d_3(v)}+\gamma_{d_3(v)}$ where $d_r(v)$ are
the three children of vertex $v$ (see Fig. \ref{fig-ternarytrees}). As usual, $V_T$ denotes the set of
vertices, with cardinality $|V_T|=\hf(n-1)$ (not counting the leaves).

\lfig{An example of rooted ternary tree belonging to $\IT_{9}^{(3)}$.}{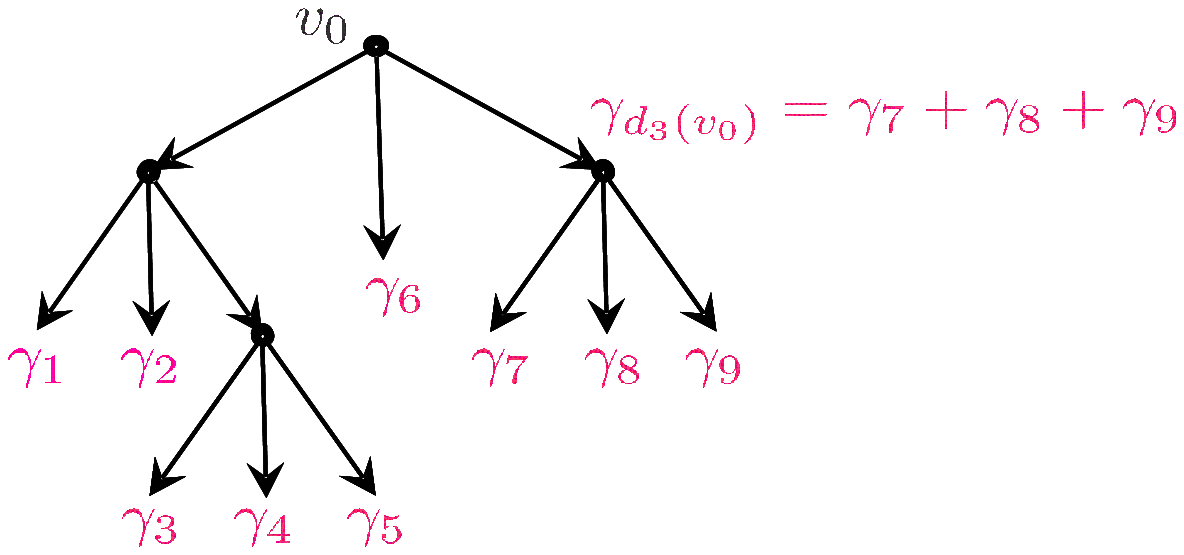}{10.cm}{fig-ternarytrees}{-1.5cm}

In terms of these notations, we then have the following result proven, in appendix \ref{ap-proofs}:

\begin{proposition}\label{prop-limitG}
The function \eqref{defgf-hPhi} is given by
\be
G_n(\{\gama_i,c_i\};\tau_2)=
\scalebox{.95}{$\frac{(-1)^{n-1+\sum_{i<j} \gamma_{ij} }}{2^{n-1} n!}$}
\sum_{m=0}^{[(n-1)/2]}\frac{(-1)^m}{(4\pi\tau_2)^m}
\!\!\!\!\sum_{\cT\in\, \IT_{n-2m,m}^\ell}
\prod_{\ver\in V_\cT}\cP_{m_\ver}(\{p_{\ver,s}\})\prod_{e\in E_{\cT}}\gamma_{s(e) t(e)}\,\sign(\cs_e),
\label{asymp-compl}
\ee
where, for each edge $e$ joining the subtrees $\cT_e^s$ and $\cT_e^t$,
\be
\cs_e=\sum_{i\in V_{\cT_e^s}}c_i
=\sum_{i\in V_{\cT_e^s}}\sum_{j\in V_{\cT_e^t}}\Im\bigl[ Z_{\gamma_i}\bZ_{\gamma_j}\bigr],
\label{defce}
\ee
and $\cP_m$ is the weight corresponding to the marks, given by
\be
\cP_{m}(\{p_s\})=\sum_{T\in\, \IT_{2m+1}^{(3)}}\frac{1}{T!}\prod_{v\in V_T}(p_{d_1(v)}p_{d_2(v)}p_{d_3(v)}).
\label{defPver-tree}
\ee
Here $T!$ is the tree factorial
\be
T!= \prod_{v\in V_T}n_v(T)\, ,
\label{defcT}
\ee
where $n_v(T)$, as in Theorem \ref{theorem}, denotes the number of vertices (excluding the leaves) of the subtree of $T$ rooted at $v$.
\end{proposition}

To demystify the origin of these structures, note that
the sum over $m$ in \eqref{asymp-compl} arises because of the contributions coming from the mutual action of covariant derivative operators
$\cD(\bfv_{s(e) t(e)})$ \eqref{defcDif} in the definition of the function $\tPhi^E_{n-1}$ \eqref{deftPhigen}.
Such action is non-vanishing for any pair of intersecting edges $(e_1,e_2)$ and is proportional to
$(p_{\ver_1}p_{\ver_2}p_{\ver_3})$ where $\ver_1,\ver_2,\ver_3$ are the three vertices
belonging to the edges. While in appendix \ref{ap-twistint} it is shown that such contributions cancel in the sum over trees defining
$\intPhi_n$ \eqref{Phin-final}, this is not so for $\tPhif_n(\bfx)$: instead of the identity \eqref{case3}, one
has to apply the sign identity \eqref{signfactorscomb} which produces a constant term.
As a result, instead of the standard sign factors assigned to $(e_1,e_2)$, one finds the weight factor $(p_{\ver_1}p_{\ver_2}p_{\ver_3})$.
Furthermore, the sum over trees ensures that all other factors except
for this weight depend on the charges $\gamma_{\ver_1},\gamma_{\ver_2},\gamma_{\ver_3}$ only through their sum.
One can keep track of such contributions by collapsing the corresponding pairs of edges
on the tree and marking the remaining vertices with weights $m_\ver$.
Essentially, the only non-trivial point is to understand the form of the weight factor $\cP_m$ for $m>1$.
In this case more than one pair of joint edges collapse.
The representation \eqref{defPver-tree} in terms of rooted ternary trees is obtained by
collapsing $m$ pairs of edges successively,
where the coefficient $1/T!$ takes into account that such procedure leads
to an overcounting of different assignments of labels.

\lfig{Three unrooted trees constructed from the same three subtrees.}
{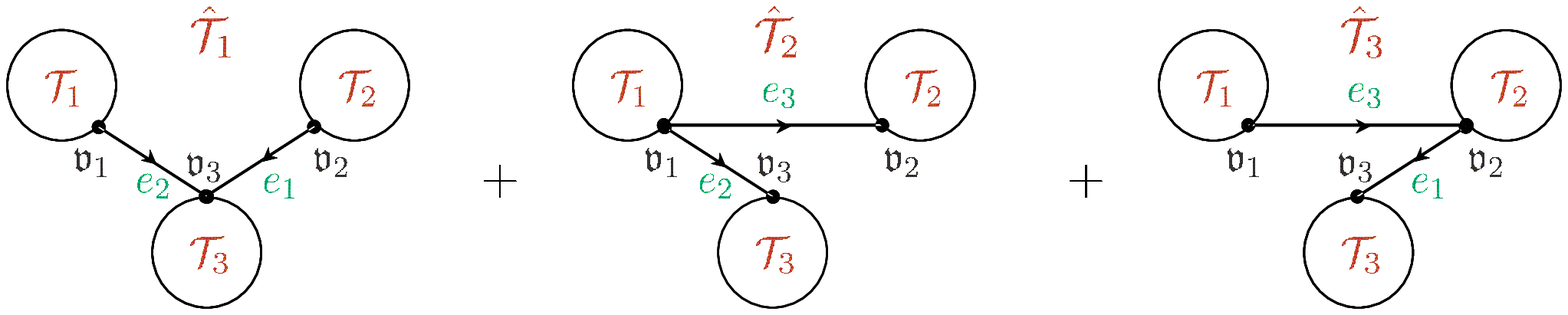}{17.5cm}{fig-Vign3}{-1.5cm}

\medskip

Unfortunately, the function \eqref{asymp-compl} cannot yet be taken as an ansatz for $\gf_n$ because it depends non-trivially on
the modulus $\tau_2$, which is in tension with the observation made at the end of the previous
subsection. Therefore, we shall
modify the function \eqref{fullkeru} into a function which is still a smooth solution of Vign\'eras' equation,
but whose large $\bfx$ limit is independent of $\tau_2$.  Taking cue from
the structure of \eqref{asymp-compl},
we define\footnote{The term $m=0$ in \eqref{fullker-mod} reduces to the original function \eqref{fullkeru},
while the terms with $m>0$ are the afore-mentioned modification.}
\be
\tcEPhi_n(\bfx)=
\frac{1}{2^{n-1} n!}\sum_{m=0}^{[(n-1)/2]}\sum_{\cT\in\, \IT_{n-2m,m}^\ell}
\[\prod_{\ver\in V_\cT}\cD_{m_\ver}(\{\gama_{\ver,s}\})\]
\tPhi^E_{n-2m-1}(\{ \bfu_e\}, \{\bfv_{s(e) t(e)}\};\bfx),
\label{fullker-mod}
\ee
where
\be
\cD_{m}(\{\gama_s\})=\sum_{\cT'\in\, \IT_{2m+1}^\ell} a_{\cT'}\prod_{e\in E_{\cT'}}\cD(\bfv_{s(e)t(v)}).
\label{defcDcT}
\ee
The coefficients $a_\cT$ are rational numbers which depend (up to a sign determined by the orientation of edges)
only on the topology of the tree.
We fix them by requiring that they satisfy the following system of equations
\be
a_{\hcT_1}+a_{\hcT_2}-a_{\hcT_3}=a_{\cT_1}a_{\cT_2}a_{\cT_3},
\label{eqcoef}
\ee
where $\cT_r$ ($r=1,2,3$) are arbitrary unrooted trees with marked vertex $\ver_r$ and $\hcT_r$ are the three trees
constructed from $\cT_r$ as shown in Fig. \ref{fig-Vign3}.\footnote{For these equations, it is important to take into account
the orientation of the edges: a change of orientation of an edge flips the sign of $a_\cT$.
The equations \eqref{eqcoef} are written assuming the orientation shown in Fig. \ref{fig-Vign3}, namely
$e_1=(\ver_2,\ver_3)$, $e_2=(\ver_1,\ver_3)$, $e_3=(\ver_1,\ver_2)$.}
This system of equations is imposed in order to ensure certain properties of the operators \eqref{defcDcT}
which are crucial for the cancellation of $\tau_2$-dependent terms in the asymptotics of $\tcEPhi_n$
(see \eqref{cond-coefaT}).
It is readily seen that the system \eqref{eqcoef} fixes all coefficients $a_\cT$
recursively starting from $a_\bullet=1$, $a_{\bullet\!\mbox{-}\!\bullet}=0$
and going to trees with $n\geq 3$ vertices. (Starting from $n=7$, the system \eqref{eqcoef} is overdetermined,
but it can be checked that it does have a unique solution,
with $a_\cT=0$ for trees with even number of vertices.)
For a generic tree, in appendix \ref{ap-proofs} we prove the following

\begin{proposition}\label{prop-coeff}
For a tree $\cT$ with $n$ vertices the coefficient $a_\cT$ is given recursively by
\be
a_\cT=\frac{1}{n}\sum_{\ver\in V_\cT} \epsilon_\ver\prod_{s=1}^{n_\ver} a_{\cT_s(\ver)},
\label{res-aT}
\ee
where $n_\ver$ is the valency of the vertex $\ver$,
$\cT_s(\ver)$ are the trees obtained from $\cT$ by removing this vertex, and $\epsilon_\ver$ is the sign determined
by the choice of orientation of edges,
$
\epsilon_\ver=(-1)^{n_\ver^+}
$
with $n_\ver^+$ being the number of incoming edges at the vertex.
(See Fig. \ref{fig-coef} for an example.)
\end{proposition}

\lfig{An example of calculation of the coefficient $a_\cT$ for a tree with 7 vertices.
We indicated in red the only vertices which produce non-vanishing contributions to the sum over vertices.
For other vertices at least one of the trees $\cT_s(\ver)$ has even number of vertices and hence vanishing coefficient.}
{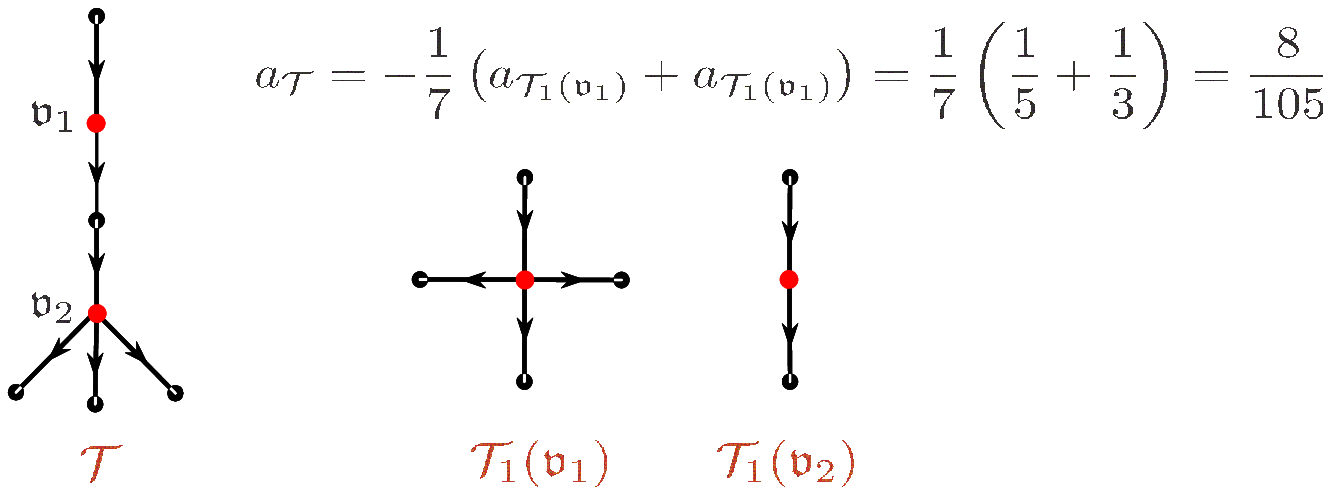}{12.6cm}{fig-coef}{-1.4cm}

Let us now return to the function \eqref{fullker-mod}, which is now fully specified given the prescription \eqref{res-aT}
for the coefficients $a_\cT$.
Similarly to \eqref{fullkeru} (which coincides with the $m=0$ contribution
in \eqref{fullker-mod}), it is a solution of Vign\'eras' equation for $\lambda=n-1$.
We claim that in the large $\bfx$ asymptotics of this function, all $\tau_2$-dependent contributions,
coming from the mutual action of derivative operators $\cD$, cancel.
More precisely, the asymptotics is given by the following

\begin{proposition}\label{prop-limitcompl}
\be
\lim_{\bfx\to\infty}\tcEPhi_n(\bfx)=\frac{1}{2^{n-1} n!}\sum_{m=0}^{[(n-1)/2]}
\sum_{\cT\in\, \IT_{n-2m,m}^\ell}
\prod_{\ver\in V_\cT}\cV_{m_\ver}(\{\gama_{\ver,s}\})
\prod_{e\in E_{\cT}}(\bfv_{s(e) t(e)},\bfx)\,\sign(\bfu_e,\bfx),
\label{asymphcEPhi}
\ee
where
\be
\cV_{m}(\{\gama_{s}\})
=\sum_{\cT\in\, \IT_{2m+1}^\ell} a_\cT\prod_{e\in E_\cT}(\bfv_{s(e)t(v)},\bfx).
\label{defcV-mw}
\ee
\end{proposition}

Importantly the function \eqref{asymphcEPhi} is locally a homogeneous polynomial of degree $n-1$.
Therefore, when written in terms of charges, the $\tau_2$-dependence factorizes
and is cancelled after the same rescaling as in \eqref{defgf-hPhi}.
This naturally leads us to the following ansatz for the function $\gf_n$,
which we prove in appendix~\ref{ap-proofs}:
\begin{proposition}
\label{prop-solveiterg0}
The function
\be
\gf_n(\{\gama_i,c_i\})=\frac{(-1)^{n-1+\sum_{i<j} \gamma_{ij} }}{2^{n-1} n!}
\sum_{m=0}^{[(n-1)/2]}\sum_{\cT\in\, \IT_{n-2m,m}^\ell}
\prod_{\ver\in V_\cT}\tcV_{m_\ver}(\{\gama_{\ver,s}\})
\prod_{e\in E_{\cT}}\gamma_{s(e) t(e)}\,\sign(\cs_e),
\label{defDf-gen}
\ee
where
\be
\tcV_{m}(\{\gama_{s}\})=\sum_{\cT'\in\, \IT_{2m+1}^\ell} a_{\cT'}\prod_{e\in E_{\cT'}}\gamma_{s(e)t(e)},
\label{deftcV-mw}
\ee
satisfies the recursive equation \eqref{iterseed}.
\end{proposition}
\noindent
Note that the same recursive equation \eqref{iterseed} is also obeyed by
the contribution of unmarked trees (i.e. $m=0$) to $\gf_n$, which we denote by $\gs_n$ (see \eqref{defgf0}).
However, the latter  cannot be obtained as the large $\bfx$ limit
of a solution of Vign\'eras' equation, which is why we have to consider the more complicated function \eqref{defDf-gen}.

\medskip

Given the result for $\gf_n$ and the relation \eqref{rel-gE}, one obtains an explicit
expression for the non-decreasing part of $\cE_n$,
\be
\cEf_n(\{\gama_i\})=
\frac{(-1)^{\sum_{i<j} \gamma_{ij} }}{2^{n-1} n!}
\sum_{m=0}^{[(n-1)/2]}\sum_{\cT\in\, \IT_{n-2m,m}^\ell}
\prod_{\ver\in V_\cT}\tcV_{m_\ver}(\{\gama_{\ver,s}\})
\prod_{e\in E_{\cT}}\gamma_{s(e) t(e)}\,\sign(\Gamma_e),
\label{defEn0new}
\ee
where
\be
\Gamma_e=\sum_{i\in V_{\cT_e^s}}\sum_{j\in V_{\cT_e^t}}\gamma_{ij}.
\label{defGame}
\ee
From Proposition \ref{prop-solvecond} we know that the kernel $\cEPhi_n$
corresponding to the full function $\cE_n$ must satisfy Vign\'eras' equation
and have the asymptotics captured by \eqref{defEn0new}.
But we already know a solution of Vign\'eras' equation with a very similar asymptotics,
namely the function \eqref{fullker-mod}, whose asymptotics differs only in the fact that
the vectors $\bfv_e$ are replaced by $\bfu_e$. Since this replacement does not affect the proof of Vign\'eras' equation,
we immediately arrive at the following result:
\be
\cE_n(\{\gama_i\};\tau_2)=
\frac{(-1)^{\sum_{i<j} \gamma_{ij} }}{(\sqrt{2\tau_2})^{n-1}}\, \cEPhi_n(\bfx),
\label{rescEn}
\ee
where
\be
\cEPhi_n(\bfx)=
\frac{1}{2^{n-1} n!}\sum_{m=0}^{[(n-1)/2]}
\sum_{\cT\in\, \IT_{n-2m,m}^\ell}
\[\prod_{\ver\in V_\cT}\cD_{m_\ver}(\{\gama_{\ver,s}\})\]
\tPhi^E_{n-2m-1}(\{ \bfv_e\}, \{\bfv_{s(e) t(e)}\};\bfx).
\label{fullker}
\ee
In words, the function $\cE_n$
is a sum over marked unrooted labelled trees of solutions
of Vign\'eras equation with $\lambda=n-1$, obtained from the standard
generalized error functions $\Phi^E_{n-2m-1}$ by acting with $n-1$ raising operators.

\subsection{The structure of the completion}

After substituting into the iterative ansatz \eqref{iterDn},
the two results \eqref{defDf-gen} and \eqref{fullker} completely specify the coefficients
of the expansion \eqref{multihd-full} of the generating function of DT invariants in terms of $\whh_{p,\mu}$.
The result of the iteration can in fact be written explicitly as a sum over  Schr\"oder trees.
Adopting the same notations as in Proposition \ref{prop-iterwhg},
one has
\begin{proposition}\label{prop-reswhg}
The function $\whg_n(\{\gama_i,c_i\})$ is given by the sum over Schr\"oder trees with $n$
leaves,
\be
\whg_n= \Sym\left\{\sum_{T\in\IT_n^{\rm S}}(-1)^{n_T-1} \(\gf_{v_0}-\cE_{v_0}\)\prod_{v\in V_T\setminus{\{v_0\}}}\cE_{v}\right\},
\label{soliterg}
\ee
where $v_0$ is the root of the tree.
\end{proposition}

We conclude that the properties of $\gf_n$ and $\cE_n$ ensure that the theta series appearing
in the corresponding decomposition of $\cG$ is modular so that
the modularity of $\cG$ requires the function $\whh_{p,\mu}$ to be a vector valued (non-holomorphic)
modular form of weight $(-\hf b_2-1,0)$. The functions $R_n$ entering the
non-holomorphic  completion $\whh_{p,\mu}$,
are then given by the following proposition, whose
proof can again be found in appendix~\ref{ap-proofs}:

\begin{proposition}\label{prop-resR}
Inverting the relations \eqref{defSm} and \eqref{Emrec}, one obtains\footnote{In the proof of this proposition
we obtain a similar formula \eqref{solWn} for the coefficients $W_n$ appearing in the inverse relation \eqref{exp-hwh}.}
\be
R_n= \Sym\left\{\sum_{T\in\IT_n^{\rm S}}(-1)^{n_T-1} \cEp_{v_0}\prod_{v\in V_T\setminus{\{v_0\}}}\cEf_{v}\right\}.
\label{solRn}
\ee
\end{proposition}

It is important to check that the resulting functions $R_n$ are invariant under the spectral flow of the total charge.
This is in fact a simple consequence of the fact that the functions $\cE_n$ entering their definition
depend on the electric charges only through the DSZ products $\gamma_{ij}$.
As follows from \eqref{gamZZ}, all such products are invariant under the spectral
flow of the total charge, hence $R_n$ are invariant as well.


Since $\cEf_n$ are $\tau_2$-independent, all non-holomorphic dependence of $\whh_{p,\mu}$ comes from the factor $\cEp_{v_0}$
in the expression \eqref{solRn} for $R_n$. Expressing the holomorphic anomaly in terms of the modular functions $\whh_{p_i,\mu_i}$,
one obtains the following

\begin{proposition}\label{prop-derbtau}
The holomorphic anomaly of the completion is given by
\be
\p_{\bar\tau}\whh_{p,\mu}(\tau)= \sum_{n=2}^\infty
\sum_{\sum_{i=1}^n \gama_i=\gama}
\cJ_n(\{\gama_i\},\tau_2)
\, e^{\pi\I \tau Q_n(\{\gama_i\})}
\prod_{i=1}^n \whh_{p_i,\mu_i}(\tau),
\label{exp-derwh}
\ee
where
\be
\cJ_n= \frac{\I}{2}\Sym\left\{\sum_{T\in\IT_n^{\rm S}}(-1)^{n_T-1} \p_{\tau_2}\cE_{v_0}\prod_{v\in V_T\setminus{\{v_0\}}}\cE_{v}\right\}.
\label{solJn}
\ee
\end{proposition}

Note that this result is consistent with the fact that $\tau_2^2\p_{\bar\tau}\whh_{p,\mu}$ is a modular function of weight $(-\hf b_2-3,0)$.
Indeed, the sum over D2-brane charges forms a theta series for a lattice of dimension $(n-1)b_2$ (see footnote \ref{foot-D2specflow})
and  quadratic form $-Q_n$ of signature $(n-1,(n-1)(b_2-1))$.
Furthermore, since
\be
2\tau_2\p_{\tau_2}\cE_n=
\frac{(-1)^{\sum_{i<j} \gamma_{ij} }}{(\sqrt{2\tau_2})^{n-1}}\(\bfx\cdot\p_\bfx-(n-1)\) \cEPhi_n(\bfx)
\ee
and $V_{\lambda-2}\(\bfx\cdot\p_\bfx-\lambda\)= \(\bfx\cdot\p_\bfx-(\lambda-2)\)V_\lambda$,
the function $\tau_2\cJ_n$ is a solution of Vign\'eras' equation with $\lambda=n-3$.
Thus, after multiplying Eq. \eqref{exp-derwh} by $\tau_2^2$, the theta series generated by the sum over D2-brane charges
is modular of weight $(\frac{n-1}{2}\, b_2+n-3,0)$.
Combining it with $n$ factors of $\whh_{p_i,\mu_i}$, one recovers the correct weight
for $\tau_2^2\p_{\bar\tau}\whh_{p,\mu}$.


In appendix \ref{ap-explicit} we present explicit expressions for various elements of our construction
up to order $n=4$. Based on these results, we conjecture a general formula for the kernel $\whPhi^{{\rm tot}}_n$ of the theta series appearing
in the expansion of $\cG$ in terms of $\whh_{p,\mu}$:
\begin{conj}
The instanton generating function has the theta series decomposition
\be
\cG=\sum_{n=1}^\infty\frac{2^{-\frac{n}{2}}}{\pi\sqrt{2\tau_2}}\[\prod_{i=1}^{n}
\sum_{p_i,\mu_i}\sigma_{p_i}\whh_{p_i,\mu_i}\]
e^{-S^{\rm cl}_p}\vartheta_{\bfp,\bfmu}\bigl(\whPhi^{{\rm tot}}_{n},n-2\bigr),
\label{treeFh-flh}
\ee
where the kernels of the theta series are given by
\be
\whPhi^{{\rm tot}}_n= \intPhi_1\Sym\left\{\sum_{T\in\IT_n^{\rm S}}(-1)^{n_T-1}
\(\tcEPhi_{v_0}-\cEPhi_{v_0}\)\prod_{v\in V_T\setminus{\{v_0\}}}\cEPhi_{v}\right\}.
\label{kertotsol}
\ee
\end{conj}
\noindent
Note that this result automatically implies that the theta series
$\vartheta_{\bfp,\bfmu}\bigl(\whPhi^{{\rm tot}}_{n},n-2\bigr)$ is modular, since its kernel
is a solution of Vign\'eras' equation.

\medskip

We end the discussion of these results by observing that the formula \eqref{soliterg} allows to
obtain a new representation of the tree index $\gtr$. Indeed, as
observed at the end of Subsection \ref{subsec-complh}, $\gtr$ agrees with
$\whg_n$ in  the limit $\tau_2\to\infty$, which amounts to replacing
$\cE_{v}$ by its non-decaying part $\cEf_{v}$ in \eqref{soliterg}. In fact,
as shown in appendix \ref{ap-proofs}, all contributions due to trees with non-zero number of marks
in \eqref{defDf-gen} and \eqref{defEn0new} cancel in the resulting expression, leaving only the contributions coming from the terms with $m=0$.
Thus, one arrives at the following representation:
\begin{proposition}\label{prop-treeindex}
The tree index $\gtr$ defined in \eqref{defgtree} can be expressed as
\be
\begin{split}
\gtr
=&\, \Sym\left\{\sum_{T\in\IT_n^{\rm S}}(-1)^{n_T-1} \(\gs_{v_0}-\cEs_{v_0}\)\prod_{v\in V_T\setminus{\{v_0\}}}\cEs_{v}\right\},
\end{split}
\label{soliterg-tr}
\ee
where $\gs_n$ is defined in \eqref{defgf0} and $\cEs_n=\gs_n(\{\gama_i,\beta_{ni}\})$.
\end{proposition}
\noindent
Note that this representation is more explicit than the one given in section \ref{subsec-partialindex}
since it does not require taking the limit $y\to 1$.

\medskip

In fact, the mechanism by which contributions of marked trees cancel in the sum over Schr\"oder trees,
explained in the proof of Proposition \ref{prop-treeindex}, is very general and
applies just as well to all the above equations \eqref{soliterg}-\eqref{kertotsol}.
As a result, all of them remain valid if we replace $\tcEPhi_n$ and $\cEPhi_n$, respectively, by
$\tPhi_n$ \eqref{fullkeru} and by its cousin with the vectors $\{\bfu_e\}$ replaced by $\{\bfv_e\}$,
\be
\Phi_n(\bfx)=
\frac{1}{2^{n-1} n!}\sum_{\cT\in\, \IT_n^\ell}
\tPhi^E_{n-1}(\{ \bfv_e\}, \{\bfv_{s(e) t(e)}\};\bfx).
\label{fullkerv}
\ee
Correspondingly, $\gf_n$ and $\cEf_n$ should then be replaced by their asymptotics
at large $\tau_2$, given by $G_n(\{\gama_i,c_i\})$ and $G_n(\{\gama_i,\beta_{ni}\})$.
The proof that such replacement is possible is completely analogous to the proof of Proposition \ref{prop-treeindex},
and we refrain from presenting it.
This shows that from the very beginning we could take the function $G_n$ as the ansatz for $\gf_n$
since its $\tau_2$-dependence  cancels in the formula for the tree index.
We could then avoid most of complications of section \ref{subsec-fullcompl}.
However, we cannot avoid the introduction of marked trees since they appear in the expression for $G_n$ \eqref{asymp-compl} anyway.
Moreover, we prefer to stick to the definition \eqref{defDf-gen} because of two reasons.
Firstly, due to its $\tau_2$-independence, $\gf_n$ might itself be interpretable as an index.
Secondly, as we discuss in the next subsection, it is related to other potentially interesting representations,
which might be at the basis of such interpretation.

\subsection{Alternative representations}
\label{subsec-refine}

The construction presented above provides an explicit expression for the completion $\whh_{p,\mu}$ and all other relevant quantities.
Roughly, it can be split into three levels:
\begin{enumerate}
\item
sum over Schr\"oder trees \eqref{solRn};
\item
sum over marked unrooted labelled trees \eqref{fullker};
\item
sum over unrooted labelled trees defining the weights associated to marks \eqref{defcDcT}.
\end{enumerate}
The complication due to the appearance of marks and hence the last level arises due to
the non-orthogonality of the vectors $\bfv_{s(e) t(e)}$ appearing as arguments
of the generalized error functions uplifted to solutions of Vign\'eras' equation for $\lambda=n-1$.
However, the two sums over unrooted trees are organized is such a way that all additional contributions
due to this non-orthogonality cancel. This suggests that there should exist a simpler representation
where the vectors appearing in the second argument of functions $\tPhi^E_{n-1}(\cV, \tcV;\bfx)$
are mutually orthogonal from the very beginning.
Below we present some results showing that such a representation does exist at least for
low values of $n$. We then move on to present yet another representation of the functions
$\gf_n$ which is significantly simpler than \eqref{defDf-gen}, although its equivalence
with the latter remains conjectural.

\subsubsection{Simplified representation via flow trees}

We have shown, analytically for $n=2,3,4$ (see appendix \ref{ap-explicit}) and numerically for $n=5$,
that the function \eqref{defDf-gen} can be rewritten as
\be
\gf_n(\{\gama_i,c_i\})
= \Sym\left\{\frac{d_n}{2^{n-1}}\sum_{T\in \IP_n}\kappa(T) \prod_{k=1}^{n-1}\sgn(\cs_k)\right\}
\qquad (n\le 5),
\label{gfn-upto5}
\ee
where $\kappa(T)$ is the factor defined in \eqref{kappaT0},
whereas $d_n$ and $\IP_n\subset\IT_n^{\rm af}$ are numerical coefficients and subsets of
flow trees with $n$ leaves, respectively, which can be chosen as follows:
\be
\begin{split}
d_1=&1,
\qquad
d_2=\hf\, ,
\qquad
d_3=\frac16\, ,
\qquad
d_4=\frac{1}{12}\, ,
\qquad
d_5=\frac{1}{30}\, ,
\\
 \quad
\IP_1=&\{(1)\},
\qquad
\IP_2=\{(12)\},
\qquad
\IP_3=\{((12)3),(1(23))\},
\\
\IP_4=&\{((1(23))4),((12)(34)),(1((23)4))\},
\\
\IP_5=& \{((((31)4)2)5),(((12)(34))5),((1((23)4))5),(((12)3)(45)),
\\
&
\qquad\
((12)(3(45))),(1((2(34))5)),(1((23)(45))),(1(4(2(53))))\},
\end{split}
\label{knowncases}
\ee
where we labelled ordered flow trees using 2-bracketings as explained in footnote \ref{foobracket}.

Unfortunately, the simple ansatz \eqref{gfn-upto5} fails beyond the fifth order. 
For instance, numerical experiments indicate that for $n=6$,
one should include an additional term given by a product of $n-3=3$ sign functions.
This suggests the following more general ansatz:

\begin{conj}
The function \eqref{defDf-gen} can be rewritten as
\be
\gf_n(\{\gama_i,c_i\})
= \Sym\left\{
\sum_{n_1+\cdots +n_m= n\atop n_k\ge 1,
\ n_k{\rm -odd}}
P_{n,m}\prod_{k=1}^{m-1}\sgn(\cs_{j_k})\right\},
\label{gfn-conj}
\ee
where as in \eqref{groupindex} $j_k=n_1+\cdots + n_k$
and $P_{n,m}$ are polynomials in $\kappa(\gamma_{ij})$ homogeneous of degree $n-1$ which can be represented as
sums over subsets $\IP_{n,m}$ of flow trees with $n$ leaves
\be
P_{n,m}=\frac{1}{2^{n-1}}\sum_{T\in \IP_{n,m}}d_T\,\kappa(T)
\label{PPn}
\ee
with some numerical coefficients $d_T$.
\end{conj}

In fact, such a representation (if it exists) is highly ambiguous since there are many linear relations between polynomials $\kappa(T)$
induced by the identity
\be
\gamma_{12}\gamma_{1+2,3}+\gamma_{23}\gamma_{2+3,1}+\gamma_{31}\gamma_{1+3,2}=0,
\ee
and even more relations after multiplication by sign functions and symmetrization.
Eq. \eqref{gfn-upto5} suggests that there should exist a choice of subsets $\IP_{n,m}$ which leads to simple values of the coefficients $d_T$.
However, we do not know of a procedure which  would allow to determine it systematically
beyond $n=5$.

The main advantage of the representation \eqref{gfn-conj} is that the factors $\kappa(T)$ multiplying the products of sign functions
are proportional to $\prod_{v\in V_T}(\tbfv_v,\bfx)$ where the vectors $\tbfv_v$ defined in \eqref{deftv} are all mutually orthogonal.
This property allows  to immediately promote the kernel corresponding to the function \eqref{gfn-conj},
or to its value at the attractor point
\be
\cEf_n(\{\gama_i\})=\gf_n(\{\gama_i,\beta_{ni}\})=\Sym\left\{
\sum_{n_1+\cdots +n_m= n \atop n_k\ge 1,
\ n_k{\rm -odd}
}
P_{n,m}\prod_{k=1}^{m-1}\sgn(\Gamma_{nj_k})\right\},
\label{Efn-alt}
\ee
to a solution of Vign\'eras' equation,
without going through the complicated construction of subsection \ref{subsec-fullcompl}.
Indeed, both such kernels are linear combinations of terms, labelled by flow trees, of the type \eqref{asympt-tphi}
where the vectors $\tbfv_i$ coincide with the vectors $\tbfv_v$ for a given flow tree.
Due to their mutual orthogonality, the solutions of Vign\'eras' equation with $\lambda=n-1$ with such asymptotics
are given by
\be
\cEPhi_n(\bfx)=\frac{1}{2^{n-1}} \Sym\left\{\sum_{n_1+\cdots +n_m= n \atop n_k\ge 1,
\ n_k{\rm -odd}
}
\sum_{T\in \IP_{n,m}}d_T\, \tPhi_{m-1,n-1}^E(\{ \bfv_{j_k}\}, \{\tbfv_v\};\bfx)\right\},
\label{fullker-alt}
\ee
and similarly for $\tcEPhi_n$ with exchange of $\bfv_{j_k}$ by $\bfu_{j_k}$.
The vectors $\bfv_\ell$, $\bfu_\ell$ defined in \eqref{deftvl}
can be seen as vectors $\bfv_e$, $\bfu_e$ for the unrooted tree of trivial topology
$\bullet\!\mbox{---}\!\bullet\!\mbox{--}\cdots \mbox{--}\!\bullet\!\mbox{---}\!\bullet\,$.
Thus, instead of a sum over marked unrooted trees with factors themselves given
by another sum over unrooted trees as in \eqref{fullker},
we arrive at a representation involving only a sum over a suitable subset of flow trees and a single unrooted tree.\footnote{Of course,
the function \eqref{fullker-alt} must be equal to \eqref{fullker}, which is guaranteed by
the fact that they are solutions of Vign\'eras' equation with the same asymptotics.}
It becomes particularly simple in the case $n\le 5$ where, as follows from \eqref{gfn-upto5}, one can drop the sum over partitions,
take $m=n$ and equate all coefficients $d_T$ to $d_n$.

\subsubsection{Refinement}
\label{subsubsec-ref}

Finally, there is yet an alternative way of obtaining the functions $\gf_n$,
which also sheds light on the origin of the representation \eqref{gfn-conj}.
This representation is  inspired by the solution for the tree index presented in section \ref{subsec-partialindex}.
As in that case, the idea is to introduce a refinement parameter $y$,
 perform manipulations keeping $y\ne 1$ and take the limit $y\to 1$ in the end.

Let us therefore introduce a `refined' analogue of $\gf_n$, which we call $\gref_n$ and
which satisfies the refined version of the recursive relation \eqref{iterseed} where
the $\kappa$-factor is replaced by its $y$-dependent version \eqref{kappadef}.
This refined equation can be solved by the following ansatz (cf. \eqref{gtF})
\be
\gref_n(\{\gama_i,c_i\},y)
= \frac{(-1)^{n-1+\sum_{i<j} \gamma_{ij} }}{(y-y^{-1})^{n-1}}
\,\Sym\Bigl\{
\Fref_n(\{\gama_i,c_i\})\,y^{\sum_{i<j} \gamma_{ij}}\Bigr\}.
\label{whgF}
\ee
Indeed, it is easy to see that $\gref_n$ satisfies \eqref{iterseed} with the $y$-dependent $\kappa$-factor provided
$\Fref_n$ satisfies the recursive relation
\be
\begin{split}
&\hf\sum_{\ell=1}^{n-1}\bigl(\sgn(\cs_\ell)-\sgn(\Gamma_{n\ell})\bigr)\,\Fref_\ell(\{\gama_i,\cl_i\}_{i=1}^\ell)\,
\Fref_{n-\ell}(\{\gama_i,\cl_i\}_{i=\ell+1}^n)
\\
&\qquad\qquad
=\Fref_n(\{\gama_i,c_i\})-\Fref_n(\{\gama_i,\beta_{ni}\}).
\end{split}
\label{itereqFwl}
\ee
This equation is simpler than \eqref{iterseed} in  that it is $y$-independent and does not involve any $\kappa$-factors.
Furthermore, we already know  one solution  --- the functions $\Fwl_n$ \eqref{Fwlstar-a}
which can be shown to satisfy \eqref{itereqFwl} with help of the sign identity \eqref{signident3}.
However, this solution is not yet suitable because, starting from $n=3$,
its substitution into the r.h.s. of \eqref{whgF} does not produce a Laurent polynomial in $y$,
but rather a rational function with a pole at $y\to 1$.

The solution can be promoted to a Laurent polynomial by noting that the recursive equation \eqref{itereqFwl} determines $\Fref_n$
in terms of $\Fref_k$, $k<n$ only up to an additive constant $b_n$. Thus, at each order one can adjust this constant so that to ensure that
$\gref_n$ given by \eqref{whgF} is regular at $y=1$.
Then we arrive at the following solution
\be
\Fref_n(\{c_i\})=\frac{1}{2^{n-1}}\sum_{n_1+\cdots +n_m= n\atop n_k\ge1}
\prod_{k=1}^m  b_{n_k}\prod_{k=1}^{m-1}\sgn(\cs_{j_k}),
\label{defFref}
\ee
where as usual $j_k=n_1+\cdots + n_k$. Remarkably, this solution depends
only on the stability parameters $c_i$, despite the fact that the recursion \eqref{itereqFwl} also involves the DSZ products $\gamma_{ij}$.
It satisfies this recursion for arbitrary coefficients $b_n$, as can be  easily checked by extracting
contributions with the same sets of $\{n_k\}$ and applying the sign identity \eqref{signident3}.
Moreover, for arbitrary values of these coefficients, the numerator
of \eqref{whgF} evaluated at $y=1$ turns out\footnote{After this work was first released,  Don Zagier
communicated to us a proof of this assertion, and of the conjecture \eqref{coefbn} below,
based on the more elementary observation
that $\Sym \Fwl_n(\{c_i\})=2^{1-n}/n $ for $n$ even,
or 0 for $n$ odd, irrespective of the value of the $c_i$'s, where $\Fwl_n(\{c_i\})$
is defined in \eqref{Fwlstar-a}.}
 to be a constant, independent of the $c_i$'s.
The coefficients $b_n$ are then fixed uniquely by requiring that this constant vanishes, namely
\be
\Sym\Fref_n(\{c_i\})=0\, .
\label{eqB}
\ee
Choosing one particular configuration of the moduli, say $c_i>0$ for $i=1,\dots,n-1$
and $c_n=-\sum_{i=1}^{n-1} c_i<0$,   one finds numerically that all coefficients $b_n$ with $n$
even vanish,\footnote{This follows from the fact that under the permutation $c_i\mapsto c_{n-i+1}$
one has $\cs_i\mapsto -\cs_{n-i}$ and therefore $\Fref_n$ flips sign. The deeper reason for this is that
the terms which are products of even and odd number of signs cannot mix. In contrast, sign identities
such as \eqref{signprop-ap} can decrease the number of signs in a product by even number.
Thus, a constant can appear in $\Fref_n$ only for $n$ odd. The same fact ensures the
vanishing of the coefficients $a_\cT$ in \eqref{defcDcT} for trees with even number of vertices.}
while those with $n$ odd  are given by
\be
b_1=1,
\quad
b_3=-\frac13,
\quad
b_5=\frac{2}{15},
\quad
b_7=-\frac{17}{315},
\quad
b_9=\frac{62}{2835},
\quad
b_{11}=-\frac{1382}{155925},
\quad \dots
\ee
We observe that these coefficients coincide with the first coefficients in the Taylor series of
$\tanh x = x - \frac13 x^3+ \frac{2}{15} x^5 +\dots$. We therefore conjecture that
this identification continues to hold in general, so that $b_n$ is expressed in terms
of the Bernoulli number $B_{n+1}$ through
\be
b_{n-1}=\frac{2^{n} (2^{n} - 1)}{n!}\, B_{n} .
\label{coefbn}
\ee

Using the iterative equation \eqref{iterseed}, the constraint \eqref{eqB} and proceeding by induction, it is easy to see that
$\gref_n$ defined by \eqref{whgF} and its first $n-2$ derivatives with respect to $y$ vanish at $y=1$.
This ensures that $\gref_n$ is smooth and its limit $y\to 1$ is well-defined.

\begin{conj}
The function $\gref_n$, defined by \eqref{whgF}, \eqref{defFref} and \eqref{coefbn},
 reproduces the function $\gf_n$ \eqref{defDf-gen} in the unrefined limit,
\be
\lim_{y\to 1}\gref_n(\{\gama_i,c_i\},y)=\gf_n(\{\gama_i,c_i\}).
\label{limitgerf}
\ee
\end{conj}

We have checked this conjecture numerically up to $n=6$ evaluating
the limit $y\to 1$
using  l'H\^opital's rule. Unfortunately, the new representation of $\gf_n$ obtained in this way
is not helpful for the purposes of this work because it does not lead to any new representation for the completion.
Nevertheless, the comparison of \eqref{defFref} and \eqref{gfn-conj} shows that the sum over partitions and the simple
form of products of sign functions characterizing the representation discussed in the previous subsection find their origin in the formula for $\Fref_n$.
Furthermore, this construction strongly suggests that the refinement may be compatible with S-duality,
and that a suitably defined generating function of refined DT invariants may also possess interesting modular properties,
which can even be simpler than those of the usual DT invariants.

\section{Discussion}
\label{sec-disc}

In this paper we studied the modular properties of the generating function $h_{p,\mu}(\tau)$ of MSW invariants encoding BPS degeneracies
of D4-D2-D0 black holes in Type IIA string theory on a Calabi-Yau threefold, with fixed
D4-charge $p^a$, D2-brane charge (up to spectral flow) $\mu_a$ and  invariant D0-brane charge
$\hat q_0$ (defined in \eqref{defqhat}) conjugate to the modular parameter $\tau$.
These properties  follow from demanding
that the vector multiplet moduli space in $D=3$ (or the hypermultiplet moduli
space in the dual type IIB picture) admits an isometric action of $SL(2,\IZ)$.
Our main result is an explicit formula for the non-holomorphic
modular completion $\whh_{p,\mu}$ \eqref{exp-whh}, where the functions
$R_n(\{\gamma_i\},\tau_2)$ are
given by Eq. \eqref{solRn}, with $\cE_v=\cE_v^{(0)}+\cE_v^{(+)}$ specified
by \eqref{defEn0new} and \eqref{fullker}.
This result applies for D4-branes wrapping
a general effective divisor which may be the sum of an arbitrary number $n$ of irreducible divisors.

The existence of such completion was not guaranteed. We arrived at this result by cancelling
the modular anomaly of an indefinite theta series of signature $(n, n b_2(\CY)-n)$ with
a complicated kernel \eqref{totker}. In particular, Vign\'eras' equation encoding this anomaly involves
K\"ahler moduli in a non-trivial way through the walls of marginal stability,
whereas the functions $R_n$ in our ansatz \eqref{exp-whh}
can only depend (non-holomorphically) on the modular parameter $\tau$
and on the charges $\gamma_i$.
Note also that the transformation \eqref{SL2phi} of the contact potential,
which was the starting point of our analysis, is only a necessary condition for the existence of
an isometric action of $SL(2,\IZ)$ on $\cM_H$. To demonstrate that this condition is sufficient,
one would have to construct suitable complex Darboux coordinates on the twistor space such that,
after expressing them in terms of the completed generating function $\whh_{p,\mu}$ found in this paper,
$SL(2,\IZ)$ acts on these coordinates by a complex contact transformation. Such coordinates were constructed
at the two-instanton level in \cite{Alexandrov:2017qhn}, and it is an interesting challenge
to extend this construction to arbitrary $n$.

The structure of the completion suggests that, for a divisor decomposable into a sum of $n$ effective divisors,
the holomorphic generating function $h_{p,\mu}$ is a mixed vector valued
mock modular form of higher depth, equal to $n-1$.
Such objects have recently appeared in various mathematical and physical contexts
\cite{2017arXiv170406891B,Manschot:2017xcr,Bringmann:2018,Bringmann:2018cov,Gupta:2018krl}
and correspond to holomorphic functions whose completion is constructed from period (or Eichler)
integrals of mock modular forms of smaller depth (with depth 0 mock modular forms being
synonymous with ordinary modular forms). For instance, in the case of standard mock modular forms
of depth one, the completion is given by an Eichler integral of a modular function \cite{Zwegers-thesis}.
In our case this structure follows from the fact that the completion $\whh_{p,\mu}$ is built from the generalized error functions
which have a representation in terms of iterated period integrals \cite{Alexandrov:2016enp,2017arXiv170406891B}.
In order to make manifest the mixed mock modular nature of $h_{p,\mu}$,
one should represent the antiholomorphic derivative of
its completion $\whh_{p,\mu}$, computed in \eqref{exp-derwh},
as a sum of  products of completions of mixed mock modular forms of lower
depth and  anti-holomorphic modular forms. This is a non-trivial task which
involves the technique of lattice factorization for indefinite theta series, which we leave
for future research.

An important caveat in our construction is that
we did not establish the convergence of the various generating functions and indefinite theta series.
Technically, we constructed a theta series
$\vartheta_{\bfp,\bfmu}\bigl(\whPhi^{{\rm tot}}_{n},n-2\bigr)$ whose kernel
satisfies Vign\'eras' equation \eqref{Vigdif}, but we did not demonstrate that it decays as required by
Vign\'eras' criterium. (For $n=2$ this is easy to show \cite{Alexandrov:2016tnf}
since the kernel is a product of an exponentially suppressed factor and
a difference of two error functions defined by two positive vectors.)
In fact, convergence issues already arise when expressing
the generating function of DT invariants in terms of MSW invariants in \eqref{multih}, and are related
to the problem of convergence of the BPS black hole partition function. The latter was proven to converge in the large volume
limit for the case involving up to $n\le 3$ centers in \cite{Manschot:2010xp}.
We hope that the results for the tree index obtained in \cite{Alexandrov:2018iao} will allow to
extend this result to arbitrary $n$.

{We note that many of the complications in our construction originate from the
occurrence of polynomial factors in the kernel of theta series, which in turn can be traced
to the factors of DSZ products $\gamma_{ij}$ in the  wall-crossing formula. Presumably, most of
these complications would disappear if one could find a generalization of
the contact potential involving the refined DT invariants, so that all these
prefactors can be traded for explicit powers of the fugacity $y$, as in section \ref{subsubsec-ref}.
It is possible that twistorial topological strings \cite{Cecotti:2014wea}
may provide the appropriate framework for finding such a
refinement, at least in the context of string vacua on non-compact Calabi-Yau threefolds.}

On the physics side, we expect that our results will have important implications for the physics of BPS black holes.
Indeed, the mock modular property of the generating function $h_{p,\mu}$
will affect the growth of its Fourier coefficients \cite{Bringmann:2010sd},
and should be taken into account when performing
precision tests of the microscopic origin of the Bekenstein-Hawking entropy.
A natural question is to understand the origin of the non-holomorphic correction terms $R_n$,
in terms of the quantum mechanics of $n$-centered black holes \cite{Alexandrov:2014wca,Pioline:2015wza}.
Another outstanding question is to understand the
physical significance of the instanton generating function $\cG$ defined in
\eqref{defcF2}. Its modular properties are exactly those expected from the elliptic genus
of a superconformal field theory, except for the fact that it is not holomorphic in $\tau$.
It is natural to expect that this non-holomorphy can be traced to the existence of a
continuum of states with a non-trivial spectral asymmetry between bosons and
fermions \cite{Troost:2010ud}.
It would be very interesting to understand the origin of this continuum in terms of the
worldvolume theory of an M5-brane wrapped on a reducible divisor. Moreover, for special $\cN=2$
theories in $D=4$ obtained by circle compactification of an $\cN=1$ theory in $D=5$,
$\cG$ is closely related to the index considered in \cite{Alexandrov:2014wca}.
It would be interesting to understand the physical origin of its modular invariance
in this context.

Finally, we expect that the structure found in the context of generalized DT-invariants
on Calabi-Yau threefolds will also arise in the study of other types of
BPS invariants  where higher depth mock modular forms are expected to occur, such as
Vafa-Witten invariants and Donaldson invariants of four-folds with $b_2^+=1$ \cite{Alexandrov:2016enp,Manschot:2017xcr}.
In particular, it would be interesting to determine the modular
completion of the generating function of Vafa--Witten invariants on the complex projective plane computed
for all ranks in \cite{Manschot:2014cca},
generalizing the rank 3 case studied in \cite{Manschot:2017xcr}, and compare with the completion found in the present paper.

\acknowledgments

The authors are grateful to Sibasish Banerjee and Jan Manschot for useful discussions and
collaboration on \cite{Alexandrov:2016tnf,Alexandrov:2016enp,Alexandrov:2017qhn}
which paved the way for the present work,  to Karen Yeats for useful communication
about the combinatorics of rooted trees, and to Don Zagier for providing a proof of
the conjecture \eqref{coefbn}.  The hospitality and financial support of
the Theoretical Physics Department of CERN, where this work was initiated,
is also gratefully acknowledged. The research of BP is supported in part by French state funds
managed by the Agence Nationale de la Recherche (ANR) in the context of the LABEX ILP (ANR-11-IDEX-0004-02, ANR-10- LABX-63).

\newpage

\appendix

\addtocontents{toc}{\vspace{-0.1cm}}
\section{Proof of the Theorem}
\label{ap-theorem}

In this appendix we prove Theorem \ref{theorem} presented in the Introduction.
The crucial ingredient is provided by the following Lemma\footnote{We were informed that
this lemma appears as an exercise in \cite[p. 70]{knuthart}.}

\begin{lemma}\label{lemma-ntrees}
The number of ways of labelling the vertices of a rooted ordered tree $T$ with increasing labels is given by
\be
N_T=n_T! \prod_{v\in V_T}\frac{1}{n_v(T)}\, ,
\label{numtrees}
\ee
where $n_T$ is the number of vertices of the tree $T$.
\end{lemma}

\begin{proof}
First, let us find a recursive relation between the numbers $N_T$.
Namely, consider a vertex $v$ and the set of its children $\Ch(v)$.
Denote by $T_v$ the tree rooted at $v$ and with leaves coinciding with leaves of $T$, for which one has $n_{T_v}=n_v$.
Then it is clear that
\be
N_{T_v}=\frac{\(\sum_{v'\in\Ch(v)}n_{v'}\)! }{\prod_{v'\in\Ch(v)}n_{v'}!}\prod_{v'\in\Ch(v)}N_{T_{v'}}\, .
\label{numtrees-ver}
\ee
From this relation, one  easily shows by recursion that \eqref{numtrees} follows.
Indeed, assuming that it holds for the subtrees $T_{v'}$
and taking into account that $\sum_{v'\in\Ch(v)}n_{v'}=n_v-1$,
the r.h.s. can be rewritten as
\be
(n_v-1)!\prod_{v'\in\Ch(v)}\prod_{v''\in V_{T_{v'}}}\frac{1}{n_{v''}}
=n_v!\prod_{v'\in V_{T_v}}\frac{1}{n_{v'}},
\ee
which coincides with \eqref{numtrees} when $v$ is the root of $T$.
\end{proof}

Given this Lemma, we can now rewrite the l.h.s. of \eqref{combident} as
\be
\sum_{T'\subset T}\frac{N_{T'}}{m!}\prod_{v\in V_{T'}}n_v(T)
=\frac{n!}{N_{T}}\sum_{T'\subset T}\frac{N_{T'}}{m!}\prod_{T_r\subset T\setminus T'}\frac{N_{T_r}}{n_{T_r}!},
\label{combident2}
\ee
where the last product goes over the subtrees which complement $T'$ to the full tree $T$.
On the other hand, it is easy to check that the relation \eqref{numtrees-ver} is the special
case $m=1$ of a more general relation
\be
N_{T}=\frac{\(\sum_{T_r\subset T\setminus T'}n_{T_r}\)!}{\prod_{T_r\subset T\setminus T'}n_{T_r}!}
\sum_{T'\subset T}N_{T'}\prod_{T_r\subset T\setminus T'}N_{T_r}.
\ee
where the sum runs over subtrees $T'$ with $m$ vertices.
Taking into account that $\sum_{T_r\subset T\setminus T'}n_{T_r}=n-m$ and substituting
the resulting expression into \eqref{combident2}, one recovers the binomial coefficient as stated in
the Theorem.

\addtocontents{toc}{\vspace{-0.1cm}}
\section{D3-instanton contribution to the contact potential}
\label{ap-contact}

To evaluate $(e^\phi )_{\rm D3}$, we first replace in \eqref{phiinstmany} the full prepotential $F$ by its classical, cubic
part $\Fcl$ \eqref{Fcl} and take into account that the sum over $\gamma\in-\Gamma_+$ is complex conjugate to the sum over $\gamma\in \Gamma_+$.
This gives
\be
e^{\phi} \approx \frac{\tau_2^2}{12}\,((\Im u)^3)
+\frac{\I\tau_2}{8}\,\sum_{\gamma\in\Gamma_+}\Re
\int_{\ell_\gamma} \frac{\text{d}t}{t} \( t^{-1} Z_\gamma(u^a)-t\bZ_\gamma(\ub^a)\) H_\gamma.
\label{phi-lv}
\ee
Next, we substitute the quantum corrected mirror map \eqref{inst-mp},
change the integration variable $t$ to $z$, and take the combined limit $t^a\to\infty$, $z\to 0$.
Keeping only the leading contributions, one obtains
\bea
(e^\phi )_{\rm D3} &=&-\frac{\tau_2}{2}\sum_{\gamma \in\Gamma_+}\Re
\int_{\ell_{\gammap}}\de z
\[\hat q_0+\hf\, (q+b)^2
+2(pt^2)zz_\gamma-\frac{3}{2}\, z^2(pt^2)\]H_{\gamma}
\nn\\
&&
-\frac{1}{4}\sum_{\gamma_1,\gamma_2\in\Gamma_+}(tp_1p_2)
\(\Re\int_{\ell_{\gammap_1}}\!\!\de z_1\, H_{\gamma_1}\)\( \Re\int_{\ell_{\gammap_2}}\!\!\de z_2\, H_{\gamma_2}\).
\label{Phitwo-lv}
\eea

To further simplify this expression, we note that
\be
\begin{split}
0=&\, \frac{1}{4\pi}\sum_{\gamma \in\Gamma_+}\int_{\ell_{\gammap}}\de z\,
\p_z\( z\,H_\gamma\)
\\
=&\,\sum_{\gamma \in\Gamma_+}\int_{\ell_{\gammap}}\de z \[\frac{1}{4\pi}+\tau_2(pt^2) (zz_\gamma- z^2)
-\frac{\I}{4}\sum_{\gamma'\in\Gamma_+}\langle\gamma,\gamma'\rangle\int_{\ell_{\gamma'}} \frac{\de z'}{z-z'}\, H_{\gamma'}\]H_\gamma,
\end{split}
\label{totder}
\ee
where we used the integral equation \eqref{expcX}.
Multiplying this identity by 3/4 and adding its real part to \eqref{Phitwo-lv}, one finds
\bea
(e^\phi )_{\rm D3} &=&\frac{\tau_2}{2}\sum_{\gamma \in\Gamma_+}\Re
\int_{\ell_{\gammap}}\de z\, a_{\gamma,-\frac{3}{2}}(z)\,H_{\gamma}
\nn\\
&&
-\frac{1}{4}\sum_{\gamma_1,\gamma_2\in\Gamma_+}\[(tp_1p_2)
\(\Re\int_{\ell_{\gammap_1}}\!\!\de z_1\, H_{\gamma_1}\)\( \Re\int_{\ell_{\gammap_2}}\!\!\de z_2\, H_{\gamma_2}\)
\right.
\label{Phi-lv}
\\
&&\left.\qquad
+\frac{3}{4}\,\Re \(\int_{\ell_{\gammap_1}}\!\!\de z_1\, H_{\gamma_1}\int_{\ell_{\gammap_2}}\!\!\de z_2\, H_{\gamma_2} \,
\frac{\I\langle\gamma,\gamma'\rangle}{z-z'}\)\],
\nn
\eea
where we introduced
\be
a_{\gamma,\wh}(z)=-\(\hat q_0+\hf\,(q+b)^2
+\frac{1}{2}\, (pt^2)zz_\gamma+\frac{\wh}{4\pi\tau_2}\).
\label{defag}
\ee
The meaning of this function is actually very simple: it gives the action of the modular covariant derivative operator
$\cD_{\wh}$ \eqref{modcovD} on the classical part of the Darboux coordinate \eqref{clactinst},
\be
\cD_{\wh}\cXcl_{\gamma} =a_{\gamma,\wh}(z)\cXcl_{\gamma}.
\ee
Combining this fact with equation \eqref{expcX}, it is easy to check that the function \eqref{defcF2} satisfies,
for any value of the weight $\wh$,
\bea
\cD_{\wh}\cG&=& \sum_{\gamma\in\Gamma_+}\int_{\ell_{\gammap}} \de z\,a_{\gamma,\wh}(z)\, H_{\gamma}
\\
&&
+\frac{1}{8\tau_2}\sum_{\gamma_1,\gamma_2\in\Gamma_+}
\int_{\ell_{\gamma_1}}\de z_1\,\int_{\ell_{\gamma_2}} \de z_2
\[(tp_1p_2)+2\wh\((tp_1p_2)+\frac{\I\langle\gamma_1,\gamma_2\rangle}{z_1-z_2} \)\]H_{\gammap_1}(z_1)H_{\gammap_2}(z_2),
\nn
\\
\p_{\tc_a}\cG&=& 2\pi\I \sum_{\gamma\in\Gamma_+}p^a \int_{\ell_{\gamma}} \de z\, H_{\gamma},
\label{prop-tcFc}
\eea
which allows to rewrite \eqref{Phi-lv} exactly as in \eqref{Phitwo-tcF}.

\addtocontents{toc}{\vspace{-0.1cm}}
\section{Smoothness of the instanton generating function}
\label{ap-smooth}

In this appendix we prove the smoothness of the function $\cG$ across walls of marginal stability.
The starting point is the representation \eqref{treeFh-fl} where the potential discontinuities are hidden
in the kernel of the theta series \eqref{totker}.
Given that $\intPhi_{m}$ is represented as a sum over unrooted labelled trees \eqref{kerPhin},
whereas $\gPhi_{n_k}$ appear as sums over flow trees \eqref{ker-fl},
the kernel $\Phi^{{\rm tot}}_{n}$ can be viewed as a sum over `blooming trees'
which are unrooted trees with a flow tree (the `flower') growing from each vertex.
Then the idea is that the discontinuities due to flow trees (i.e. due to DT invariants) of a blooming tree with $m$ vertices
are cancelled by the discontinuities due to exchange of integration contours in $\intPhi_{m+1}$ corresponding
to blooming trees with $m+1$ vertices.

\lfig{Combination of trees showing cancellation of discontinuities across the wall of marginal stability
corresponding to the decay $\gamma_\ver\to \gamma_L+\gamma_R$. The parts corresponding to attractor flow trees are drown in blue.}
{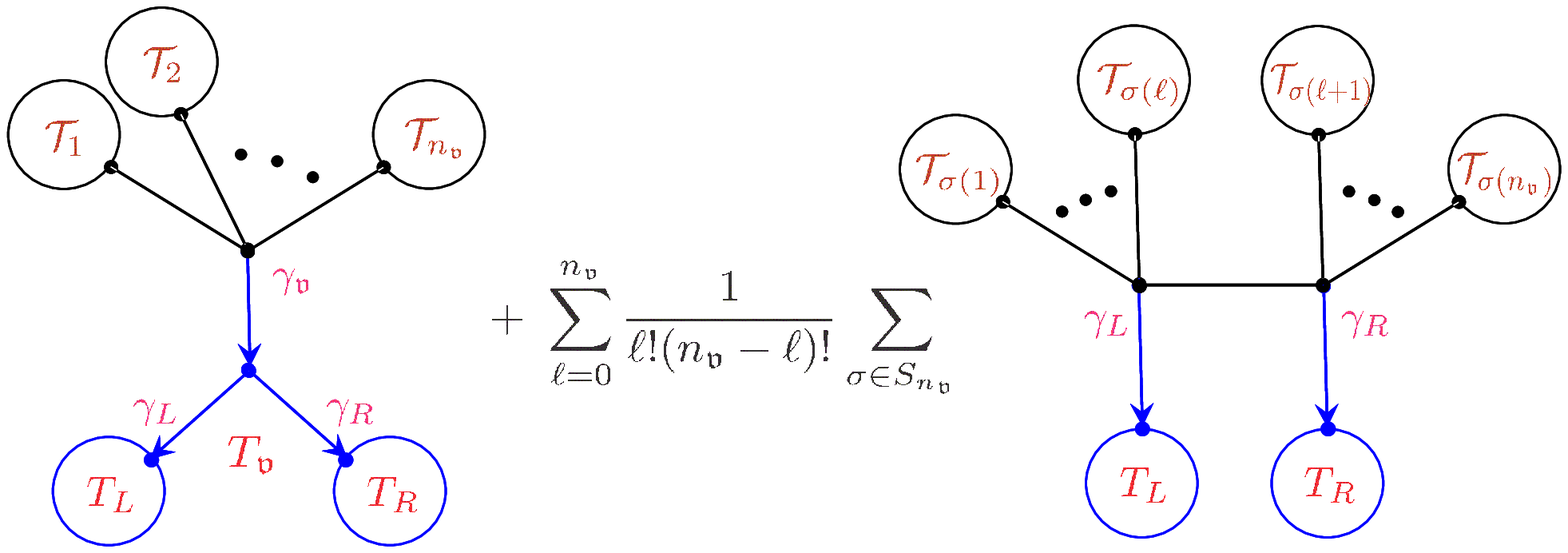}{17.5cm}{fig-smooth}{-1.3cm}

Following this idea, let us consider a tree which has an attractor flow tree $T_\ver$ growing from a vertex $\ver$ (see Fig. \ref{fig-smooth}).
We denote by $\cT_i$, $i=1,\dots,n_\ver$, the blooming subtrees connected to this vertex and by $T_{L}$
and $T_{R}$ the two parts (which may be trivial) of $T_\ver$ with the total charges $\gamma_L$ and $\gamma_R$
so that $\gamma_\ver=\gamma_L+\gamma_R$.
Together with the contribution of this tree, we consider the contributions
of the trees obtained by splitting the vertex $\ver$ into
two vertices $\ver_L$ and $\ver_R$
connected by an edge,
 carrying charges $\gamma_L$ and $\gamma_R$
and all possible allocations of the subtrees $\cT_i$ to these two vertices.
Different allocations are accounted for by the sum over permutations,
whilethe weight $\frac{1}{\ell!(n_\ver-\ell)!}$ takes into account the fact that
permutations between subtrees connected to one vertex are redundant.
The attractor flow trees $T_{L}$ and $T_{R}$
are then connected to $\ver_L$ and $\ver_R$, respectively, as shown in Fig. \ref{fig-smooth}.

The contribution corresponding to the first blooming tree has a discontinuity at the wall of marginal stability for the bound state
$\gamma_L+\gamma_R$ and originating from the factor $\Delat{z}_{\gamma_L\gamma_R}$ assigned to the root vertex of $T_\ver$.
The other contributions have discontinuities at the same wall
due to the exchange of the contours $\ell_{\gamma_L}$ and $\ell_{\gamma_R}$ for the integrals assigned to $\ver_L$ and $\ver_R$, respectively.
They are given by the residues at the pole of the integration kernel $K_{\gamma_L\gamma_R}$.
It is clear that the structure of all jumps is very similar since different subtrees produce essentially the same weights.
Let us analyze what differences may arise.
\begin{itemize}
\item
First, the contributions of the flow trees $T_{L}$ and $T_{R}$ could differ in the two cases
because they have different starting points for the attractor flows:
for the trees on the right side of Fig. \ref{fig-smooth} this is $z^a\in\cM_K$, whereas for the tree on the left
this is the point on the wall for $\gamma_L+\gamma_R$ reached by the flow from $z^a$.
But we are evaluating the discontinuity exactly on the wall where the two points coincide.
Thus, the contributions are the same.

\item
Although the subtrees $\cT_i$ give rise to the same contributions for all blooming trees shown in Fig. \ref{fig-smooth},
the contributions of the edges connecting them to either $\ver$ or $\ver_L$, $\ver_R$ are not exactly the same.
Each of them contributes the factor given by the kernel \eqref{defkerK}.
After taking the residue, the $z$-dependence of these kernels for the trees shown on the left and the right sides of the picture is identical.
However, their charge dependence is different: for the tree on the left they depend on $\gamma_\ver$, whereas for the trees on the right
they depend on $\gamma_L$ or $\gamma_R$, depending on which vertex they are connected to.
But it is easy to see that the sums over $\ell$ and permutations
produce the standard binomial expansion of a single product of kernels which all depend on $\gamma_L+\gamma_R=\gamma_\ver$
and thus coinciding with the contribution of the tree on the left.

\item
Finally, one should take into account that the discontinuity of $\Delat{z}_{\gamma_L\gamma_R}$
in the first contribution gives the factor $\langle\gamma_L,\gamma_R\rangle$.
But exactly the same factor arises as the residue of $K_{\gamma_L\gamma_R}$ corresponding to the additional edge.
\end{itemize}

Thus, it remains only to check that all numerical factors work out correctly.
Leaving aside the factors which are common to both contributions, we have
\be
-\frac{(-1)^{\langle\gamma_L,\gamma_R\rangle}}{2}\,\frac{\sigma_{\gamma_\ver}}{(2\pi)^2}\, \frac{2m}{ m!}
+(2\pi\I)(-2\pi\I)\,\frac{\sigma_{\gamma_L}\sigma_{\gamma_R}}{(2\pi)^4}\,\frac{m(m+1)}{(m+1)!}.
\label{numweights}
\ee
Here $-\frac{(-1)^{\langle\gamma_L,\gamma_R\rangle}}{2}$ comes from the factor $-\Delat{z}_{\gamma_L\gamma_R}$ in $T_\ver$,
the factors with quadratic refinement are due to functions $H_\gamma$ \eqref{prepHnew} assigned to $\ver$ or $\ver_L$, $\ver_R$,
$(2\pi\I)$ is the standard weight of the residue, $(-2\pi\I)$ is the residue of $K_{\gamma_L\gamma_R}$ \eqref{defkerK},
and factorials are the weights of the trees in the expansion \eqref{expl-tcA}.
Finally, the factors $2m$ and $m(m+1)$ arise due to the freedom to relabel charges assigned to the marked vertices:
on the left these are vertex $\ver$ and the two children of the root in $T_\ver$, whereas on the right these are $\ver_L$ and $\ver_R$.
It is immediate to check that all these numerical weights cancel, which ensures
that the function $\cG$ is continuous across walls of marginal stability.
Moreover, in this cancellation the condition that we sit on the wall was used only in locally constant factors.
Therefore, this reasoning proves not only that $\cG$ is continuous, but that it is actually smooth around these loci.

\addtocontents{toc}{\vspace{-0.1cm}}
\section{Indefinite theta series and generalized error functions}
\label{ap-indef}

\subsection{Vign\'eras' theorem}
\label{subsec-Vign}

Let $\Lat$ be a $d$-dimensional lattice equipped with a bilinear form
$(\bfx,\bfy)\equiv \bfx\cdot\bfy$, where $\bfx,\bfy\in \Lat \otimes \IR$, such that its associated quadratic form
has signature $(n,d-n)$ and is integer valued, i.e. $\bfk^2\equiv \bfk\cdot \bfk\in\IZ$ for $\bfk\in\Lat$.
Furthermore, let $\bfp\in\Lat$  be a characteristic vector
(such that $\bfk\cdot(\bfk+ \bfp)\in 2\IZ$, $\forall \,\bfk \in \Lat$),
$\bfmu\in\Lat^*/\Lat$ a glue vector, and $\lambda$ an arbitrary integer.
We consider the following family of theta series
\be
\label{Vignerasth}
\vartheta_{\bfp,\bfmu}(\Phi,\lambda;\tau, \bfb, \bfc)=\tau_2^{-\lambda/2}
\!\!\!\!
\sum_{{\bfk}\in \Lat+\bfmu+\hf\bfp}\!\!
(-1)^{\bfk\cdot\bfp}\,\Phi(\sqrt{2\tau_2}(\bfk+\bfb))\, \expe{- \tfrac{\tau}{2}\,(\bfk+\bfb)^2+\bfc\cdot (\bfk+\haf\bfb)}
\ee
defined by a kernel $\Phi(\bfx)$ such that the function $f(\bfx)\equiv \Phi(\bfx)\, e^{\tfrac{\pi}{2}\,\bfx^2}
\in L^1(\Lat\otimes\IR)$ so that the sum is absolutely convergent.
Irrespective of the choice of this kernel and of the parameter $\lambda$, any such theta series satisfies the
following elliptic properties
\be
\begin{split}
\label{Vigell}
\vartheta_{\bfp,\bfmu}\left({\Phi} ,\lambda; \tau, \bfb+\bfk,\bfc\right) =&(-1)^{\bfk\cdot\bfp}\,
\expe{-\haf\, \bfc\cdot \bfk} \vartheta_{\bfp,\bfmu}\left({\Phi} ,\lambda; \tau, \bfb,\bfc\right),
\\
\vphantom{A^A \over A_A}
\vartheta_{\bfp,\bfmu}\left({\Phi}, \lambda; \tau, \bfb,\bfc+\bfk \right)=&(-1)^{\bfk\cdot\bfp}\,
\expe{\haf\, \bfb\cdot \bfk} \vartheta_{\bfp,\bfmu}\left({\Phi} ,\lambda; \tau, \bfb,\bfc\right).
\end{split}
\ee

Now let us require that in addition the kernel satisfies the following two conditions:
\begin{enumerate}
\item
Let $D(\bfx)$ be any differential operator of order $\leq 2$, and
$R(\bfx)$ any polynomial of degree $\leq 2$. Then $f(\bfx)$ defined above must
be such that $f(\bfx)$, $D(\bfx)f(\bfx)$ and $R(\bfx)f(\bfx)\in
L^2(\Lat\otimes\IR)\bigcap L^1(\Lat\otimes\IR)$.
\item
$\Phi(\bfx)$ must satisfy
\be
\label{Vigdif}
V_\lambda\cdot \Phi(\bfx)=0,
\qquad
V_\lambda= \partial_{\bfx}^2   + 2\pi \( \bfx\cdot \pa_{\bfx}  - \lambda\)  .
\ee
\end{enumerate}
Then in \cite{Vigneras:1977} it was proven that
the theta series \eqref{Vignerasth} transforms as a vector-valued modular form of
weight $(\lambda+d/2,0)$ (see Theorem 2.1 in \cite{Alexandrov:2016enp} for the
detailed transformation under $\tau\to-1/\tau$). We refer to $V_\lambda$ as Vign\'eras' operator.
The simplest example is the Siegel theta series for which the kernel is $\Phi(\bfx)=e^{-\pi \bfx_+^2}$ where $\bfx_+$
is the projection of $\bfx$ on a fixed positive plane of dimension $n$.
This kernel is annihilated by $V_{-n}$.

In this paper we apply the Vign\'eras' theorem to the case of $\Lat=\oplus_{i=1}^n \Lambda_i$.
Thus, the charges appearing in the description of the theta series \eqref{Vignerasth} are
of the type $\bfk=(k_1^a,\dots,k_n^a)$, whereas the vectors $\bfb$ and $\bfc$ are taken with $i$-independent components, namely,
$\bfb_i^a=b^a$, $\bfc_i^a=c^a$ for $i=1,\dots, n$.
The lattices $\Lambda_i$ carry the bilinear forms $\kappa_{i,ab}=\kappa_{abc}p_i^c$ which are all of signature $(1,b_2-1)$.
This induces a natural bilinear form on $\Lat$:
\be
\bfx\cdot\bfy=\sum_{i=1}^n (p_ix_iy_i).
\label{biform}
\ee

Note also that the sign factor $(-1)^{\bfk\cdot\bfp}$ in \eqref{Vignerasth}
can be identified with the quadratic refinement provided we choose the latter as
\be
\sigma_\gamma=\sigma_{p,q}\equiv \expe{\hf\, p^a q_a}\sigma_p,
\qquad
\sigma_p=\expe{\hf\, A_{ab}p^ap^b}.
\label{qf}
\ee
The matrix $A_{ab}$, satisfying
\be
\label{ALS}
A_{ab} p^p - \frac12\, \kappa_{abc} p^b p^c\in \IZ \quad \text{for}\ \forall p^a\in\IZ\, ,
\ee
appears due to the non-trivial quantization of charges on the type IIB side \eqref{fractionalshiftsD5}
and can be used to perform a symplectic rotation to identify them
with mirror dual integer charges on the type IIA side \cite{Alexandrov:2010ca}.
It is easy to check that the quadratic refinement \eqref{qf} satisfies \eqref{defqf}.

\subsection{Generalized error functions}
\label{ap-generror}

An important class of solutions of Vign\'eras' equation is given by the error function and its generalizations constructed
in \cite{Alexandrov:2016enp} and further elaborated in \cite{Nazaroglu:2016lmr}.
Let us take
\bea
\label{defM1int}
M_1(u)&=&-\sgn(u)\, \Erfc(|u|\sqrt{\pi}) = \frac{\I}{\pi} \int_{\ell}\frac{\de z}{z}\,
e^{-\pi z^2 -2\pi \I z u},
\\
E_1(u)&=&\sgn(u)+M_1(u)
\nn\\
&=& \Erf(u\sqrt{\pi})= \int_{\IR} \de u' \, e^{-\pi(u-u')^2} \sign(u'),
\label{defE1int}
\eea
where the contour $\ell=\IR-\I u$ runs parallel to the real axis through the saddle point at $z=-\I u$.
Then, given a vector with a positive norm $\bfv^2>0$ so that $|\bfv|=\sqrt{\bfv^2}$, we define
\be
\Phi_1^E(\bfv;\bfx)=E_1\(\frac{\bfv\cdot\bfx}{|\bfv|}\),
\qquad
\Phi_1^M(\bfv;\bfx)=M_1\(\frac{\bfv\cdot\bfx}{|\bfv|}\).
\ee
It is easy to check that the first function is a smooth solution of \eqref{Vigdif} with $\lambda=0$,
whereas the second is exponentially suppressed at large $\bfx$ and also solves the same equation,
but only away from the locus $\bfv\cdot\bfx=0$ where it has a discontinuity.

Generalizing the integral representations \eqref{defM1int} and \eqref{defE1int}, we define
the generalized (complementary) error functions
\bea
M_n(\cM;\vu)&=&\(\frac{\I}{\pi}\)^n |\det\cM|^{-1} \int_{\IR^n-\I \vu}\de^n z\,
\frac{e^{-\pi \vz^{\rm tr} \vz -2\pi \I \vz^{\rm tr} \vu}}{\prod(\cM^{-1}\vz)},
\label{generr-M}
\\
E_n(\cM;\vu)&=& \int_{\IR^n} \de \vu' \, e^{-\pi(\vu-\vu')^{\rm tr}(\vu-\vu')} \sign(\cM^{\rm tr} \vu'),
\label{generr-E}
\eea
where $\vz=(z_1,\dots,z_n)$ and $\vu=(u_1,\dots,u_n)$ are $n$-dimensional vectors, $\cM$ is $n\times n$ matrix of parameters,
and we used the shorthand notations $\prod \vz=\prod_{i=1}^n z_i$ and $\sign(\vu)=\prod_{i=1}^n \sign(u_i)$.
The detailed properties of these functions can be found in \cite{Nazaroglu:2016lmr}.
Here we mention only a few:
\begin{itemize}
\item
$M_n$ are exponentially suppressed for large $\vu$ as $M_n\sim \frac{(-1)^n}{\pi^n}|\det\cM|^{-1}\,
\frac{e^{-\pi \vu^{\rm tr} \vu}}{\prod(\cM^{-1}\vu)}$,
whereas $E_n$ are locally constant for large $\vu$ as $E_n\sim \sign(\cM^{\rm tr} \vu)$.
\item
More generally, $E_n$ can be expressed as a linear combination of $M_k$, $k=0,\dots,n$, multiplied by $n-k$ sign functions,
generalizing the first relation in \eqref{defE1int} (see Eq. \eqref{expPhiE} below for a precise statement).
\item
From \eqref{generr-E} it follows that every identity between products of sign functions
implies an identity between generalized error functions $E_n$.
Moreover, expanding the $E_n$ functions in terms of $M_k$'s and sign functions,
one obtains similar identities for functions $M_n$.
For instance, the identity
\be
(\sgn(x_1)+\sgn(x_2))\,\sgn(x_1+x_2)=1+\sgn(x_1)\,\sgn(x_2)
\label{signprop-ap}
\ee
implies
\be
\begin{split}
E_2((\vv_1,\vv_1+\vv_2);\vu)+E_2((\vv_2,\vv_1+\vv_2);\vu)=&\, 1+E_2((\vv_1,\vv_2);\vu),
\\
M_2((\vv_1,\vv_1+\vv_2);\vu)+M_2((\vv_2,\vv_1+\vv_2);\vu)=&\, M_2((\vv_1,\vv_2);\vu),
\end{split}
\label{genEprop-ap}
\ee
where $\vv_1,\vv_2$ are two-dimensional vectors used to encode the $2\times 2$ matrix of parameters.
\end{itemize}

The main reason to introduce these functions is that, similarly to the usual error
and complementary error functions, they can be used to produce solutions of
Vign\'eras' equation on $\IR^{n,d-n}$.
To write them down, let us consider $d\times n$ matrix $\cV$ which can be viewed as a collection of $n$ vectors,
$\cV=(\bfv_1,\dots,\bfv_n)$. We assume that these vectors span a positive definite subspace,
i.e. $\cV^{\rm tr}\cdot\cV$ is a positive definite matrix.
Let $\cB$ be $n\times d$ matrix whose rows define an orthonormal basis for this subspace.
Then we define the {\it boosted} generalized error functions
\be
\Phi_n^M(\cV;\bfx)=M_n(\cB\cdot \cV;\cB\cdot \bfx),
\qquad
\Phi_n^E(\cV;\bfx)=E_n(\cB\cdot \cV;\cB\cdot \bfx).
\label{generrPhiME}
\ee
It can be shown that both these functions satisfy Vign\'eras' equation
(for $\Phi_n^M$ one should stay away from its discontinuities, i.e. loci where $\sign((\cB\cdot \cV)^{-1}\cB\cdot\bfx)=0$).
Moreover, they are symmetric under permutation of the vectors $\bfv_i$.
Since at large $\bfx$ one has $\Phi_n^E\sim \sign(\cV^{\rm tr}\cdot\bfx)=\prod_{i=1}^n \sign(\bfv_i\cdot\bfx)$,
one can think about this function as providing the modular completion for (indefinite) theta series with kernel given by a product of signs.

The relation between functions $E_n$ and $M_n$ mentioned above implies a similar relation between the functions \eqref{generrPhiME}.
For generic $n$, it takes the following form\footnote{This relation reduces
to Eq. (3.53) and (6.18) in \cite{Alexandrov:2016enp} for $n=2,3$, and follows  from
Eq. (63) in \cite{Nazaroglu:2016lmr} for any $n$.}
\be
\Phi_n^E(\cV;\bfx)=\sum_{\cI\subseteq \Zv_n}\Phi_{|\cI|}^M(\{\bfv_i\}_{i\in\cI};\bfx)
\prod_{j\in \Zv_n\setminus \cI}\sign(\bfv_{j\perp \cI},\bfx),
\label{expPhiE}
\ee
where the sum goes over all possible subsets (including the empty set) of the set $\Zv_{n}=\{1,\dots,n\}$, $|\cI|$ is the cardinality of $\cI$,
and $\bfv_{j\perp \cI}$ denotes the projection of $\bfv_j$ orthogonal to the subspace spanned by $\{\bfv_i\}_{i\in\cI}$.
The cardinality $|\cI|$ can also be interpreted as the number of directions in $\IR^d$
along which the corresponding contribution has an exponential fall off.

Finally, note that a solution of Vign\'eras' equation  can be uplifted to a solution
of the same equation with $\lambda$ shifted to $\lambda+1$
by acting with the differential operator
\be
\cD(\bfv)=\bfv\cdot\(\bfx+\frac{1}{2\pi}\,\p_\bfx\),
\label{defcDif}
\ee
which realizes the action of the covariant derivative raising the holomorphic weight by 1.
In particular, we can construct solutions with $\lambda=m$ which behave for large $\bfx$ as products of $n$ sign functions.
To this end, it is enough to act on $\Phi_n^E$ by this operator $m$ times. Thus, we define
the {\it uplifted boosted} error function
\be
\tPhi_{n,m}^E(\cV,\tcV;\bfx)=\[\prod_{i=1}^m \cD(\tbfv_i)\]\Phi_n^E(\cV;\bfx),
\label{deftPhigen}
\ee
where $\tcV=(\tbfv_1,\dots,\tbfv_m)$ encodes the vectors contracted with the covariant derivatives.
Since the operators $\cD(\tbfv_i)$ commute, \eqref{deftPhigen} is invariant under independent
permutations of the vectors $\bfv_i$ and $\tbfv_i$.
In the case where all $\tbfv_i$ are mutually orthogonal,
one finds the following asymptotics at large $\bfx$
\be
\lim_{\bfx\to\infty}\tPhi_{n,m}^E(\cV,\tcV;\bfx)= \prod_{i=1}^m (\tbfv_i,\bfx)\prod_{j=1}^n  \sign(\bfv_j,\bfx).
\label{asympt-tphi}
\ee
Note that the derivative $\p_\bfx$ in \eqref{deftPhigen} does not act on sign functions since $\Phi_n^E$ is smooth
and all discontinuities due to signs are guaranteed to cancel.
Similarly to \eqref{deftPhigen}, we can also define $\tPhi_{n,m}^M(\cV,\tcV;\bfx)$ where the action of derivatives on the discontinuities
of $\Phi_n^M$ is ignored as well. In the particular case $m=n$, we will omit the second label and simply write $\tPhi_{n}^E$ or $\tPhi_{n}^M$.

\addtocontents{toc}{\vspace{-0.1cm}}
\section{Twistorial integrals and generalized error functions}
\label{ap-twistint}

In this appendix we evaluate the kernels $\intPhi_n$ \eqref{kerPhin} and show that they can be expressed through
the generalized error functions introduced in appendix \ref{ap-generror}.
To this end, let us note the following identity
\be
\I\,\cD(\bfv_{ij})\,
\frac{W_{p_i}(x_i,z_i)\,W_{p_j}(x_j,z_j)}{z_i-z_j}
= \hK_{ij}\, W_{p_i}(x_i,z_i)\,W_{p_j}(x_j,z_j).
\ee
By virtue of this relation, the kernel can be represented as
\be
\intPhi_n(\bfx)=
\frac{1}{n!}\sum_{\cT\in\, \IT_n^\ell}\[\prod_{e\in E_\cT}
\cD(\bfv_{s(e) t(e)}) \]
\Phi_\cT(\bfx),
\label{PhiPhiT}
\ee
where
\be
\Phi_\cT(\bfx) =\frac{\I^{n-1}}{(2\pi)^n}\[\prod_{i=1}^n\int_{\ell_{\gamma_i}}\de z_i \, W_{p_i}(x_i,z_i)\]
\frac{1}{\prod_{e\in E_\cT}\(z_{s(e)}-z_{t(e)}\)}\, .
\label{PhiT}
\ee

One may think that the representation \eqref{PhiPhiT} misses
contributions from the covariant derivatives acting on each other.
However, such contributions are proportional to the scalar products of two vectors $\bfv_{s(e) t(e)}$ and
are non-vanishing provided the two edges have a common vertex.
If this is the case, the two edges generate the following factor
\be
\frac{(p_{\ver_1}p_{\ver_2}p_{\ver_{3}})}{(z_{\ver_1}-z_{\ver_3})(z_{\ver_2}-z_{\ver_{3}})}\,  ,
\label{factorder}
\ee
where $\ver_1,\ver_2,\ver_3$ are the tree vertices joint by the edges and $\ver_{3}=e_1\cap e_2$.
The crucial observation is that if we pick up 3 subtrees, each with a marked vertex,
there are exactly 3 ways to form a labelled tree out of them by joining the marked vertices as shown in Fig. \ref{fig-Vign3} on page \pageref{fig-Vign3}.
Each subtree contributes the same factor in all 3 cases, whereas the joining edges and the sum over trees
give rise to the vanishing factor
\be
\frac{(p_1p_2p_3)}{(z_1-z_3)(z_2-z_3)}+\frac{(p_1p_2p_3)}{(z_1-z_2)(z_1-z_3)}-\frac{(p_1p_2p_3)}{(z_2-z_3)(z_1-z_2)}=0.
\label{case3}
\ee
This ensures that no additional contributions arise and thereby proves \eqref{PhiPhiT}.
Note that for this proof it was crucial that inequivalent labelled trees enter the sum with the same weight.


Then let us do the change of variables
\be
z_i=z'-\frac{\I(pxt)}{\sqrt{2\tau_2}(pt^2)}+\sum_{\alpha=1}^{n-1} e_i^\alpha z'_\alpha,
\label{chvar-e}
\ee
where $x^a=\kappa^{ab}\sum\kappa_{i,bc} x^c_i$ (cf. \eqref{defprimevar}), and
$e_i^\alpha$ are such that
\be
\sum_{i=1}^n \ptt_i e_i^\alpha=0,
\qquad
\sum_{i=1}^n \ptt_i e_i^\alpha e_i^\beta=0, \quad \alpha\ne \beta
\label{defchangebas}
\ee
and we introduced the convenient notation $\ptt_i=(p_it^2)$.
Labeling the $n-1$ edges of the tree by the same index $\alpha$,
one can rewrite the function \eqref{PhiT} in the new variables as
\be
\Phi_\cT(\bfx) =\frac{\I^{n-1} \sqrt{\Delta}}{(2\pi)^n}\,
e^{-\frac{\pi(pxt)^2}{(pt^2)}}\int \de z' \, e^{-2\pi\tau_2(pt^2)z'^2}\prod_{\alpha=1}^{n-1}\int
\frac{\de z'_\alpha \,e^{-2\pi\tau_2\, \Delta_\alpha (z'_\alpha)^2 -2\pi\I\sqrt{2\tau_2}\, w^\alpha z'_\alpha}}
{\sum_{\beta=1}^{n-1}\(e^\beta_{s(\alpha)}-e^\beta_{t(\alpha)}\)z'_\beta}\, ,
\ee
where
\be
w^\alpha=\sum_{i=1}^n (p_ix_i t) e_i^\alpha,
\qquad
\Delta_\alpha=\sum_{i=1}^n \ptt_i (e_i^\alpha)^2,
\qquad
\Delta=\frac{\ptt\prod_{\alpha=1}^{n-1} \Delta_\alpha}{\prod_{i=1}^{n}\ptt_i}\, .
\label{defDeltaa}
\ee
The integral over $z'$ is Gaussian and is easily evaluated.
In the remaining integrals, rescaling the integration variables by $\sqrt{2\tau_2\Delta_\alpha}$,
one recognizes the generalized error functions \eqref{generr-M}.
Thus, one obtains
\be
\Phi_\cT(\bfx) =\frac{1}{2^{n-1}}\,\frac{\sqrt{\Delta}|\det\cM|}{\prod_{\alpha=1}^{n-1} \Delta_\alpha}\,
\intPhi_1(x)\, M_{n-1}\(\cM;\left\{\frac{w^\alpha}{\sqrt{\Delta_\alpha}}\right\} \).
\label{resPhiT}
\ee
where we used the function $\intPhi_1$ \eqref{defPhi1-main}
and introduced the matrix $\cM$ such that
\be
\label{Mabinv}
\cM^{-1}_{\alpha\beta}=(\Delta_\alpha\Delta_\beta)^{-1/2}\(e^\beta_{s(\alpha)}-e^\beta_{t(\alpha)}\).
\ee

A simple solution to the conditions \eqref{defchangebas} may be constructed as follows.
Let $T$ be a rooted ordered binary tree with $n$ leaves decorated by $\gamma_i$, $i=1\dots n$. As usual for such trees,
other vertices $v$ carry charges given by the sum of
charges of their children, i.e. $\gamma_v=\sum_{i\in \cI_v}\gamma_i$ where $\cI_v$ is the set of leaves which are descendants of $v$.
There are $n-1$ such vertices which we label by index $\alpha$. Then we can choose
\be
\label{eaichoose}
e^\alpha_i=\sum_{j\in\cI_{\Lv{v_\alpha}}}\sum_{k\in\cI_{\Rv{v_\alpha}}} \(\delta_{ij}\,\ptt_k-\delta_{ik}\,\ptt_j\),
\ee
which satisfy \eqref{defchangebas} as can be easily checked.
For this choice (cf. \eqref{deftv})
\be
\begin{split}
w^\alpha=&\, (\tbfu_{\alpha},\bfx),
\quad \mbox{where}\quad
\tbfu_\alpha=\sum_{i\in\cI_{\Lv{v_\alpha}}}\sum_{j\in\cI_{\Rv{v_\alpha}}}\bfu_{ij},
\\
\Delta_\alpha=&\,\tbfu_\alpha^2=\ptt_{v_\alpha} \ptt_{\Lv{v_\alpha}}\ptt_{\Rv{v_\alpha}}.
\end{split}
\label{dataechoice}
\ee
Note that the vectors $\tbfu_\alpha$ are mutually orthogonal.

In principle, any rooted binary tree $T$ is suitable for the above construction.
However, given the unrooted tree $\cT$, there is a simple (but non-unique) choice of $T$ which
simplifies  the resulting matrix $\cM$. To define it, let us construct a partially increasing family
of subtrees of $\cT$, such that two members in this family are either disjoint,
or contained in one another. Moreover, we require that the largest subtree is $\cT$ itself,
while each subtree containing more than one vertex is obtained by joining
two smaller subtrees along an edge of $\cT$.
Any such family contains $2n-1$ subtrees $\cT_{\hat\alpha}$ labelled by $\hat\alpha=1,\dots,2n-1$.
Among them, $n-1$ subtrees, which we label by $\alpha=\hat\alpha=1,\dots,n-1$, contain several vertices,
while the remaining  $n$ subtrees with label $\hat\alpha=n,\dots,2n-1$ have only one vertex.
For each subtree $T_\alpha$, we denote by $e_\alpha$ the edge of $\cT$
which is used to reconstruct $T_\alpha$ from two smaller subtrees $\cT_{\alpha_L}, \cT_{\alpha_R}$.
From this data, we construct a rooted binary tree $T$ with $n-1$ vertices in one-to-one correspondence with the
subtrees $\cT_\alpha$ and $n$ leaves in one-to-one correspondence with
the one-vertex subtrees $\cT_{\hat\alpha}$ with $\hat\alpha\ge n$.
In this correspondence, the two children of a vertex associated to $\cT_{\alpha}$ are the vertices associated
to the two subtrees $\cT_{\alpha_L}, \cT_{\alpha_R}$.
The ordering at each vertex is defined to be such that the subtree containing the source/target vertex of the corresponding edge $e_\alpha$
is on the left/right.\footnote{The orientation of the edges of $\cT$ is fixed already in \eqref{PhiPhiT},
but the full kernel $\intPhi_n$ does not depend on its choice.}
Of course, this construction is not unique since there are many ways to decompose $\cT$ into such a set of subtrees (see Fig. \ref{fig-basis}).

\lfig{An example of an unrooted labelled tree with 4 vertices and two choices of decompositions into subtrees with
the corresponding rooted binary trees. The edges  are labelled $e_1,e_2,e_3$ from left to right
along $\cT$.}
{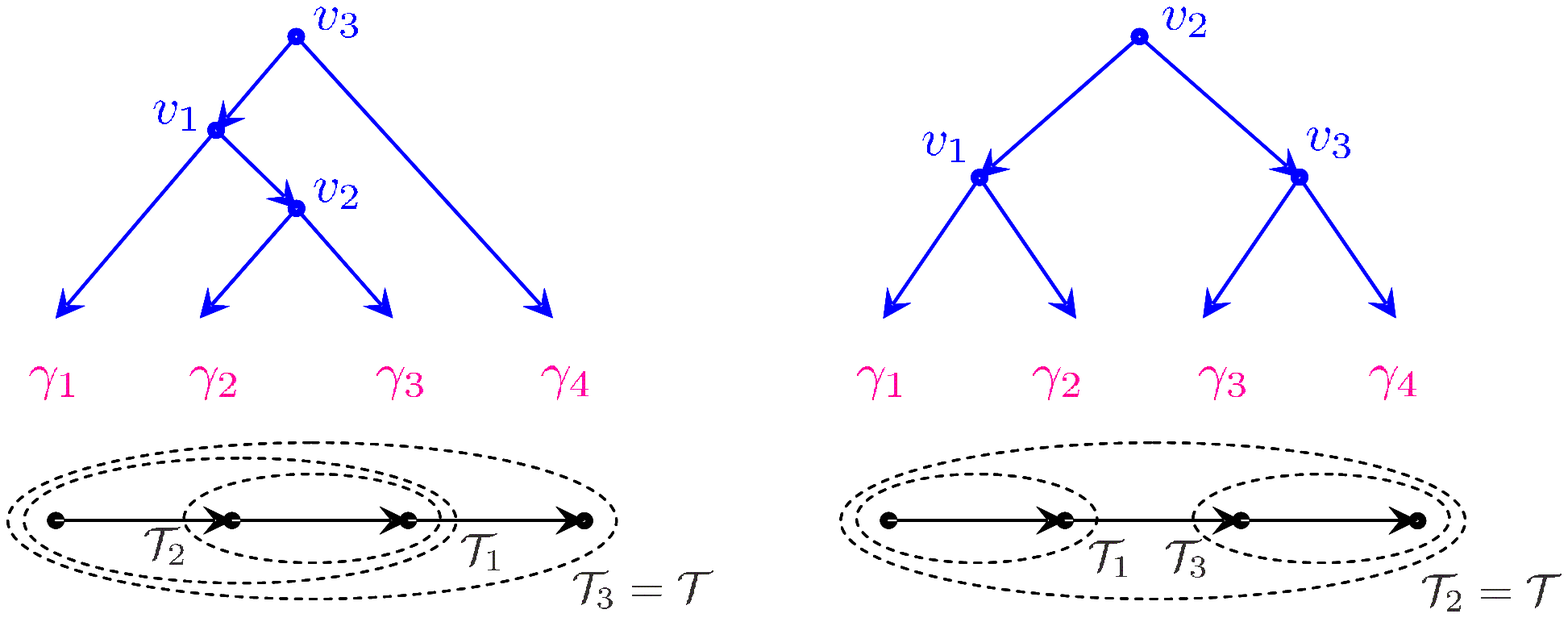}{15.cm}{fig-basis}{-0.6cm}

Applying the above construction to this particular choice of rooted tree, one finds
\be
\sqrt{\Delta_\alpha\Delta_\beta}\,\cM^{-1}_{\alpha\beta}=\left\{
\begin{array}{ll}
\ptt_{v_\alpha}, \quad & \alpha=\beta,
\\
\eps_{\alpha\beta} \ptt_{\Lv{v_\beta}}, \quad
& e_\alpha\cap \cT_{\beta_R} \ne \emptyset,\ e_\alpha\nsubseteq \cT_\beta,
\\
\eps_{\alpha\beta} \ptt_{\Rv{v_\beta}}, \quad
& e_\alpha\cap \cT_{\beta_L} \ne \emptyset,\ e_\alpha\nsubseteq \cT_\beta,
\\
0, & e_\alpha\cap \cT_\beta = \emptyset \ \mbox{ or }\ e_\alpha\subset \cT_\beta,
\end{array}
\right.
\label{invmatM}
\ee
where $\eps_{\alpha\beta}=-1$ if the orientations of $e_\alpha$ and $e_\beta$ on the path joining them are the same and $+1$ otherwise.
This result shows that the matrix $\cM^{-1}$ turns out to be triangular which makes it much simpler to find its inverse.
On the basis of \eqref{invmatM}, below we will prove the following

\begin{lemma}
One has $\cB\cdot \cV=\cM$ and $\cB\cdot \bfx=\left\{\frac{w^\alpha}{\sqrt{\Delta_\alpha}}\right\}$ provided
\be
\begin{split}
\cB=&\, \(\frac{\tbfu_1}{\sqrt{\Delta_1}},\dots, \frac{\tbfu_{n-1}}{\sqrt{\Delta_{n-1}}}\)^{\rm tr},
\\
\cV=&\, \ptt^{-1}\Bigl(\sqrt{\Delta_1}\bfu_1,\dots, \sqrt{\Delta_{n-1}}\bfu_{n-1}\Bigr),
\end{split}
\label{matrix-gener}
\ee
where the vectors $\bfu_\alpha$ are defined in \eqref{defue}.
Moreover, the vectors $\tbfu_\alpha$ form an orthogonal basis in the subspace
spanned by $\bfu_\alpha$.
\end{lemma}

This lemma allows to reexpress the kernel $\Phi_\cT$ \eqref{resPhiT}
in terms of the boosted generalized error function $\Phi^M_{n-1}$ \eqref{generrPhiME}.
It is important that its argument $\cV$ does not depend on the choice of
the binary rooted tree $T$, but only on the unrooted tree $\cT$.
In addition, the function actually does not
depend on the normalization of the vectors composing $\cV$.
Given also that the determinant of $\cM$ is found to be
\be
|\det\cM|=
\prod_{\alpha=1}^{n-1}\frac{ \Delta_\alpha}{\ptt_{v_\alpha}}
=\prod_{i=1}^{n}\ptt_i\prod_{v\in V_{T}\setminus{\{v_0\}}}\ptt_v
=\sqrt{\ptt^{-1}\,\prod_{i=1}^{n}\ptt_i\,\prod_{\alpha=1}^{n-1} \Delta_\alpha },
\ee
so that the prefactor in \eqref{resPhiT} cancels, one arrives at
\be
\Phi_\cT(\bfx) =\frac{1}{2^{n-1}}\,
\intPhi_1(x)\, \Phi^M_{n-1}(\{ \bfu_e\};\bfx).
\label{resPhiTend}
\ee
Finally, since the differential operator in \eqref{PhiPhiT} commutes with functions of $x$ due to the orthogonality of
$\bfv_{ij}$ and $\bft$, one can write the kernel $\intPhi_n$ as
\be
\intPhi_n(\bfx)=
\frac{\intPhi_1(x)}{2^{n-1} n!}\sum_{\cT\in\, \IT_n^\ell}\[\prod_{e\in E_\cT}\cD( \bfv_{s(e) t(e)})\]
\Phi^M_{n-1}(\{ \bfu_\alpha\};\bfx).
\label{Phin-prefinal}
\ee
which is the same as \eqref{Phin-final}.

\subsection*{Proof of Lemma 2}

We start by proving that the vectors $\tbfu_\alpha$ form an orthogonal basis in the subspace
spanned by $\bfu_\alpha$. Since the orthogonality is ensured by construction based on a rooted binary tree,
it remains to show that any vector $\bfu_\alpha$ can be decomposed as a linear combination of $\tbfu_\alpha$.
To this end, we show that the determinant of the Gram matrix constructed from the set of vectors
$\{\tbfu_\alpha\}_{\alpha=1}^{n-1}\cup\{\bfu_\beta\}$ vanishes.
This requires to calculate the scalar product $(\tbfu_\alpha, \bfu_\beta)$
which can be done using
\be
(\bfu_{ij},\bfu_{kl})=\left\{\begin{array}{ll}
\ptt_i\ptt_j\ptt_{i+j},\qquad &  i=k,\ j=l,
\\
\ptt_i\ptt_j\ptt_{l},\qquad &  i=k,\ j\ne l,
\\
0,\qquad &  i,j\ne k,l.
\end{array}\right.
\ee
Summing $i,j,k,l$ over appropriate subsets, it is immediate to see that $(\tbfu_\alpha, \bfu_\beta)=0$
if $\cT_{\alpha}\subset \cT_\beta^s$ or $\cT_\beta^t$,
which are the two trees obtained by
dividing the tree $\cT$ into two parts by cutting the edge $e_\alpha$.
In other words, it is non-vanishing only if $e_\beta\subseteq \cT_\alpha$.
Then there are two cases which give
\be
(\tbfu_\alpha, \bfu_\beta)=\left\{
\begin{array}{ll}
\ptt_\alpha^L\,\ptt_\alpha^R\,\ptt, \quad &
\alpha=\beta,
\\
\ptt_{\alpha\beta}^{Ls}\,\ptt_\alpha^R\,\ptt_\beta^t-\ptt_{\alpha\beta}^{Lt}\,\ptt_\alpha^R\,\ptt_\beta^s
+\ptt_{\alpha}^{L}\,\ptt_\alpha^R\,\ptt_\beta^s=\ptt_{\alpha\beta}^{Ls}\,\ptt_\alpha^R\,\ptt,\quad & e_\beta\subset \cT_\alpha,
\end{array}\right.
\label{res-scpr}
\ee
where we introduced
\be
\ptt_\alpha^L=\sum_{i\in\cI_{\Lv{v_\alpha}}}\ptt_i,
\qquad
\ptt_{\beta}^{s}=\sum_{i\in V_{\cT_\beta^s}}\ptt_i,
\qquad
\ptt_{\alpha\beta}^{Ls}=\sum_{i\in\cI_{\Lv{v_\alpha}}\cap V_{\cT_\beta^s}}\ptt_i,
\qquad
\ptt_{\alpha\beta}^{st}=\sum_{i\in V_{\cT_\alpha^s}\cap V_{\cT_\beta^t}}\ptt_i,
\ee
and similarly for variables with labels $R$ and $t$. In \eqref{res-scpr} in the second case we assumed that the orientation of edges is
such that $e_\beta\subset \cT_\alpha^s$ and $e_\alpha \subset \cT_\beta^t$.
If this is not the case, one should replace $s$ by $t$, $L$ by $R$ and flip the sign for each change of orientation.
Below we use the same assumption, but the computation can easily be generalized to a more general situation.

Given the result \eqref{res-scpr}, $(\tbfu_\alpha, \tbfu_\beta)=\Delta_\alpha\delta_{\alpha\beta}$ and
$(\bfu_\beta, \bfu_\beta)=\ptt_\beta^s\,\ptt_\beta^t\,\ptt$,
the determinant of the Gram matrix is easily found to be
\be
\det {\rm Gram}(\tbfu_1,\cdots, \tbfu_{n-1},\bfu_\beta)=
\ptt\prod_{\alpha=1}^{n-1}\Delta_\alpha
\[\ptt_\beta^s\,\ptt_\beta^t-\ptt \sum_{T_\alpha\supseteq e_\beta}\frac{(\ptt_{\alpha\beta}^{Ls})^2\,\ptt_\alpha^R}{\ptt_{v_\alpha}\ptt_\alpha^L } \].
\label{detGram}
\ee
Note that the subtrees $\cT_\alpha$ containing the edge $e_\beta$ form an ordered set so that the sum in the square brackets goes over
$\alpha_\ell$, $\ell=1,\dots, m$, such that $\cT_{\alpha_\ell}\subset \cT_{\alpha_{\ell+1}}$.
The first element of this set $\alpha_1=\beta$, whereas the last corresponds to the total tree, $\cT_{\alpha_m}=\cT$.
Due to $\ptt_{\alpha_m}^L=\ptt_{\alpha_m}^s$, $\ptt_{\alpha_m}^R=\ptt_{\alpha_m}^t$ and $\ptt_{\alpha_m\beta}^{Ls}=\ptt_\beta^s$,
the first term in the square brackets together with the term in the sum corresponding to $\alpha_m$ gives
\be
\ptt_\beta^s\(\ptt_\beta^t -\frac{\ptt_\beta^s\,\ptt_{\alpha_m}^t}{\ptt_{\alpha_m}^s}\)
=\frac{\ptt\,\ptt_\beta^s\,\ptt_{\alpha_m\beta}^{st}}{\ptt_{\alpha_m}^s},
\ee
where we have used that $\ptt_\beta^t=\ptt_{\alpha_m\beta}^{st}+\ptt_{\alpha_m}^t$ and $\ptt_{\alpha_m}^s=\ptt_{\alpha_m\beta}^{st}+\ptt_\beta^s$.
Thus, the expression in the square brackets in \eqref{detGram} becomes
\be
\frac{\ptt}{\ptt_{\alpha_m}^s}\[\ptt_\beta^s\,\ptt_{\alpha_m\beta}^{st}
-\ptt_{\alpha_m}^s\sum_{\ell=1}^{m-1}\frac{(\ptt_{\alpha_\ell\beta}^{Ls})^2\,\ptt_{\alpha_\ell}^R}{\ptt_{v_{\alpha_\ell}}\ptt_{\alpha_\ell}^L }\].
\ee
The new expression in the square brackets is exactly the same as in \eqref{detGram} where tree $\cT$ was replaced by subtree
$\cT_{\alpha_m}^s=\cT_{\alpha_{m-1}}$. Thus, one can repeat the above manipulation until one exhausts all terms in the sum.
As a result, the determinant of the Gram matrix turns out to be proportional to $\ptt_{\alpha_1\beta}^{st}$.
But since $\alpha_1=\beta$, this quantity, and hence the whole determinant, trivially vanish.

Next, we prove that \eqref{matrix-gener} is consistent with
$\cB\cdot \bfx=\left\{\frac{w^\alpha}{\sqrt{\Delta_\alpha}}\right\}$ and $\cB\cdot \cV=\cM$.
The first relation is a direct consequence of \eqref{dataechoice}. The second relation
requires to show that
$(\tbfu_\alpha, \bfu_\beta)=\ptt\,\Delta_\alpha^{1/2}\Delta_\beta^{-1/2}\cM_{\alpha\beta}$
or equivalently
$\sum_{\gamma}(\tbfu_\alpha, \bfu_\gamma) \sqrt{\Delta_\gamma\Delta_\beta}
\cM_{\gamma\beta}^{-1}=\ptt\,\Delta_\alpha\delta_{\alpha\beta}$.
Using the result \eqref{invmatM} for the matrix $\cM_{\gamma\beta}^{-1}$,
this relation can be written explicitly as
\be
\ptt_{v_\beta}(\tbfu_\alpha, \bfu_\beta)
+\(\sum_{e_\gamma \in E_\beta^R}
\eps_{\gamma\beta} \ptt_{\Lv{v_\beta}}
+\sum_{e_\gamma\in E_\beta^L}
\eps_{\gamma\beta} \ptt_{\Rv{v_\beta}}\)(\tbfu_\alpha, \bfu_\gamma)
=\ptt\,\Delta_\alpha\delta_{\alpha\beta},
\label{explicitmatrix}
\ee
where $E_{\beta}^{L}=\{e\in E_\cT:\ e\cap \cT_{\beta_L}\ne \emptyset, \ e\nsubseteq \cT_\beta\}$
and similarly for $E_{\beta}^{R}$.

Consider first the case $\alpha=\beta$.
From \eqref{res-scpr}, it immediately follows that the first term gives
$\ptt\,\Delta_\alpha$.
On the other hand, the second contribution sums over edges for which $\cT_\alpha\subseteq \cT_\gamma^s$ or $\cT_\gamma^t$,
and as noted above \eqref{res-scpr} this leads to vanishing of the scalar product.
Thus, in this case the relation \eqref{explicitmatrix} indeed holds.

Let us now show that it holds as well for $\alpha\ne \beta$.
To this end, one should consider several options. If $\cT_\alpha\cap \cT_\beta=\emptyset$, then $\cT_\alpha\subset \cT_\beta^s$ or $\cT_\beta^t$
which implies vanishing of the first term.
But the second term vanishes as well because the conditions $e_\gamma\subseteq \cT_\alpha$ and $e_\gamma \cap \cT_\beta\ne \emptyset$
are inconsistent with $\cT_\alpha\cap \cT_\beta=\emptyset$.

Similarly, if $\cT_\alpha\subset \cT_\beta$, one has $\cT_\alpha\subset \cT_\beta^s$ or $T_\beta^t$ which again leads to the vanishing of the first term,
whereas the vanishing of the second is a consequence of that $e_\gamma\subseteq \cT_\alpha$ implies $e_\gamma\subset \cT_\beta$
so that the sum over $e_\gamma$ is empty.

It remains to consider the case $\cT_\beta\subset \cT_\alpha$. It is clear that $\cT_\beta\subset \cT_\alpha^s$ or $\cT_\alpha^t$.
Without loss of generality, let us assume that $\cT_\beta\subset \cT_\alpha^s$ and $e_\alpha \subset \cT_\beta^t$.
Then according to \eqref{res-scpr}, the first term gives
$\ptt_{v_\beta}\,\ptt_{\alpha\beta}^{Ls}\,\ptt_\alpha^R\,\ptt$.
If instead we have chosen the orientation such that $e_\alpha \subset \cT_\beta^s$, then we would find
$-\ptt_{v_\beta}\,\ptt_{\alpha\beta}^{Lt}\,\ptt_\alpha^R\,\ptt$.
Similar results are obtained for each term in the sum of the second contribution.
Again without loss of generality we assume that for all relevant edges $e_\gamma$ one has $e_\alpha \subset \cT_\gamma^t$,
otherwise one flips their orientation.
Then the l.h.s. of \eqref{explicitmatrix} is proportional to
\be
\ptt_{v_\beta}\,\ptt_{\alpha\beta}^{Ls}
+\sum_{e_\gamma\in E_{\beta}^{R}}
\eps_{\gamma\beta} \ptt_\beta^L\,\ptt_{\alpha\gamma}^{Ls}
+\sum_{e_\gamma\in E_{\beta}^{L}}
\eps_{\gamma\beta} \ptt_\beta^R \,\ptt_{\alpha\gamma}^{Ls}.
\ee
Let $e_\alpha^\star\in E_{\beta}^{R}$ is the edge belonging to the path from $\cT_\beta$ to $e_\alpha$ (which may coincide with $e_\alpha$).
Then our choice of orientation implies that $\eps_{\gamma\beta}=-1$ for $e_\gamma\in E_\beta^L\cup\{e_\alpha^\star\}$ and
$\eps_{\gamma\beta}=1$ for $e_\gamma\in E_\beta^R\setminus\{e_\alpha^\star\}$.
Furthermore, one has
\be
\ptt_{\alpha\gamma_\alpha^\star}^{Ls}=\ptt_{v_\beta}+\sum_{e_\gamma \in E_\beta^L\cup E_\beta^R\setminus\{e_\alpha^\star\}}\ptt_{\alpha\gamma}^{Ls},
\qquad
\ptt_{\alpha\beta}^{Ls}=\ptt_\beta^L+\sum_{e_\gamma \in E_\beta^L}\ptt_{\alpha\gamma}^{Ls}.
\ee
As a result, one finds
\be
(\ptt_\beta^L+\ptt_\beta^R)\(\ptt_\beta^L+\sum_{e_\gamma \in E_\beta^L}\ptt_{\alpha\gamma}^{Ls}\)
-\ptt_\beta^L \,\ptt_{\alpha\gamma_\alpha^\star}^{Ls}
+\sum_{e_\gamma\in E_{\beta}^{R}\setminus\{e_\alpha^\star\}} \ptt_\beta^L\,\ptt_{\alpha\gamma}^{Ls}
-\sum_{e_\gamma\in E_{\beta}^{L}}
\ptt_\beta^R \,\ptt_{\alpha\gamma}^{Ls}=0.
\ee
This completes the proof of the required statement \eqref{explicitmatrix}.


\addtocontents{toc}{\vspace{-0.1cm}}
\section{Proofs of propositions}
\label{ap-proofs}

In this appendix we prove several propositions which we stated in the main text.

\subsection*{Proposition \ref{prop-iterwhg} \label{ap-prop1}}

To prove the recursive equation \eqref{recurs-whd-full}, let us note that the tree index $\gtr$ satisfies a very similar equation
(cf. \eqref{itereq} or \eqref{F-ansatz}) which can be seen as the origin of its representation \eqref{defgtree} in terms of attractor flow trees.
The only difference is the absence of the last term in \eqref{recurs-whd-full}.
Therefore, it is easy to see that this equation implies a simple relation between $\whg_n$ and $\gtr$
\be
\whg_n(\{\gama_i\},z^a)=\sum_{n_1+\cdots +n_m= n\atop n_k\ge 1}
\gtri{m}(\{\gama'_k\},z^a)
\prod_{k=1}^m W_{n_k}(\gama_{j_{k-1}+1},\dots,\gama_{j_{k}}),
\label{relwhgg}
\ee
where as usual we use notations from \eqref{groupindex}.

Next, we substitute this relation into the expansion \eqref{multihd-full}. The result can be represented in the following form
\be
h^{\rm DT}_{p,q}=\sum_{\sum_{i=1}^n \gama_i=\gama}
\gtr(\{\gama_i\},z^a)\,e^{\pi\I \tau Q_n(\{\gama_i\})}
\prod_{i=1}^n h^R_{p_i,\mu_i}(\tau),
\label{contr-dn-full}
\ee
where we introduced
\be
\begin{split}
h^{R}_{p,\mu}=&\, \sum_{\sum_{i=1}^n \gama_i=\gama}W_n(\{\gama_i\})\,e^{\pi\I \tau Q_n(\{\gama_i\})}
\prod_{i=1}^n \whh_{p_i,\mu_i}(\tau)
\\
=&\, \sum_{\sum_{i=1}^n \gama_i=\gama}\[\sum_{T\in\IT_n^{\rm S}}(-1)^{n_T}\prod_{v\in V_T} R_{v}\]\,e^{\pi\I \tau Q_n(\{\gama_i\})}
\\
&\, \times
\prod_{i=1}^n \(h_{p_i,\mu_i}
+\sum_{n_i=2}^\infty
\sum_{\sum_{j=1}^{n_i} \gama'_j=\gama}
R_{n_i}(\{\gama'_j\})
\, e^{\pi\I \tau Q_{n_i}(\{\gama'_j\})}
\prod_{j_i=1}^{n_i} h_{p'_{j_i},\mu'_{j_i}}\)
\end{split}
\label{contr-R-full}
\ee
and in the last relation we used the definition of $W_n$ \eqref{defSm} and the expansion of $\whh_{p,\mu}$ \eqref{exp-whh}.
The crucial observation is that if one picks up a factor $R_{n_i}$ from the second line of \eqref{contr-R-full},
appearing due to the expansion of $\whh_{p_i,\mu_i}$,
and combines it with the contribution of a tree $T$ from the first line, one finds the opposite of
the contribution of another tree obtained from $T$ by adding $n_i$ children to its $i$th leaf.
As a result, all such contributions cancel and the function \eqref{contr-R-full} reduces to the trivial term
\be
h^{R}_{p,\mu}=h_{p,\mu}.
\label{hRh-full}
\ee
Substituting this into \eqref{contr-dn-full}, it gives back the original expansion \eqref{multih}
of the generating function of DT invariants, which proves the recursive equation \eqref{recurs-whd-full}.

Finally, let us evaluate \eqref{multihd-full} at the attractor point $z^a_\infty(\gamma)$.
At this point the DT invariants coincide with
MSW invariants so that the l.h.s. becomes the generating function $h_{p,\mu}$.
Meanwhile,  all factors $\Delta_{\gamma_L\gamma_R}^z$ vanish at the attractor point
and $\whg_n$ reduces to $W_n$.
As a result, the relation \eqref{multihd-full} reproduces \eqref{exp-hwh}, which completes the proof of the proposition.

\subsection*{Proposition \ref{prop-solvecond} \label{ap-prop2}}

We prove Proposition \ref{prop-solvecond} by induction.
For $n=2$ the recursive relation \eqref{iterseed} reads
\be
\gf_2(\gama_1,\gama_2;c_1)-\gf_2(\gama_1,\gama_2;\beta_{21})=-\frac14\, \bigl(\sgn(c_1)-\sgn(\beta_{21}) \bigr)\, \kappa(\gamma_{12}),
\ee
where we took into account that $\cs_1=c_1$ and $\Gamma_{21}=\beta_{21}=-\gamma_{12}$.
Since $\gf_2$ is supposed to have discontinuities only at walls of marginal stability, it must not involve signs of DSZ products.
Therefore, we are led to take
\be
\gf_2(\gama_1,\gama_2;c_1)=-\frac14\, \sgn(c_1) \kappa(\gamma_{12}).
\ee
Then \eqref{rel-gE} and \eqref{Emrec} imply
\be
\cE_2=\frac14\,\sgn(\gamma_{12})\,\kappa(\gamma_{12})+R_2,
\ee
so that
the ansatz \eqref{iterDn} reads
\be
\whg_2(\gama_1,\gama_2;c_1)=\gf_2(\gama_1,\gama_2;c_1)-\cE_2(\gama_1,\gama_2)
=-\frac14\,\Bigl[\sgn(c_1)+\sgn(\gamma_{12})\Bigr]\,\kappa(\gamma_{12})-R_2(\gama_1,\gama_2),
\ee
which reproduces the recursive equation \eqref{recurs-whd-full}.
Furthermore, in appendix \ref{ap-generror} it is shown that there is a smooth solution of Vign\'eras' equation
which asymptotes the function $(\bfv,\bfx)\sign(\bfv,\bfx)$, coinciding with the (rescaled) first term in $\cE_2$
for $\bfv=\bfv_{12}$ \eqref{defvij}. It is given by
\be
\tPhi_1^E(\bfv,\bfv;\bfx)=\bfv\cdot\(\bfx+\frac{1}{2\pi}\,\p_\bfx\)\Erf\(\frac{\sqrt{\pi}\bfv\cdot\bfx}{|\bfv|}\),
\ee
which corresponds to the following choice of $R_2$ \cite{Alexandrov:2016tnf}
\be
R_2=\frac{(-1)^{\gamma_{12}}}{8\pi}\, |\gamma_{12}|
\,\beta_{\frac{3}{2}}\!\({\frac{2\tau_2\gamma_{12}^2 }{(pp_1p_2)}}\),
\ee
where $\beta_{\frac{3}{2}}(x^2)=2|x|^{-1}e^{-\pi x^2}-2\pi  \Erfc(\sqrt{\pi} |x|)$.
Note that the resulting $\cE_2$ depends on the electric charges only through the DSZ product $\gamma_{12}$.
Finally, it is immediate to see that the kernel $\whgPhi_2$ \eqref{kerPhiinth} satisfies \eqref{wanted}.

Now we assume that \eqref{recurs-whd-full} is consistent with the ansatz \eqref{iterDn} for all orders up to $n-1$
and check it at order $n$.
Denoting the second term in \eqref{iterDn} by $\whgp_n$ and substituting this ansatz into the r.h.s. of  \eqref{recurs-whd-full}, one finds
\be
\qquad\qquad
\Sym\left\{\sum_{\ell=1}^{n-1}g_2^\ell\,
\Bigl[\gf_\ell\,\gf_{n-\ell}-\whgp_\ell\,\whg_{n-\ell}-\whg_\ell\,\whgp_{n-\ell}-\whgp_\ell\,\whgp_{n-\ell}\Bigr]_{c_i\to c_i^{(\ell)}}
\right\}
+W_n,
\label{recurs-subs}
\ee
where
\be
g_2^\ell=-\hf\,\Delta_{\gamma_L^\ell\gamma_R^\ell}^z\,\kappa(\gamma_{LR}^\ell)
=\frac14\, \bigl(\sgn(\cs_\ell)-\sgn(\Gamma_{n\ell})\bigr)\,\kappa(\Gamma_{n\ell}).
\ee
In the first term proportional to $\gf_\ell\,\gf_{n-\ell}$, one can apply the relation \eqref{iterseed},
which together with \eqref{rel-gE} gives $\gf_n-\cEf_n$.
The other terms in the sum over $\ell$ can be reorganized as follows
\be
\begin{split}
&\,
-\Sym\left\{\sum_{\ell=1}^{n-1} g_2^\ell
\sum_{n_1+\cdots +n_m= n\atop n_k\ge 1, \ m<n, \ \ell\in \{j_k\}} \whg_{k_0}(c_i^{(\ell)})\,
\whg_{m-k_0}(c_i^{(\ell)})\prod_{k=1}^m \cE_{n_k}\right\}
\\
=&\,
-\Sym\left\{\sum_{n_1+\cdots +n_m= n\atop n_k\ge 1, \ 1<m<n}
\[\,\sum_{k_0=1}^{m-1} g_2^{j_{k_0}}\,
\whg_{k_0}(c_i^{(\ell)})\, \whg_{m-k_0}(c_i^{(\ell)})\,\]
\prod_{k=1}^m \cE_{n_k}\right\}.
\end{split}
\label{inductF2}
\ee
Here we first combined three contributions into one sum over splittings
by adding the condition $\ell\in \{j_k\}$, with $k_0$ being the index for which $j_{k_0}=\ell$, and then
interchanged the two sums which allows to drop the condition $\ell\in \{j_k\}$, but adds the requirement $m>1$
(following from $\ell\in \{j_k\}$ in the previous representation).
In square brackets one recognizes the first contribution from the r.h.s. of \eqref{recurs-whd-full} with $n$ replaced by $m<n$.
Hence, it is subject to the induction hypothesis which allows to replace this expression by $\whg_{m}(\{\gamma'_k\},z^a)-W_m(\{\gamma'_k\})$.
Combining all contributions together, one concludes that \eqref{recurs-subs} is equal to
\be
\gf_n-\cEf_n
-\Sym\left\{\sum_{n_1+\cdots +n_m= n\atop n_k\ge 1, \ 1<m<n}
\bigl(\whg_{m}-W_m\bigr)\prod_{k=1}^m \cE_{n_k}\right\}+W_n.
\label{rhsofreceq}
\ee
The contributions involving $W$'s can be combined into one sum by dropping the condition $m<n$.
The resulting sum coincides with the r.h.s. of \eqref{Emrec} so that these contributions can be replaced by $-\cEp_n$.
Combined with $-\cEf_n$, this gives $-\cE_n$ and allows to drop the condition $m>1$ in the remaining sum with $\whg_{m}$.
As a result, \eqref{rhsofreceq} becomes equivalent to the r.h.s. of \eqref{iterDn}, which proves the consistency of this ansatz
with the recursive equation.

Finally, let us show that the ansatz satisfies the modularity constraint \eqref{wanted}.
The crucial observation is that the vectors $\bfv_{ij}$ and $\bfu_{ij}$ \eqref{defvij} satisfy
\be
(\bfv_{i+j,k},\bfv_{ij})=(\bfu_{i+j,k},\bfv_{ij})=0,
\label{orthrel}
\ee
where we abused notation and denoted $\bfv_{i+j,k}=\bfv_{ik}+\bfv_{jk}$, etc.
These orthogonality relations, together with the assumption that $\cE_n$ depend on electric charges only through
the DSZ products $\gamma_{ij}\sim (\bfv_{ij},\bfx)$,
imply factorization of the action of Vign\'eras' operator on the kernel corresponding to the second term in \eqref{iterDn}.
Indeed, all contributions of $\p_\bfx^2$ where two derivatives act on different factors vanish and the action reduces to the sum of terms
where Vign\'eras' operator acts on one of the factors. But since it is supposed to vanish on $\cEPhi_n$, one obtains
the simple result
\be
V_{n-1} \cdot \whgPhi_{n}=V_{n-1} \cdot \gfPhi_{n}
-\Sym\left\{\sum_{n_1+\cdots +n_m= n\atop n_k\ge 1, \ m<n}
(V_{m-1}\cdot \whgPhi_m)
\prod_{k=1}^m \cEPhi_{n_k}\right\}.
\label{actVans}
\ee
Since for $n=2$ the constraint was already shown to hold, one can proceed by induction. Then in the second term
one can substitute the r.h.s. of \eqref{wanted}, whereas the first term can be evaluated using the recursive relation \eqref{iterseed}.
First of all, by the same reasoning as above, away from discontinuities, the action of Vign\'eras' operator is factorized and actually vanishes.
Furthermore, since $\gf_n$ have discontinuities only at walls of marginal stability, to obtain the complete action, it is enough
to consider it only on $\sgn(\cs_\ell)$. Since at $\cs_\ell=0$ one has $\cl_i=c_i$ (see \eqref{flowc}), one finds that $\gfPhi_n$ satisfy
exactly the same constraint as \eqref{wanted}. Thus, one can rewrite \eqref{actVans} as
\bea
V_{n-1} \cdot \whgPhi_{n}&=&
\Sym\sum_{\ell=1}^{n-1}
\Bigl(\bfu_\ell^2\,\Delta_{n,\ell}^{\gf} \,\delta'(\bfu_\ell\cdot\bfx)
+ 2\bfu_\ell\cdot\p_\bfx \Delta_{n,\ell}^{\gf} \,\delta(\bfu_\ell\cdot\bfx)\Bigr)
\label{actVans2}\\
&&
-\Sym\left\{\sum_{n_1+\cdots +n_m= n\atop n_k\ge 1, \ 1<m<n}\[
\sum_{\ell=1}^{n-1}
\Bigl(\bfu_\ell^2\,\Delta_{m,\ell}^{\whg} \,\delta'(\bfu_\ell\cdot\bfx)
+2\bfu_\ell\cdot\p_\bfx \Delta_{m,\ell}^{\whg} \,\delta(\bfu_\ell\cdot\bfx)\Bigr)\]
\prod_{k=1}^m \cEPhi_{n_k}\right\}.
\nn
\eea
Note that the orthogonality relation allows to include $\cEPhi_{n_k}$ under the derivative in the last term.
Then one can perform the same manipulations with the sum over splittings
as in \eqref{inductF2} but in the inverse direction, which directly leads to the constraint \eqref{wanted}.

\subsection*{Proposition \ref{prop-limitG}}

Our goal is to find an explicit expression for $\tPhif_n$. Using the fact that the limit of large $\bfx$ of the generalized error function $\Phi_{n}^E$
is the product of $n$ sign functions, one immediately obtains
\be
\tPhif_n(\bfx)=
\frac{1}{2^{n-1} n!}\sum_{\cT\in\, \IT_n^\ell}\[\prod_{e\in E_\cT}\cD( \bfv_{s(e) t(e)})\]
\[\prod_{e\in E_{\cT}}\sign(\bfu_e,\bfx)\],
\label{limitwhPhi}
\ee
where $\cD(\bfv)$ are the covariant derivative operators \eqref{defcDif}.
The action of the derivatives $\p_\bfx$ on the sign functions can be ignored (since the original function is smooth),
however there are additional contributions due to the mutual action of the operators $\cD$.
Similar contributions were discussed in a similar context in appendix \ref{ap-twistint},
where they were shown to cancel, but here they turn out to leave a finite remainder.
The contribution generated by the mutual action of two operators $\cD( \bfv_{s(e) t(e)})$
is non-vanishing only if the two edges $e_1,e_2$ have a common vertex. In this case it contributes the factor
\be
\frac{(p_{\ver_1}p_{\ver_2}p_{\ver_{3}})}{2\pi}\, \sign(\bfu_{e_1},\bfx)\, \sign(\bfu_{e_2},\bfx) ,
\label{factordersign}
\ee
where $\ver_1,\ver_2,\ver_3$ are the tree vertices joint by the edges $e_1=(\ver_2,\ver_3)$, $e_2=(\ver_1,\ver_3)$. \label{edges}
Again, considering the three trees shown in Fig. \ref{fig-Vign3}, one
can note that the vectors $\bfu_e$ defined by these trees satisfy the following relations:
$\bfu_{e_1}^{(1)}=-\bfu_{e_3}^{(2)}$, $\bfu_{e_2}^{(1)}=\bfu_{e_3}^{(3)}$, $\bfu_{e_2}^{(2)}=\bfu_{e_1}^{(3)}$
and $\bfu_{e_2}^{(2)}=\bfu_{e_1}^{(1)}+\bfu_{e_2}^{(1)}$,
where $e_3=(\ver_1,\ver_2)$ and we indicated by upper index the tree with respect to which the vector is defined.
Therefore, the contributions generated by these trees combine into
\bea
&& \frac{(p_{\ver_1}p_{\ver_2}p_{\ver_{3}})}{2\pi}\Bigl[ \sign(\bfu_{e_1}^{(1)},\bfx)\, \sign(\bfu_{e_2}^{(1)},\bfx)
+\sign(\bfu_{e_2}^{(2)},\bfx)\, \sign(\bfu_{e_3}^{(2)},\bfx)-\sign(\bfu_{e_1}^{(3)},\bfx)\, \sign(\bfu_{e_3}^{(3)},\bfx)\Bigr]
\nn\\
&=& -\frac{(p_{\ver_1}p_{\ver_2}p_{\ver_{3}})}{2\pi},
\label{signfactorscomb}
\eea
where we used the above relations between the vectors and the sign identity \eqref{signprop-ap} for $x_s=(\bfu_{e_s}^{(1)},\bfx)$, $s=1,2$,
Thus, unlike in  \eqref{case3}, we now get a non-vanishing result, due to  the mutual action of derivative operators.

Each pair of intersecting edges leads to a contribution which recombines the contributions of the edges from
the three trees into a single factor \eqref{signfactorscomb}.
Furthermore, the sum over trees implies that we have to sum over all possible subtrees $\cT_1,\cT_2,\cT_3$,
in particular, over all possible allocations of different subtrees to the vertices $\ver_1,\ver_2,\ver_3$.
The sign factors $\sign(\bfu_e,\bfx)$ for edges of these subtrees do not depend on this allocation, but the factors
$\cD( \bfv_{s(e) t(e)})$ do depend for the edges connecting to one of these vertices.
It is easy to see that the sum over allocations effectively replaces the three vertices by
a single one labelled
by the total charge $p_{\ver_1}+p_{\ver_2}+p_{\ver_{3}}$.
As a result, one obtains that the function \eqref{limitwhPhi} can be represented as a sum over marked trees
where a mark corresponds to a collapse of a pair of intersecting edges and contributes the factor \eqref{signfactorscomb}.
More precisely, one has
\be
\tPhif_n(\bfx)=\frac{1}{2^{n-1} n!}\sum_{m=0}^{[(n-1)/2]}\frac{(-1)^m}{(2\pi)^m}
\sum_{\cT\in\, \IT_{n-2m,m}^\ell}
\[\prod_{\ver\in V_\cT}\cP_{m_\ver}(\{p_{\ver,s}\})\]
\prod_{e\in E_{\cT}}(\bfv_{s(e) t(e)},\bfx)\,\sign(\bfu_e,\bfx),
\label{asympwhPhi}
\ee
where
\be
\cP_{m}(\{p_{s}\})=\sum_{\cI_1\cup\cdots\cup\cI_{m}=\Zv_{2m+1}}\, \prod_{j=1}^{m}(p_{j_1}p_{j_2}p_{j_3}).
\label{defPver}
\ee
This factor collects the weights \eqref{signfactorscomb} assigned to a vertex due to collapse of $m$ pairs of edges.
It is represented as a sum over all possible splittings of the set $\Zv_{2m+1}=\{1,\dots,2m+1\}$ into union of triples $\cI_j=\{j_1,j_2,j_3\}$
such that the labels in one triple $\cI_j$ are all different, two different triples can have at most one common label,
and there are no closed cycles in the sense that there are no subsets $\{{j_k}\}_{k=1}^r$ such that
$\cI_{j_k}\cap\cI_{j_{k+1}}\ne \emptyset$ where $j_{r+1}\equiv j_1$.
This sum simply counts all possible splittings of a collection of $2m$ joint edges into $m$ intersecting pairs, suppressing
for each pair the distinction between the three configurations of Fig. \ref{fig-Vign3}.

However, the representation \eqref{defPver} is not very convenient for our purposes.
An alternative representation can be obtained by noting that, instead of collapsing all edges at once,
one can collapse first one pair, sum over all configurations (i.e. allocations of subtrees), then collapse another pair, and so on.
In this approach at each step the sum over different configurations ensures that all factors assigned to the elements of
the surviving tree depend only on the sum of collapsing charges.
As a result, one obtains a hierarchical structure described by a rooted ternary tree $T$,
where the leaves correspond to the vertices of the original unrooted tree which have collapsed into one vertex with $m_\ver$ marks corresponding
to the root of $T$. The other vertices of $T$ are then in one-to-one correspondence with marked vertices appearing at
intermediate stages of the above process.

This procedure gives rise to the representation \eqref{defPver-tree} of the weight factor $\cP_m$
in terms of a sum over rooted ternary trees.
A non-trivial point which must be taken into account is that the procedure leading
to this representation overcounts different configurations.
As a result, each tree is weighted by the rational coefficient $N_{\hT}/m!$
where $\hT$ is the rooted tree obtained from $T$ by dropping all leaves and $N_{\hT}$
is the number of ways of labelling the vertices of $\hT$ with increasing labels, which already appeared in appendix \ref{ap-theorem}.
Here the numerator takes into account that the tree $T$ is generated $N_{\hT}$ times in the collapsing process, whereas the denominator
removes the overcounting produced by specifying the order in which the $m$ pairs of edges are collapsed.
Finally, we apply Lemma \ref{lemma-ntrees} from appendix \ref{ap-theorem}.
Since $n_{\hT}=m$, the coefficients coincides with the inverse of the tree factorial \eqref{defcT}.
Taking into account that substitution of \eqref{asympwhPhi} into \eqref{defgf-hPhi}
gives exactly \eqref{asymp-compl}, this completes the proof of the proposition.

\subsection*{Proposition \ref{prop-coeff}}

Our aim is to show that the recursive formula \eqref{res-aT} solves the equations \eqref{eqcoef}.
We will proceed by induction,  starting with the case $n=3$ where there is a single unrooted tree
and the system \eqref{eqcoef} contains a single equation corresponding to the trivial trees $\cT_r$ consisting of one vertex.
Thus, it is solved by $a_{\bullet\!\mbox{-}\!\bullet\!\mbox{-}\!\bullet}=\frac13$ consistently with \eqref{res-aT}.

Let us now consider trees with $n$ vertices assuming that for all trees with less number of vertices \eqref{res-aT} holds.
We start by noting that for every vertex $\ver\in V_{\cT_r}$, among the subtrees $\cT_{r,s}(\ver)\subset \cT_r$ obtained by removing the vertex $\ver$,
there is one which contains $\ver_r$, which we denote by $\cT_{r,s_0}(\ver)$.
(If $\ver=\ver_r$, we take $\cT_{r,s_0}(\ver)=\emptyset$.)
Since every tree $\hcT_r$ is a union of the three trees $\cT_r$, substituting \eqref{res-aT}
into the l.h.s. of \eqref{eqcoef}, one  obtains that the sum over vertices of $\hcT_r$ can be represented as
\be
\frac{1}{n}\sum_{r=1}^3 \sum_{\ver\in V_{\cT_r}} \epsilon_\ver\(a_{\hcT_{1,r}(\ver)} +a_{\hcT_{2,r}(\ver)} -a_{\hcT_{3,r}(\ver)}\)
\prod_{s=1\atop s\ne s_0}^{n_\ver} a_{\cT_{r,s}(\ver)},
\label{lhs-coef}
\ee
where $\hcT_{r',r}(\ver)$ is obtained from $\hcT_{r'}$ by replacing $\cT_r$ by $\cT_{r,s_0}(\ver)$.
Applying the equations \eqref{eqcoef}, this expression gives
\be
\frac{1}{n}\[ a_{\cT_2} a_{\cT_3}\sum_{\ver\in V_{\cT_1}}\epsilon_\ver\prod_{s=1}^{n_\ver} a_{\cT_{1,s}(\ver)}
+a_{\cT_1} a_{\cT_3}\sum_{\ver\in V_{\cT_2}}\epsilon_\ver\prod_{s=1}^{n_\ver} a_{\cT_{2,s}(\ver)}
+a_{\cT_1} a_{\cT_2}\sum_{\ver\in V_{\cT_3}}\epsilon_\ver\prod_{s=1}^{n_\ver} a_{\cT_{3,s}(\ver)}
\].
\label{lhs-coef2}
\ee
Since the trees $\cT_{r}$ have less than $n$ vertices, they are subject to the induction hypothesis which allows to replace
the sums over vertices by $n_r a_{\cT_r}$ where $n_r$ is the number of vertices in $\cT_r$.
Taking into account that $n_1+n_2+n_3=n$, the expression \eqref{lhs-coef2} reduces to $a_{\cT_1} a_{\cT_2}a_{\cT_3}$,
which proves that the equations \eqref{eqcoef} are indeed satisfied.

\subsection*{Proposition \ref{prop-limitcompl}}

The evaluation of the large $\bfx$ limit of the function $\tcEPhi_n(\bfx)$ is very similar to the calculation
done in the proof of Proposition \ref{prop-limitG}.
First, using the asymptotics of the generalized error functions, we can write
\be
\lim_{\bfx\to\infty}\tcEPhi_n(\bfx)=\frac{1}{2^{n-1} n!}
\!\!
\sum_{m=0}^{[(n-1)/2]}
\!\!\!
\sum_{\cT\in\, \IT_{n-2m,m}^\ell}
\[\prod_{\ver\in V_\cT}\cD_{m_\ver}(\{\gama_{\ver,s}\})\]
\[\prod_{e\in E_\cT}\cD( \bfv_{s(e) t(e)})\]
\[\prod_{e\in E_{\cT}}\sign(\bfu_e,\bfx)\]\! .
\label{limithEPhi}
\ee
We then observe that the two last factors depend only
on sums of charges appearing in the operators $\cD_{m_\ver}$ in the first factor.
This implies that the vectors on which these factors depend are orthogonal to the vectors determining the operators
in \eqref{defcDcT}. Therefore, these operators effectively act on a constant and can be expanded as
\be
\cD_{m}\cdot 1 = \sum_{k=0}^m\cV_{m,k},
\label{exp-cD}
\ee
where $\cV_{m,k}$ are homogeneous polynomials in $\bfx$ of degree $2(m-k)$. In particular,
the highest degree term coincides with the function defined in \eqref{defcV-mw},
$\cV_{m,0}(\{\gama_{s}\})= \cV_{m}(\{\gama_{s}\})$.

Next, the mutual action of the derivative operators from the product over edges in the second factor
generates contributions described by trees with pairs of collapsed edges
replaced by marks. The difference here is that the original trees were also marked.
This fact does not change the structure of the result, which is again given by a sum over marked trees,
but it affects the weight associated with marks. Denoting this weight for a vertex with total $m$ marks (old and new) by $\cVt_m$,
we recover the equation \eqref{asymphcEPhi} in the statement of the proposition, provided that $\cVt_m=\cV_m$.
We now proceed to prove the latter identity.

The weight factor $\cVt_m$ coming from the above procedure is given by
\be
\small
\cVt_m(\{\gama_{s}\})=\Sym\!\! \left\{\sum_{m_0=0}^{m}\frac{(-1)^{m_0}}{(2\pi)^{m_0}}
\!\!\!\!
\sum_{\sum\limits_{r=1}^{2m_0+1}m_r=m-m_0}
\!\!\!\!\!\! \!\!\!\!\!\!
C(\{m_r\})
\,\cP_{m_0}(\{p'_{r}\})
\!\!
\prod_{r=1}^{2m_0+1}\sum_{k_r=0}^{m_r}\cV_{m_r,k_r}(\gama_{j_{r-1}+1},\dots,\gama_{j_{r}})
\right\},
\label{defVver-tree}
\ee
where the second sum goes over all ordered decompositions of $m-m_0$ into non-negative integers,
$C(\{m_r\})=\frac{(2m+1)!}{\prod_r (2m_r+1)!}$, and we used notations similar to \eqref{groupindex},
\be
j_0=0,
\qquad
j_r=m_1+\cdots + m_r,
\qquad
\gamma'_r=\gamma_{j_{r-1}+1}+\cdots +\gamma_{j_{r}}.
\label{groupindex-ap}
\ee
Here the first factor $\cP_{m_0}$ arises due to collapse of $m_0$ pairs of edges in the marked trees one sums over in \eqref{limithEPhi},
whereas the factors given by the sums over $k_r$ are the ones corresponding to the ``old" marks assigned to that trees.

Now let us use the equations \eqref{eqcoef} determining the coefficients $a_\cT$.
It is easy to see that they imply the following constraint on the next-to-highest degree term in the expansion \eqref{exp-cD}
\be
\cV_{m,1}(\{\gama_{s}\})=\frac{1}{6\cdot 2\pi} \sum_{m_1+m_2+m_3=m-1\atop m_r\ge 0}C(\{m_r\})
\Sym\left\{
(p'_1 p'_2 p'_3)\prod_{r=1}^3 \cV_{m_r}(\{\gama_i\}_{i=j_{r-1}+1}^{j_r})
\right\},
\label{cond-coefaT}
\ee
where $j_r=\sum_{s=1}^r (2m_s+1)$ and $p'_r=\sum_{i=1}^{2m_r+1}p_{j_{r-1}+i}$.
Applying this constraint recursively, one can express all $\cV_{m,k}$ for $k>0$ through $\cV_{m}$.
The idea is to replace the factors $\cV_{m_r}$ by the operators $\cD_{m_r}$.
Then one can realize that, extracting from the resulting function the terms homogeneous in $\bfx$ of order $2(m-2)$
(which have two factors of $(p^3)$), one obtains the result for $2\cV_{m,2}$. Proceeding in this way, one arrives at
the representation very similar to \eqref{defVver-tree},
\be
\cV_{m,k}(\{\gama_{s}\})=\frac{1}{(2\pi)^k}\Sym \left\{\sum_{\sum\limits_{r=1}^{2k+1}m_r=m-k}
\!\!\!\! C(\{m_r\})\,
\cP_{k}(\{p'_{r}\})
\prod_{r=1}^{2k+1}\cV_{m_r}(\gama_{j_{r-1}+1},\dots,\gama_{j_{r}})
\right\}.
\label{resVmk}
\ee
Substituting it into \eqref{defVver-tree} and using the expression \eqref{defPver-tree} for the factors $\cP_{m}$ through
the sum over rooted ternary trees, one can recombine all these sums in the following way
\be
\begin{split}
\cVt_m=&\,\Sym \Biggl\{\sum_{k=0}^{m}\sum_{\sum\limits_{r=1}^{2k+1}m_r=m-k}
\!\!\!\!\!\! C(\{m_r\})
\sum_{T\in\, \IT_{2k+1}^{(3)}(\{\gama'_r\})}
\frac{1}{T!}
\sum_{T'\subseteq T} (-1)^{n_{T'}}
\prod_{v\in V_{T'}}\frac{n_v(T)}{n_v(T')}
\\
&\,
\times\prod_{v\in V_{T}}\frac{(p_{d_1(v)}p_{d_2(v)}p_{d_3(v)})}{2\pi}
\prod_{r=1}^{2k+1}\cV_{m_r}(\gama_{j_{r-1}+1},\dots,\gama_{j_{r}})
\Biggr\},
\end{split}
\label{defVver-tree2}
\ee
where the sum over $T'$ is the sum over subtrees of $T$ having the same root.
In terms of the variables appearing in \eqref{defVver-tree}, one can identify $k=m_0+\sum_r k_r$ and $n_{T'}=m_0$.
Thus, $T'$ is the tree appearing in the decomposition of $\cP_{m_0}$, whereas $T$ is its union with $2m_0+1$ trees $T_r$
appearing in the decomposition of $\cP_{k_r}$, i.e.
$T=T'\cup \(\cup_r T_r\)$ where $T_r$ are the trees rooted at leaves of $T'$,
Finally, we took into account that for such trees one has
\be
T!=T'!\,\prod_r T_r!\,\prod_{v\in V_{T'}}\frac{n_v(T)}{n_v(T')}\,.
\ee

Remarkably, the sum over subtrees in \eqref{defVver-tree2} factorizes and for a fixed number of vertices $n_{T'}=m_0$ is subject to
Theorem \ref{theorem} where the r\^ole of the trees is played by rooted ternary trees $T$ and $T'$ after stripping out their leaves.
As a result, one finds for $k>0$
\be
\sum_{T'\subseteq T} (-1)^{n_{T'}}
\prod_{v\in V_{T'}}\frac{n_v(T)}{n_v(T')}
=\sum_{m_0=0}^k\frac{(-1)^{m_0}k!}{m_0!(k-m_0)!}=(1-1)^k = 0.
\ee
Thus, the only non-vanishing contribution is the one with $k=0$ which coincides with $\cV_m$.
This is what we had to show and therefore completes the proof of the proposition.

\subsection*{Proposition \ref{prop-solveiterg0} \label{ap-prop3}}

For simplicity, let us first show that the recursive equation \eqref{iterseed}
is satisfied by the contribution to $\gf_n(\{\gama_i,c_i\})$ given by the trees without any marks,
i.e. by the function
\be
\gs_n(\{\gama_i,c_i\})=\frac{(-1)^{n-1+\sum_{i<j} \gamma_{ij} }}{2^{n-1} n!}
\sum_{\cT\in\, \IT_{n}^\ell}
\prod_{e\in E_{\cT}}\gamma_{s(e) t(e)}\,\sign(\cs_e).
\label{defgf0}
\ee
Then the inclusion of marks will be straightforward because the corresponding contributions can be dealt with essentially in the same way
as the contribution \eqref{defgf0}.

To start with, we substitute $\gs_n$ into the l.h.s. of the recursive equation and decompose $\Gamma_{n\ell}=-\sum_{i=1}^\ell\sum_{j=\ell+1}^n \gamma_{ij}$.
Then the crucial observation is that this double sum, the sum over $\ell$ and the two sums over trees
(over $\IT_\ell^\ell$ and $\IT_{n-\ell}^\ell$) are  all equivalent
to a single sum over trees with $n$ vertices, i.e. over $\IT_n^\ell$, supplemented by the sum over edges.
Namely, one can do the following replacement
\be
\hf \Sym
\sum_{\ell=1}^{n-1}\frac{1}{\ell!(n-\ell)!}\sum_{\cT_L\in\, \IT_\ell^\ell}\sum_{\cT_R\in\, \IT_{n-\ell}^\ell}\sum_{i=1}^\ell\sum_{j=\ell+1}^n
=\frac{1}{n!}\sum_{\cT\in\, \IT_n^\ell}\sum_{e\in E_\cT}.
\label{relsums}
\ee
The idea is that on the l.h.s. one sums over all possible splittings of unrooted labelled trees with $n$ vertices
into two trees with $\ell$ and $n-\ell$ vertices. Such splitting can be done by cutting an edge and then
$i,j$ correspond to the labels of the vertices joined by the cutting edge. The binomial
coefficient $\frac{n!}{\ell!(n-\ell)!}$
takes into account that after splitting the vertices of the first tree can have arbitrary labels from the set $\{1,\dots,n\}$
and not necessarily $\{1,\dots,\ell\}$, whereas $\hf$ avoids doubling due to the symmetry between $\cT_L$ and $\cT_R$.

It is easy to check that all factors in \eqref{iterseed} fit this interpretation and
the l.h.s. takes the following form
\be
\frac{(-1)^{n-1+\sum_{i<j} \gamma_{ij} }}{2^{n-1} n!}\sum_{\cT\in\, \IT_n^\ell}
\prod_{e\in E_{\cT}}\gamma_{s(e) t(e)}\sum_{e\in E_\cT} \bigl( \sign(\cs_e)-\sign (\Gamma_{e})\bigr)
\prod_{e'\in E_{\cT}\setminus\{e\}}\sign\(\cs_{e'}-\frac{\Gamma_{e'}}{\Gamma_{e}}\, \cs_e\),
\label{lhsrecursive}
\ee
where $\Gamma_e$ was defined in \eqref{defGame}.
Finally, we apply the following sign identity, established in \cite[Eq.(A.7)]{Alexandrov:2018iao},
\be
\sum_{\beta=1}^m \(\sgn(x_\beta)-1\) \prod_{\alpha=1\atop \alpha\ne \beta}^m \sgn(x_\alpha-x_\beta)= \prod_{\alpha=1}^m \sgn(x_\alpha)-1,
\label{signident3}
\ee
where one should take the label $\alpha$ to run over $m=n-1$ edges of a tree $\cT$,
identify $x_\alpha=\cs_\alpha/\Gamma_{\alpha}$, and multiply it by $\prod_{\alpha=1}^{n-1}\sgn(\Gamma_\alpha)$.
Then the expression \eqref{lhsrecursive} reduces to $\gs_n(\{\gama_i,c_i\})-\gs_n(\{\gama_i,\beta_{ni}\})$,
i.e. the r.h.s. of \eqref{iterseed} evaluated for function \eqref{defgf0}.
This proves that this function solves the recursive equation.

The generalization of this proof to the full ansatz \eqref{defDf-gen} is elementary.
Instead of the relation \eqref{relsums}, one now has
\be
\hf\Sym
\sum_{\ell=1}^{n-1}\frac{1}{\ell!(n-\ell)!}\sum_{k_L=0}^{[(\ell-1)/2]}
\!\!
\sum_{\cT_L\in\, \IT_{\ell-2k_L,k_L}^\ell}
\!\!\!\!\!
\sum_{k_R=0}^{[(n-\ell-1)/2]}
\!\!\!\!
\sum_{\cT_R\in\, \IT_{n-\ell-2k_R,k_R}^\ell}
\!\!
\sum_{i=1}^\ell\sum_{j=\ell+1}^n
\!=\frac{1}{n!}\sum_{k=0}^{[(n-1)/2]}
\!\!\!\!
\sum_{\cT\in\, \IT_{n-2k,k}^\ell}\sum_{e\in E_\cT},
\label{relsums-marks}
\ee
which again reflects the fact the sum over marked unrooted labelled trees can be represented as a sum over all possible splittings
into two such trees by cutting them along an edge. If the cutting edge joins a marked vertex, it counts $1+2m_\ver$ times,
producing the right factor $\gamma_{\ver_L\ver_R}$, which is reflected
by the fact that the sums over $i,j$ run over $\ell$ and $n-\ell$ values, respectively.
Other numerical factors work in the same way as before.
Thus, the l.h.s. of \eqref{iterseed} can again be rewritten as in \eqref{lhsrecursive} with the only difference that $\sum_{\cT\in\, \IT_n^\ell}$
should now be replaced by
$\sum_{k=0}^{[(n-1)/2]}\sum_{\cT\in\, \IT_{n-2k,k}^\ell}\prod_{\ver\in V_\cT}\tcV_\ver$.
Applying the same sign identity \eqref{signident3} for $m=n-1-2k$ and the same identification for $x_\alpha$,
one recovers the r.h.s. of \eqref{iterseed}.
This completes the proof of the proposition.

\subsection*{Proposition  \ref{prop-reswhg}}

This proposition trivially follows from \eqref{iterDn}.

\subsection*{Proposition  \ref{prop-resR} \label{ap-prop4}}

The easiest way to prove \eqref{solRn} is to substitute it into \eqref{defSm} and then check
that the result is consistent with the constraint \eqref{Emrec}.
The substitution generates a sum over trees which resemble the blooming trees
of appendix \ref{ap-smooth}: these are trees with vertices from which other trees grows.
But now the two types of trees, representing the `base' and the `flowers', are actually the same --- both of them are Schr\"oder trees.
The only difference is that vertices of the `base' carry weights $\cEp_v$,
whereas the vertices of `flowers' have weights $\cEf_v$.
Since the leaves of a `flower' are in one-to-one correspondence with the children of the vertex of the `base' tree from
which this flower grows, such blooming trees can be equivalently represented by the usual Schr\"oder trees
obtained by replacing the vertices of the `base' by their `flowers'.
In this way, we obtain
\be
W_n= \Sym\left\{\sum_{T\in\IT_n^{\rm S}}(-1)^{n_T} \sum_{\cup\, T'=T}
\prod_{T'}\[\cEp_{v_0(T')}\prod_{v\in V_{T'}\setminus \{v_0(T')\}}\cEf_{v}\]\right\},
\label{solW-sub}
\ee
where the second sum goes over decompositions of $T$ into subtrees such that the root $v_0(T')$ of a subtree $T'$
is a leaf of another subtree (except, of course, the root of the total tree).

Let us now consider a vertex $v$ whose only children are leaves of $T$.
Then all decompositions into subtrees $T=\cup_i\, T'_i$ can be split into pairs such that two decompositions
differ only by whether $v$ (together with its leaves) represents a separate subtree or it is a part of a bigger subtree.
The contributions of two such decompositions into \eqref{solW-sub} differ only by the factors assigned to the vertex $v$
and therefore they combine into the factor $\cE_v$ assigned to this vertex.
As a result, $\cE_v$ appears as a common factor and the vertex $v$ can be excluded from the following consideration.
Proceeding in the same way with the tree obtained by removing this vertex, one finds that
the sum over decompositions can be evaluated explicitly and gives 
\be
W_n= \Sym\left\{\sum_{T\in\IT_n^{\rm S}}(-1)^{n_T} \cEp_{v_0}\prod_{v\in V_T\setminus{\{v_0\}}}\cE_{v}\right\}.
\label{solWn}
\ee

Given this result, the proof of the constraint
\eqref{Emrec} is analogous to the proof of the relation \eqref{hRh-full}:
the contribution of each tree $T$ (from the sum in \eqref{solWn}) and a splitting with $n_k>1$ (from the sum in \eqref{Emrec})
is cancelled by the contribution of another tree obtained from $T$ by adding $n_k$ children to its $k$th leaf and the same splitting
but with $n_k$ replaced by $1+\cdots +1$ (repeated $n_k$ times).
The only contribution which survives is the one generated by the tree with a single
vertex and $n$ leaves
and the splitting with all $n_k=1$. It is given by $\cEp_n$, which verifies the constraint and proves the proposition.

\subsection*{Proposition \ref{prop-derbtau}}

Our starting point to prove the proposition is the formula
\be
\p_{\bar\tau}\whh_{p,\mu}(\tau)=
\frac{\I}{2}\sum_{n=2}^\infty
\sum_{\sum_{i=1}^n \gama_i=\gama}
\p_{\tau_2}R_n(\{\gama_i\},\tau_2)
\, e^{\pi\I \tau Q_n(\{\gama_i\})}
\prod_{i=1}^n h_{p_i,\mu_i}(\tau).
\label{derwhh}
\ee
Substituting $\p_{\tau_2}R_n$ following from \eqref{solRn} and the inverse formula \eqref{exp-hwh} expressing $h_{p,\mu}$
in terms of the completion, with the functions $W_n$ found in \eqref{solWn},
one arrives at the result \eqref{exp-derwh} where the functions $\cJ_n$ are given by
\be
\cJ_n= \frac{\I}{2}\Sym\left\{\sum_{T\in\IT_n^{\rm S}}(-1)^{n_T-1} \p_{\tau_2}\cE_{v_0}
\sum_{T'\subseteq T}\prod_{v\in V_{T'}\setminus{\{v_0\}}}\cEf_{v}
\prod_{v\in L_{T'}}\cEp_{v}\prod_{v\in V_T\setminus (V_{T'}\cup L_{T'})}\cE_{v}
\right\},
\label{inter-cJ}
\ee
where the second sum goes over all subtrees $T'$ of $T$ containing its root and $L_{T'}$ is the set of their leaves.
Here the subtree $T'$ corresponds to the tree in the formula \eqref{solRn} for $R_n$, whereas the subtrees starting from its leaves
correspond to the trees in the expression \eqref{solWn} for $W_n$.
The sum over subtrees can be evaluated in the same way as the sum over decompositions into subtrees in \eqref{solW-sub}:
the two contributions differing only by whether the vertex $v$ belongs to $V_{T'}$ or $L_{T'}$
combine into the factor $\cE_v$ assigned to this vertex. Performing this recombination for all vertices of the tree $T$,
one obtains the expression \eqref{solJn}, which proves the proposition.

\subsection*{Proposition  \ref{prop-treeindex}}

Our goal is to prove that in the expression
\be
\gtr
= \Sym\left\{\sum_{T\in\IT_n^{\rm S}}(-1)^{n_T-1} \(\gf_{v_0}-\cEf_{v_0}\)\prod_{v\in V_T\setminus{\{v_0\}}}\cEf_{v}\right\},
\label{soliterg-tr-ap}
\ee
following from \eqref{soliterg},
all contributions due to marked trees cancel leaving only contributions of trees without marks, i.e with $m=0$.
There are several possible situations which we need to analyze.

\lfig{Combination of two Schr\"oder trees ensuring the cancellation of contributions generated by marked trees.}
{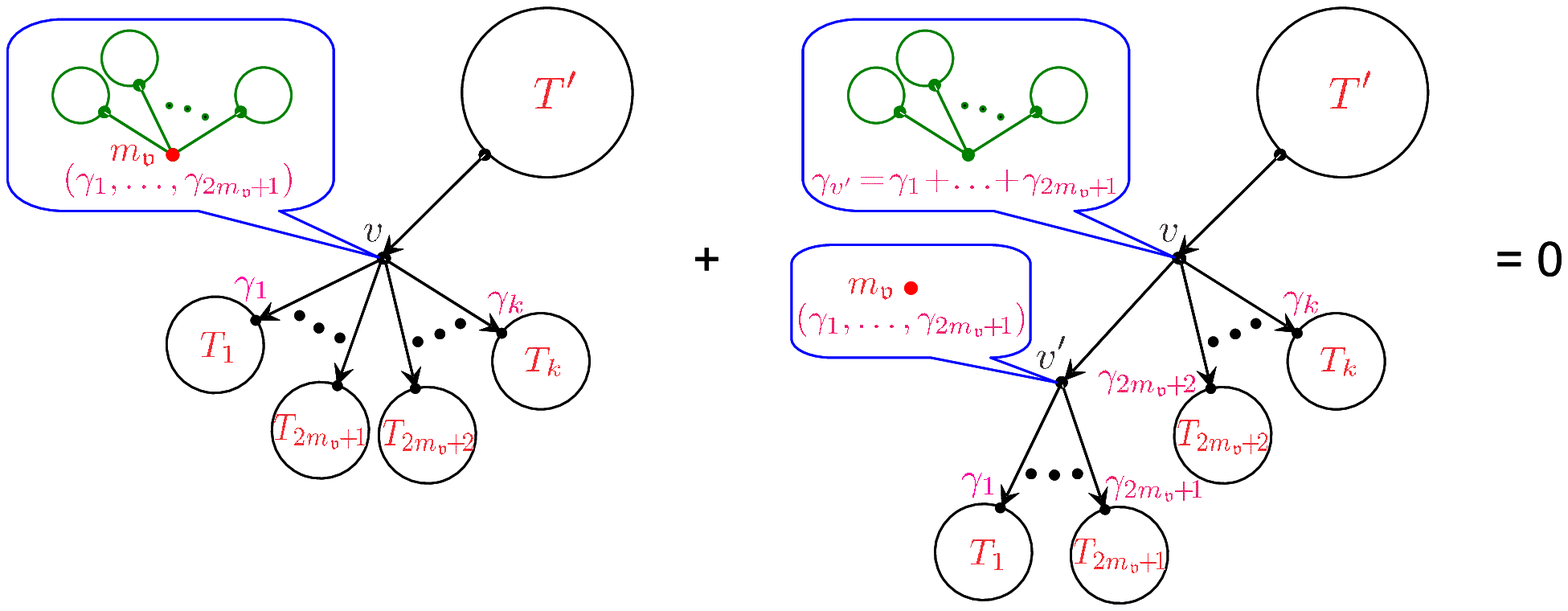}{17.cm}{fig-Mcancel}{-1.5cm}

First, let us consider the contributions generated by non-trivial marked trees, i.e. trees having more than one vertex and at least one mark.
Let us focus on the contribution corresponding to a vertex $\ver$ with $m_\ver>0$ marks of a tree $\cT$,
which appears in the sum over marked unrooted trees living at a vertex $v$ of a Schr\"oder tree $T$. Let $k=n_v$ be the number
of children of the vertex $v$ and $\gamma_i$ ($i=1,\dots,k$) their charges so that $\gamma_s$ ($s=1,\dots,2m_\ver+1$) are the charges
labelling the marked vertex $\ver$. Note that $k\ge 2m_\ver+2$ because the tree $\cT$ has at least one additional vertex except $\ver$.
Then the contribution we described is cancelled by the contribution coming from another Schr\"oder tree, which is obtained from $T$ by adding an edge
connecting the vertex $v$ to a new vertex $v'$, whose children are the $2m_\ver+1$ children of $v$ in $T$ carrying charges $\gamma_s$
(see Fig. \ref{fig-Mcancel}).\footnote{The new tree is of Schr\"oder type because its vertex $v$ has $k-2m_{\ver}\ge 2$ children
and vertex $v'$ has $2m_\ver+1\ge 3$ children.}
Indeed, choosing the same tree $\cT$ as before in the sum over marked trees at vertex $v$, but now with $m_\ver=0$, and
in the sum at vertex $v'$ the trivial tree having one vertex and $m_\ver$ marks, one gets exactly the same contribution as before,
but now with an opposite sign due to the presence of an additional vertex in the Schr\"oder tree.
Thus, all contributions from non-trivial marked trees are cancelled.

As a result, we remain only with the contributions generated by trivial marked trees, i.e. having only one vertex and $m_\ver$ marks.
One has to distinguish two cases: either the corresponding vertex $v$ of the Schr\"oder tree is the root or not.
In the former case, this contribution is trivially cancelled in the difference $\gf_{v_0}-\cEf_{v_0}$ in \eqref{soliterg-tr-ap}.
In the latter case, this is precisely the contribution used above to cancel the contributions
from non-trivial marked trees. This exhausts all possibilities and we arrive at the formula \eqref{soliterg-tr}.

\addtocontents{toc}{\vspace{-0.1cm}}
\addtocontents{toc}{\protect\nopagebreak}
\section{Explicit results up to 4th order}
\label{ap-explicit}

In this appendix we provide explicit expressions for various functions appearing in our construction up to the forth order.
To write them down, we will use the shorthand notation $\gamma_{i+j}=\gamma_i+\gamma_j$, $c_{i+j}=c_i+c_j$, etc. as well as
indicate the arguments of functions through their indices, for instance, $\cE_{i_1\cdots i_n}=\cE_n(\gama_{i_1},\dots,\gama_{i_n})$.
These expressions are obtained by simple substitutions using the results found in the main text and the sets of trees
shown in Fig. \ref{fig-table}. For $n=2$ they all agree with the results of \cite{Alexandrov:2016tnf}.

The results \eqref{soliterg} and \eqref{solRn} generate the following expansions
\bea
h^{\rm DT}_{p,q}&=&\whh_{p,\mu}+
\sum_{\gama_1+\gama_2=\gama}
\[\gf_{12}-\cE_{12}\]e^{\pi\I \tau Q_2(\{\gama_i\})}
\whh_{p_1,\mu_1}\whh_{p_2,\mu_2}
\nn\\
&& +\sum_{\sum_{i=1}^3 \gama_i=\gama}
\[\gf_{123}-\cE_{123}-2\(\gf_{1+2,3}-\cE_{1+2,3}\)\cE_{12}\] e^{\pi\I \tau Q_3(\{\gama_i\})}
\prod_{i=1}^3 \whh_{p_i,\mu_i}
\label{exphDT}\\
&&
+\sum_{\sum_{i=1}^4 \gama_i=\gama}
\[ \gf_{1234}-\cE_{1234}-2\(\gf_{1+2+3,4}-\cE_{1+2+3,4}\)\(\cE_{123} -2\cE_{1+2,3}\cE_{12}\)
\right.
\nn\\
&& \left.\qquad
-3\(\gf_{1+2,34}-\cE_{1+2,34}\)\cE_{12}
+\(\gf_{1+2,3+4}-\cE_{1+2,3+4}\)\cE_{12}\cE_{34}
\]e^{\pi\I \tau Q_4(\{\gama_i\})}
\prod_{i=1}^4 \whh_{p_i,\mu_i}+\cdots,
\nn
\eea
\bea
\whh_{p,q}&=&h_{p,\mu}+
\sum_{\gama_1+\gama_2=\gama}
\cEp_{12}\,e^{\pi\I \tau Q_2(\{\gama_i\})}
h_{p_1,\mu_1}h_{p_2,\mu_2}
\nn\\
&& +\sum_{\sum_{i=1}^3 \gama_i=\gama}
\Bigl[\cEp_{123}-2\cEp_{1+2,3}\cEf_{12}\Bigr] \,e^{\pi\I \tau Q_3(\{\gama_i\})}
\prod_{i=1}^3 h_{p_i,\mu_i}
\label{exphwhh}\\
&&
+\sum_{\sum_{i=1}^4 \gama_i=\gama}
\Bigl[\cEp_{1234}-2\cEp_{1+2+3,4}\(\cEf_{123} -2\cEf_{1+2,3}\cEf_{12}\)
-3\cEp_{1+2,34}\cEf_{12}
\Bigr.
\nn\\
&& \Bigl. \qquad
+\cEp_{1+2,3+4}\cEf_{12}\cEf_{34}
\Bigr]\, e^{\pi\I \tau Q_4(\{\gama_i\})}
\prod_{i=1}^4 h_{p_i,\mu_i}+\cdots,
\nn
\eea
where the functions $\gf_n$ and $\cE_n$ can be read off from \eqref{defDf-gen}, \eqref{fullker} and \eqref{rescEn},
\be
\begin{split}
\gf_2=&\, \frac{(-1)^{1+\gamma_{12}}}{4}\, \gamma_{12}\,\sgn(c_1),
\\
\gf_3=&\, \frac{(-1)^{1+\gamma_{12}+\gamma_{1+2,3}}}{8}\,\Sym\biggl\{ \gamma_{12}\,\gamma_{23}\,
\sgn(c_1)\, \sgn(c_3)+\frac13\,\gamma_{12}\,\gamma_{23} \biggr\},
\\
\gf_4=&\, \frac{(-1)^{\gamma_{12}+\gamma_{1+2,3}+\gamma_{1+2+3,4}}}{16}\,
\Sym\biggl\{\gamma_{12}\gamma_{23}\gamma_{34}\,\sgn(c_1)\sgn(c_{1+2})\sgn(c_4)
\\
&\, \qquad
-\frac13\,\gamma_{12}\gamma_{23}\gamma_{24}\,\sgn(c_1)\sgn(c_3)\sgn(c_4)
-\frac13\, \gamma_{12}\gamma_{23} \gamma_{1+2+3,4}\sgn(c_4)
\biggr\},
\end{split}
\label{resgf234}
\ee
\bea
\cE_{12} &=&\frac{(-1)^{\gamma_{12}}}{4\sqrt{2\tau_2}}\,
\tPhi_{1}^E\bigl(\bfv_{12},\bfv_{12}\bigr)
\nn\\
&=&\frac{(-1)^{\gamma_{12}}}{4}\[\gamma_{12}\, E_1\(\frac{\sqrt{2\tau_2}\gamma_{12} }{\sqrt{(pp_1p_2)}}\)
+\frac{\sqrt{(pp_1p_2)}}{\pi\sqrt{2\tau_2}}\, e^{-\frac{2\pi\tau_2\gamma_{12}^2 }{(pp_1p_2)}} \] ,
\nn\\
\cE_{123} &=&\frac{(-1)^{\gamma_{12}+\gamma_{1+2,3}}}{8}\,
\Sym\left\{\frac{1}{2\tau_2}\,\tPhi_{2}^E\bigl((\bfv_{1,2+3},\bfv_{1+2,3}),(\bfv_{12},\bfv_{23})\bigr)
\right.
\nn\\
&&\left. \qquad
-\frac13\(\gamma_{12}\gamma_{23}-\frac{(p_1p_2p_3)}{4\pi\tau_2} \)
\right\},
\label{rescE4}
\\
\cE_{1234} &=&
\frac{(-1)^{\gamma_{12}+\gamma_{1+2,3}+\gamma_{1+2+3,4}}}{16\sqrt{2\tau_2}}\,
\Sym\biggl\{\frac{1}{2\tau_2}\biggl(\tPhi_{3}^E\bigl((\bfv_{1,2+3+4},\bfv_{1+2,3+4},\bfv_{1+2+3,4}),(\bfv_{12},\bfv_{23},\bfv_{34})\bigr)
\nn\\
&&
+\frac13\,\tPhi_{3}^E\bigl((\bfv_{1,2+3+4},\bfv_{1+2+4,3},\bfv_{1+2+3,4}),(\bfv_{12},\bfv_{23},\bfv_{24})\bigr)\biggr)
\nn\\
&& -\frac13\(\gamma_{12}\gamma_{23}-\frac{(p_1p_2p_3)}{4\pi\tau_2}\)\tPhi_{1}^E\bigl(\bfv_{1+2+3,4},\bfv_{1+2+3,4}\bigr)
\biggr\},
\nn
\eea
where all generalized error functions are evaluated at $\bfx=\sqrt{2\tau_2}(\bfq+\bfb)$.
The last terms appearing in the above quantities for $n=3$ and $n=4$ correspond to contributions of trees with one mark ($m=1$).
In fact, at these orders these results can be rewritten in a simpler form, which coincides with the representations \eqref{gfn-upto5}
and \eqref{fullker-alt} (in the latter formula one should drop the sum over partitions and take $d_T=d_n$):
\bea
\gf_3&=& \frac{(-1)^{1+\gamma_{12}+\gamma_{1+2,3}}}{12}\,\Sym\Bigl\{ \gamma_{12}\,\gamma_{1+2,3}\,
\sgn(c_1)\, \sgn(c_3) \Bigr\},
\label{resgf34simple}\\
\gf_4&=&\textstyle{\frac{(-1)^{\gamma_{12}+\gamma_{1+2,3}+\gamma_{1+2+3,4}}}{96}}
\Sym\!\Bigl\{\! \Bigl(2\,\gamma_{23}\gamma_{1,2+3}\gamma_{1+2+3,4}+\gamma_{12}\gamma_{34}\gamma_{1+2,3+4}\Bigr)\sgn(c_1)\sgn(c_{1+2})\sgn(c_4)
\Bigr\},
\nn
\eea
\bea
\cE_{123} &=&\frac{(-1)^{\gamma_{12}+\gamma_{1+2,3}}}{24\tau_2}\,
\Sym\left\{\tPhi_{2}^E\bigl((\bfv_{1,2+3},\bfv_{1+2,3}),(\bfv_{12},\bfv_{1+2,3})\bigr)
\right\},
\label{rescE4alt}\\
\cE_{1234} &=&
\frac{(-1)^{\gamma_{12}+\gamma_{1+2,3}+\gamma_{1+2+3,4}}}{96(2\tau_2)^{3/2}}\,
\Sym\biggl\{2\,\tPhi_{3}^E\bigl((\bfv_{1,2+3+4},\bfv_{1+2,3+4},\bfv_{1+2+3,4}),(\bfv_{23},\bfv_{1,2+3},\bfv_{1+2+3,4})\bigr)
\nn\\
&&
+\tPhi_{3}^E\bigl((\bfv_{1,2+3+4},\bfv_{1+2,3+4},\bfv_{1+2+3,4}),(\bfv_{12},\bfv_{34},\bfv_{1+2,3+4})\bigr)\biggr)
\biggr\}.
\nn
\eea
The simplest way to prove the equality of the two representations of $\gf_n$ is to expand the DSZ products appearing in \eqref{resgf34simple}
into elementary $\gamma_{ij}$'s and then, using symmetrization, bring all their products to the form appearing in \eqref{resgf234}.
These products are multiplied by combinations of sign functions which can be recombined with the
help of the identity \eqref{signprop-ap}.
As an example, let us perform these manipulations  for $n=3$:
\be
\begin{split}
&\,
\Sym\Bigl\{ \gamma_{12}\,\gamma_{1+2,3}\,\sgn(c_1)\, \sgn(c_3) \Bigr\}
=\Sym\Bigl\{ \gamma_{12}\,\gamma_{23}\bigl(\sgn(c_1)-\sgn(c_2) \bigr)\,\sgn(c_3)\Bigr\}
\\
=&\,
\Sym\Bigl\{ \gamma_{12}\,\gamma_{23}\Bigl(\sgn(c_1)\,\sgn(c_3)+\hf\, \sgn(c_{1+3})\bigl(\sgn(c_1)+\sgn(c_3) \bigr)\Bigr)\Bigr\}
\\
=&\,
\hf\,\Sym\Bigl\{ \gamma_{12}\,\gamma_{23}\Bigl(3\, \sgn(c_1)\,\sgn(c_3)+1\Bigr)\Bigr\},
\end{split}
\label{relsigns3}
\ee
where we used that $c_2=-(c_1+c_3)$. This identity then shows the equality of the two forms of $\gf_3$ given above.
For $\gf_4$ the manipulations are very similar, but a bit more cumbersome.
The equality of the two forms of $\cE_n$ follows from the equality of their asymptotics $\cEf_n$,
which is in turn ensured by the same identities as for $\gf_n$.

Finally, we provide expressions for the kernels $\whPhi^{{\rm tot}}_n$
of the indefinite theta series appearing in the expansion \eqref{treeFh-flh}
of $\cG$ in powers of $\whh_{p,\mu}$. For the first two orders, one has
\be
\begin{split}
\whPhi^{{\rm tot}}_1=&\,\intPhi_1,
\\
\whPhi^{{\rm tot}}_2=&\,
\intPhi_2+\intPhi_1 \whgPhi_2
=\frac{1}{4}\, \intPhi_1\(\tPhi_1^E(\bfu_{12},\bfv_{12})-\tPhi_1^E(\bfv_{12},\bfv_{12})\),
\end{split}
\ee
where $\intPhi_1(x)$ is defined in \eqref{defPhi1-main}.
At the next order,
\be
\whPhi^{{\rm tot}}_3=\intPhi_3+2\Sym\Bigl\{\intPhi_2(x_{1+2},x_3)\whgPhi_2(x_1,x_2) \Bigr\}+\intPhi_1 \whgPhi_3.
\ee
To get an explicit expression in terms of smooth solutions of Vign\'eras' equation,
one should use the relation \eqref{expPhiE}. Applying it to the case $n=2$ with
$\cV=(\bfu_{1,2+3},\bfu_{1+2,3})$, using the orthogonality properties
\be
\bfu_{(1,2+3)\perp(1+2,3)}=\bfu_{12},
\qquad
\bfu_{(1+2,3)\perp(1,2+3)}=\bfu_{23},
\label{orth-vec}
\ee
and acting by the operator $\cD(\bfv_{12})\cD(\bfv_{23})$,
one can show that
\bea
&&\Sym\Bigl\{\tPhi^E_{2}\bigl((\bfu_{1,2+3},\bfu_{1+2,3}), (\bfv_{12},\bfv_{23})\bigr)\Bigr\}=
\Sym\biggl\{\tPhi^M_{2}\bigl((\bfu_{1,2+3},\bfu_{1+2,3}), (\bfv_{12},\bfv_{23})\bigr)
\nn\\
&& \qquad
+(\bfv_{12},\bfx)\,\sgn(\bfu_{12},\bfx)\, \tPhi^M_1\bigl(\bfu_{1+2,3},\bfv_{1+2,3}\bigr)
+(\bfv_{12},\bfx)\,(\bfv_{23},\bfx)\,\sgn(\bfu_{1,2+3},\bfx)\,\sgn(\bfu_{1+2,3},\bfx)
\nn\\
&& \qquad
-\frac{(p_1p_2p_3)}{6\pi}
\biggr\}.
\label{relPhiEM3}
\eea
This result allows to obtain the following representation for the kernel
\bea
\whPhi^{{\rm tot}}_3&=&\frac18\,\intPhi_1 \Sym \biggl\{\tPhi^E_{2}\bigl((\bfu_{1,2+3},\bfu_{1+2,3}), (\bfv_{12},\bfv_{23})\bigr)
- \tPhi^E_{2}\bigl((\bfv_{1,2+3},\bfv_{1+2,3}), (\bfv_{12},\bfv_{23})\bigr)
\nn\\
&&
-\(\tPhi_1^E(\bfu_{1+2,3},\bfv_{1+2,3})-\tPhi_1^E(\bfv_{1+2,3},\bfv_{1+2,3})\)\tPhi^E_1(\bfv_{12},\bfv_{12})
\biggr\}.
\eea
Using the identity between functions $\tPhi^E_{2}$ implied by the identity \eqref{relsigns3}, the kernel
can also be rewritten as
\bea
\whPhi^{{\rm tot}}_3&=&\frac18\,\intPhi_1 \Sym \biggl\{\frac23\, \Bigl(\tPhi^E_{2}\bigl((\bfu_{1,2+3},\bfu_{1+2,3}), (\bfv_{12},\bfv_{1+2,3})\bigr)
- \tPhi^E_{2}\bigl((\bfv_{1,2+3},\bfv_{1+2,3}), (\bfv_{12},\bfv_{1+2,3})\bigr)\Bigr)
\nn\\
&&
-\(\tPhi_1^E(\bfu_{1+2,3},\bfv_{1+2,3})-\tPhi_1^E(\bfv_{1+2,3},\bfv_{1+2,3})\)\tPhi^E_1(\bfv_{12},\bfv_{12})
\biggr\},
\label{reshPhi3alt}
\eea
so that the vectors appearing in the second argument of $\tPhi_2^E$ are now mutually orthogonal.

For $n=4$, the kernel is given by
\be
\begin{split}
\whPhi^{{\rm tot}}_4=&\,  \intPhi_4+\Sym\Bigl\{3\intPhi_3(x_{1+2},x_3,x_4)\whgPhi_2(x_1,x_2)+\intPhi_2(x_{1+2},x_{3+4})\whgPhi_2(x_1,x_2)\whgPhi_2(x_3,x_4)
\\
&\,
+2\intPhi_2(x_{1+2+3},x_4)\whgPhi_3(x_1,x_2,x_3)\Bigr\}+\intPhi_1 \whgPhi_4.
\end{split}
\ee
Proceeding in the same way as for $n=3$, obtaining a generalization of \eqref{relPhiEM3} to $n=4$,
one arrives at
\bea
\whPhi^{{\rm tot}}_4&=&\frac{1}{16}\, \intPhi_1
\Sym \biggl\{\tPhi^E_{3}\bigl((\bfu_{1,2+3+4},\bfu_{1+2,3+4},\bfu_{1+2+3,4}), (\bfv_{12},\bfv_{23},\bfv_{34})\bigr)
\nn\\
&&\qquad
-\tPhi^E_{3}\bigl((\bfv_{1,2+3+4},\bfv_{1+2,3+4},\bfv_{1+2+3,4}), (\bfv_{12},\bfv_{23},\bfv_{34})\bigr)
\nn\\
&&\quad
+\frac13\(\tPhi^E_{3}\bigl((\bfu_{1,2+3+4},\bfu_{1+2+4,3},\bfu_{1+2+3,4}), (\bfv_{12},\bfv_{24},\bfv_{34})\bigr)
\right.
\nn\\
&&\qquad
\left.
-\tPhi^E_{3}\bigl((\bfv_{1,2+3+4},\bfv_{1+2+4,3},\bfv_{1+2+3,4}), (\bfv_{12},\bfv_{24},\bfv_{34})\bigr)\)
\nn\\
&&
-\(\tPhi^E_{2}\bigl((\bfu_{1+2,3+4},\bfu_{1+2+3,4}), (\bfv_{1+2,3},\bfv_{34})\bigr)
-\tPhi^E_{2}\bigl((\bfv_{1+2,3+4},\bfv_{1+2+3,4}), (\bfv_{1+2,3},\bfv_{34})\bigr)
\right.
\nn\\
&&\quad
+\hf\(\tPhi^E_{2}\bigl((\bfu_{1+2+4,3},\bfu_{1+2+3,4}), (\bfv_{1+2,3},\bfv_{1+2,4})\bigr)
\right.
\\
&&\quad\qquad
\left.\left.
-\tPhi^E_{2}\bigl((\bfv_{1+2+4,3},\bfv_{1+2+3,4}), (\bfv_{1+2,3},\bfv_{1+2,4})\bigr)
\)\)\tPhi^E_1(\bfv_{12},\bfv_{12})
\nn\\
&&
-\(\tPhi_1^E(\bfu_{1+2+3,4},\bfv_{1+2+3,4})-\tPhi_1^E(\bfv_{1+2+3,4},\bfv_{1+2+3,4})\)
\tPhi^E_{2}\bigl((\bfv_{1,2+3},\bfv_{1+2,3}), (\bfv_{12},\bfv_{23})\bigr)
\nn\\
&&
+\(\tPhi_1^E(\bfu_{1+2+3,4},\bfv_{1+2+3,4})-\tPhi_1^E(\bfv_{1+2+3,4},\bfv_{1+2+3,4})\)
\tPhi^E_1(\bfv_{12},\bfv_{12})\tPhi^E_1(\bfv_{1+2,3},\bfv_{1+2,3})
\nn\\
&&
+\frac14\(\tPhi_1^E(\bfu_{1+2,3+4},\bfv_{1+2,3+4})-\tPhi_1^E(\bfv_{1+2,3+4},\bfv_{1+2,3+4})\)
\tPhi^E_1(\bfv_{12},\bfv_{12})\tPhi^E_1(\bfv_{34},\bfv_{34})
\biggr\}.
\nn
\eea
It is possible also to rewrite this expression in terms of generalized error functions $\tPhi^E_n(\cV,\tcV)$
where the vectors entering the second argument are mutually orthogonal, as in \eqref{reshPhi3alt}.
The reader can easily guess the result by comparing \eqref{rescE4} and \eqref{rescE4alt}.

The explicit results for $\whPhi^{{\rm tot}}_n$, $n\leq 4$, presented above are the basis
for the conjectural formula \eqref{kertotsol}.
Note that all terms in these expressions have the sum of ranks of the generalized error functions equal to $n-1$.
This shows that all contributions due to trees with non-zero number of marks cancel in the sum over Schr\"oder trees.

\newpage

\addtocontents{toc}{\vspace{-0.1cm}}
\section{Index of notations}
\label{sec_index}

\begin{longtable}{lp{10cm}l}
Symbol & Description & Appears/defined in \\
\hline
\rule{0pt}{20pt}
$\cA(\cT) $ & contribution of tree $\cT$ to the integrand of the multi-instanton expansion of $H_\gamma$ and $\cG$ & \eqref{amplit} \\
$a_\cT$ & coefficient of unrooted labelled tree $\cT$ in $\cD_{m}(\{\gama_s\})$ & \eqref{defcDcT}, \eqref{res-aT} \\
$\beta_{k\ell}$ &  DSZ product $\langle \gamma_1+\dots+\gamma_k,\gamma_\ell\rangle$ & \eqref{notaion-c} \\
$b_2=b_2(\CY)$ & second Betti number of $\CY$ &p.\pageref{sec-twistinst} \\
$b^a=\Re(z^a)$ & periods of the Kalb-Ramond field & p.\pageref{subsec-MSW} \\
$b_n$ & rational coefficients in the expansion of $\Fref_n$ & \eqref{defFref}, \eqref{coefbn} \\
$c_i$ & stability parameters & \eqref{fiparam} \\
$c_i^{(\ell)}$ & stability parameters after attractor flow & \eqref{flowc}\\
$c_{2,a}$ & components of the second Chern class of $\CY$ & \eqref{fractionalshiftsD5}\\
$d=n b_2$ & dimension of the lattice $\Lat=\oplus_{i=1}^n \Lambda_i$ & p.\pageref{biform}\\
$d_n,d_T$ & rational weights in the representation of $\gf_n$ via flow tree & \eqref{gfn-upto5},\eqref{PPn} \\
$\Delta(T)$, $\Delta_{\gamma_L\gamma_R}^z$ &sign factors assigned to attractor flow tree $T$ & \eqref{kappaT} \\
$\cD_{\wh}$ & Maass raising operator & \eqref{modcovD}\\
$\cD(\bfv)$ & modular-covariant derivative contracted with vector $\bfv$ &\eqref{defcDif}\\
$\cD_{m}(\{\gama_s\}) $ & derivative operator assigning weight to vertices with $m$ marks & \eqref{defcDcT} \\
$E_n(\cM;\vu)$ & generalized error function on $\IR^n$ &\eqref{generr-E} \\
$\cE_n=\cEf_n+\cEp_n$ & function encoding the modular completion & \eqref{rescEn}, \eqref{defEn0new} \\
$\phi$ & (logarithm of) contact potential on $\cM_H$ & \eqref{phiinstmany} \\
$\Phi_n^E$, $\Phi_n^M$  & boosted (complementary) error functions & \eqref{generrPhiME}\\
$\tPhi_{n,m}^E$, $\tPhi_{n,m}^M$ & uplifted boosted error functions in the kernel of $\cV_ m$
(denoted by $\tPhi^E_n$, $\tPhi^M_n$ when $n=m$) & \eqref{deftPhigen} \\
$\Phi_\cT$ & contribution of tree $\cT$ to $\intPhi_n$ & \eqref{PhiPhiT}, \eqref{resPhiTend} \\
$\intPhi_n$ & kernel defined by twistorial integrals & \eqref{totker}, \eqref{Phin-final} \\
$\cEPhi_n$ & kernel corresponding to function $\cE_n$ & \eqref{fullker} \\
$\tcEPhi_n$ & kernel promoting $\gf_n$ to a solution of Vign\'eras' equation & \eqref{fullker-mod} \\
$\gPhi_{n}$ & kernel corresponding to the tree index & \eqref{totker}, \eqref{kerPhiint} \\
$\whgPhi_{n}$ & kernel corresponding to completed tree index $\whg_n$ &\eqref{kerPhiinth} \\
$\Phi^{{\rm tot}}_n$ & total kernel in the expansion of $\cG$ in terms of $h_{p,\mu}$  & \eqref{totker} \\
$\whPhi^{{\rm tot}}_n$ & total kernel in the expansion of $\cG$ in terms of $\whh_{p,\mu}$  & \eqref{kertotsol} \\
$F(X)$ & holomorphic prepotential & p.\pageref{Fcl} \\
$\Ftr{n} (\{\gamma_{i}\},z^a)$ & partial tree index & \eqref{defFpl} \\
$\Fref_{n} (\{c_i\})$ &  partial contribution in $\gref_n$  & \eqref{whgF} \\
$\cF$ & image of $\cG$ under Euler operator & \eqref{Fexpand} \\
$\Gamma$ & charge lattice inside $H_{\rm even}(\CY,\IQ)$ & \eqref{fractionalshiftsD5}\\
$\Gamma_+$ & positive cone in the charge lattice & \eqref{khcone} \\
$\Gamma_{k\ell}, \Gamma_e$ & sums of DSZ products &  \eqref{notaion-c}, \eqref{defGame} \\
$\gamma= (0,p^a,q_a,q_0)$
&charge vector of a generic D4 (or D3) brane & \eqref{fractionalshiftsD5}\\
$\gama=(p^a,q_a)$ & projection of $\gamma$ on $H_4(\CY,\mathbb{Q})\oplus H_2(\CY,\mathbb{Q})$
& \eqref{multih}\\
$\gamma_{ij}=\langle \gamma_i,\gamma_j\rangle$
& Dirac-Schwinger-Zwanziger product, or Euler pairing & \eqref{gamZZ}\\
$\gtr(\{\gamma_i\},z^a)$ & tree index,  also denoted by $\gtr(\{\gamma_i,c_i\})$  & \eqref{defgtree}\\
$\whg_n(\{\gama_i\},z^a,\tau_2)$ & completed tree index, also denoted by $\whg_n(\{\gama_i,c_i\})$
& \eqref{multihd-full}, \eqref{soliterg}\\
$\gf_n(\{\gama_i,c_i\})$ & seed term in recursion for $\whg_n$ & \eqref{iterDn}\\
$\gref_n(\{\gama_i,c_i\},y)$ & refined version of $\gf_n$ & \eqref{whgF} \\
$\cG$ &instanton generating function & \eqref{defcF2} \\
$\cG_n(\{\gamma_i,z_i\})$ & integrand in the $n$-instanton contribution to $\cG$ & \eqref{expl-tcA}\\
$G_n(\{\gama_i,c_i\};\tau_2)$ & large $\bfx$ limit of $\tPhi_n(\bfx)$ & \eqref{defgf-hPhi}, \eqref{asymp-compl}\\
$h^{\rm DT}_{p,q}(\tau,z^a)$ & generating function of DT invariants & \eqref{defchimu} \\
$h_{p,\mu}(\tau) $ & generating function of MSW invariants & \eqref{defhDT} \\
$\whh_{p,\mu}(\tau)$ & modular completion of $h_{p,\mu}(\tau)$ & \eqref{exp-whh} \\
$H_\gamma(z)$ & generator of contact transformation across $\ell_\gamma$ & \eqref{prepHnew}
\\
$\Hcl_\gamma(z)$ & classical, large volume limit of $H_\gamma(z)$  & \eqref{expcX}\\
$\kappa(T)$ & weight of attractor flow tree $T$ & \eqref{kappaT0}\\
$\kappa(x)$ & BPS index of two-centered solutions & \eqref{kappadef}\\
$\kappa_{abc}$ & intersection numbers  on a fixed basis of $H_4(\CY)$ & p.\pageref{fractionalshiftsD5}\\
$K_{\gamma_1\gamma_2}(z_1,z_2)$ &  integration kernel in large volume limit & \eqref{defkerK}\\
$\hK_{ij}(z_i,z_j)$ & rescaled integration kernel in large volume limit & \eqref{defkerKr}\\
$\cJ_n(\{\gama_i\},\tau_2)$ &  coefficient in the formula for the shadow of  $\whh_{p,\mu}$ & \eqref{exp-derwh} \\
$\Lambda=H_4(\CY,\IZ)$ & lattice equipped with quadratic form $\kappa_{ab}=\kappa_{abc} p^a$ & p.\pageref{khcone}\\
$\Lambda_i=H_4(\CY,\IZ)$ & lattice equipped with quadratic form $\kappa_{i,ab}=\kappa_{abc} p_i^a$ & p.\pageref{biform}\\
$\Lat$ & lattice $\Lat=\oplus_{i=1}^n \Lambda_i$ spanned by $n$ D1-brane charges & p.\pageref{biform}\\
$\lambda$ & eigenvalue under Vign\'eras' operator & \eqref{Vigdif} \\
$\ell_\gamma$ & BPS ray on the twistor fiber, or its large volume limit & \eqref{defBPSay}, p.\pageref{defkerK} \\
$\mu_a$ & residue class of $q_a$ modulo spectral flow & \eqref{defmu} \\
$M_n(\cM;\vu)$ & generalized complementary error function on $\IR^n$&\eqref{generr-M} \\
$\cM_{\alpha\beta}$ & matrix of parameters in the generalized error functions & \eqref{generr-M}, \eqref{generr-E} \\
$\cM_{\cK}(\CY)$ & complexified K\"ahler moduli space of $\CY$ & p.\pageref{fractionalshiftsD5}\\
$\cM_H$ & hypermultiplet moduli space in IIB/$\CY$, or vector multiplet moduli space in
IIA/$(\CY\times S^1)=$ M/$(\CY\times T^2)$ & p.\pageref{sec-twist}\\
$n_T$ & number of vertices of rooted tree $T$ excluding the leaves & \eqref{soliterg} \\
$n_v(T)$ & number of descendants of the vertex $v$ in $T$ plus one & \eqref{combident} \\
$n_\ver$ & valency of vertex $\ver$ of an unrooted tree & \eqref{res-aT} \\
$p^a$ & homology class of the divisor wrapped by the D3-brane  & \eqref{khcone} \\
$\cP_{m}(\{p_s\})$ & weight of a vertex with $m$ marks in $G_n$ & \eqref{defPver-tree}\\
$\hat q_0$ & invariant D0-brane charge & \eqref{defqhat} \\
$Q_n(\{\gama_i\})$ & difference of quadratic forms for constituents and the total charge & \eqref{defQlr} \\
$R_n(\{\gama_i\},\tau_2)$ & non-holomorphic correction in the completion $\whh_{p,\mu}(\tau)$  & \eqref{multihd-full}, \eqref{solRn} \\
$\sigma_\gamma$ & quadratic refinement & \eqref{defqf}, \eqref{qf} \\
$\cs_k, \cs_e$ & sums of stability parameters & \eqref{notaion-c}, \eqref{defce}\\
$S^{\rm cl}_p$ & classical action of a D3-instanton in large volume limit & \eqref{Xtheta}\\
$\vartheta_{\bfp,\bfmu}\bigl(\Phi,\lambda)$ & indefinite theta series with kernel $\Phi$ & \eqref{Vignerasth} \\
$\tau=\tau_1+\I\tau_2$ & 4D axio-dilaton in IIA, or torus modulus in M theory & p.\pageref{sec-twistinst}\\
$t$ & complex coordinate on the twistor fiber & p.\pageref{defXsf} \\
$t^a=\Im(z^a)$ & K\"ahler moduli on $\CY$ & p.\pageref{subsec-MSW} \\
$T$ & rooted tree with charges assigned to the leaves & footnote \ref{foot-trees} \\
$\cT$ &  tree with charges assigned to vertices & footnote \ref{foot-trees} \\
$\IT_n, \IT_n^\ell$ & set of unrooted (labelled) trees with $n$ vertices
& \eqref{expl-tcA}\\
$\IT_{n,m}, \IT_{n,m}^\ell$ & set of unrooted (labelled) trees with $n$ vertices and $m$ marks &
\eqref{asymp-compl}\\
$\IT_n^{\rm r}$ & set of rooted trees with $n$ vertices & \eqref{Hexpand} \\
$\IT_n^{\rm af}$&set of attractor flow trees with $n$ leaves  & \eqref{defgtree}\\
$\IT_n^{\rm S}$ & set of Schr\"oder trees with $n$ leaves & \eqref{defSm} \\
$\IT_{2m+1}^{(3)}$ & set of rooted ternary trees with $n$ leaves & \eqref{defPver-tree}\\
$u^\Lambda=(1,u^a)$ &  complex structure moduli of mirror threefold $\CYm$  & p.\pageref{defXsf} \\
$\bfu_{ij}$, $\bfu_e$, $\bfu_\ell$&  vectors in $\IR^d$ associated to $-2\Im[Z_{\gamma_i}\bZ_{\gamma_j}]$, $-\cs_e$ and $-\cs_\ell$
& \eqref{defvij}, \eqref{defue}, \eqref{deftvl} \\
$\bfv_{ij}$, $\bfv_e$, $\bfv_\ell$
& vectors in $\IR^d$ associated to $\langle\gamma_i,\gamma_j\rangle$, $\Gamma_e$ and $-\Gamma_{n\ell}$
& \eqref{defvij}, \eqref{defue}, \eqref{deftvl} \\
$V_T$ & set of vertices of rooted tree $T$ excluding leaves & p.\pageref{kappaT} \\
$V_\lambda$ & Vign\'eras' operator &  \eqref{Vigdif} \\
$\cV_m$ & weight of a vertex with $m$ marks in large $\bfx$ limit of $\tcEPhi_n$ & \eqref{defcV-mw} \\
$\tcV_m$ & weight of a vertex with $m$ marks in $\gf_n$ & \eqref{deftcV-mw} \\
$W_n(\{\gama_i\},\tau_2)$ & coefficient in the formula for $h_{p,\mu}$ in terms of $\whh_{p_i,\mu_i}$ & \eqref{exp-hwh} \\
$\cX_\gamma$ & holomorphic Fourier modes on the twistor space of $\cM_H$ & \eqref{eqTBA}\\
$\cXsf_\gamma$ & semi-flat limit of $\cX_\gamma$ &\eqref{defXsf} \\
$\cXcl_\gamma$ & classical limit of $\cX_\gamma$ & \eqref{clactinst}\\
$\cXt_{p,q}$ & $\hat q_0$-independent part of $\cXcl_\gamma$ & \eqref{Xtheta} \\
$\bfx$ & $d$-dimensional vector, argument of kernels of theta series & p.\pageref{Vignerasth} \\
$z$ & coordinate on the twistor fiber, after Cayley transf. & \eqref{Cayley} \\
$z_\gamma$ & saddle point on twistor fiber & \eqref{saddle} \\
$z^a=b^a+\I t^a$ & complexified K\"ahler moduli  of $\CY$ & p.\pageref{sec-twistinst}\\
$z^a_*(\gamma)$ & attractor moduli for charge $\gamma$ & p.\pageref{defntilde}\\
$z^a_\infty(\gamma)$  & large volume attractor point for charge $\gamma$ & \eqref{lvap} \\
$Z_\gamma(z^a)$ & central charge & \eqref{defZg}\\
$\Omega(\gamma,z^a)$ & generalized Donaldson-Thomas invariant & p.\pageref{flow}\\
$\bar\Omega(\gamma,z^a)$ & rational DT invariant & \eqref{defntilde} \\
$\bOm_*(\gamma)$ & attractor index & p.\pageref{defntilde}\\
$\bOmMSW(\gamma)$ & MSW invariant, also denoted by $\bOm_{p,\mu}( \hat q_0)$
 & p.\pageref{defmu}
\end{longtable}

\newpage

\addtocontents{toc}{\vspace{-0.5cm}}

\providecommand{\href}[2]{#2}\begingroup\raggedright\endgroup


\end{document}